\PassOptionsToPackage{table}{xcolor}
\documentclass[%
 reprint,
superscriptaddress,
nofootinbib,
 amsmath,amssymb,
 aps,
pra,
]{revtex4-2}

\usepackage{tikz}
\usepackage{amsmath,amsfonts,amsthm, mathtools}
\usepackage{float}
\usepackage{dsfont}
\usepackage{csquotes}
\usepackage{amssymb}
\usepackage{braket}
\usepackage{dsfont}
\usepackage[table]{xcolor}
\usepackage{verbatim}
\usepackage{mathtools}
\usepackage{rotating}
\usepackage{thmtools, thm-restate}
\usepackage{bm}
\usepackage{hyperref}
\usepackage{mathbbol}
\usepackage{makecell}
\usepackage{pifont}
\usepackage{appendix}
\usepackage{diagbox}
\usepackage{soul}

\definecolor{darkred}  {rgb}{0.5,0,0}
\definecolor{darkblue} {rgb}{0,0,0.5}
\definecolor{darkgreen}{rgb}{0,0.5,0}
\hypersetup{
  colorlinks = true,
  urlcolor  = blue,
  linkcolor = darkblue,
  citecolor = darkgreen,
  filecolor = darkred
}
\theoremstyle{definition}

\newtheorem{corollary}{Corollary}
\newtheorem{definition}{Definition}
\newtheorem{conjecture}{Conjecture}
\newtheorem{lemma}{Lemma}
\newtheorem*{lemma*}{Lemma}
\newtheorem{proposition}{Proposition}

\newtheorem{theorem}{Theorem}
\newtheorem*{theorem*}{Theorem}

\newtheorem*{remark}{Remark}

 \newtheorem{applemma}{Lemma}
\newtheorem{appcorollary}{Corollary}

\newcommand{\conv}{\text{conv}}
\newcommand{\mbf}{\mathbf}
\newcommand{\mbb}{\mathbb}
\newcommand{\mc}{\mathcal}
\newcommand{\msf}{\mathsf}
\newcommand{\mf}{\mathfrak}

\newcommand{\tr}{\textrm{Tr}}

\newcommand{\supp}{\text{supp}}
\newcommand{\A}{\mathsf{A}}
\newcommand{\B}{\mathsf{B}}
\newcommand{\tD}{\text{D}}

\usepackage{multirow}
\newcommand{\op}[2]{|#1\rangle\langle#2|}

\definecolor{cool_green}{rgb}{0.0, 0.5, 0.0}
\newcommand{\eric}[1]{{\color{cool_green} #1}}
\newcommand{\norm}[1]{\left\lVert#1\right\rVert}
\usepackage{blkarray}
\usepackage{adjustbox}
\usepackage{graphicx}
\usepackage{dcolumn}

\usepackage[normalem]{ulem}

\newcommand{\yujie}{\color{purple}}
\newcommand{\blk}{\color{black}}
\newcommand{\cg}{\color{cool_green}}
\begin{document}


\title{Destroying and Preserving Measurement Incompatibility over a Quantum Channel}



\author{Eric Chitambar}
\email{echitamb@illinois.edu}
\affiliation{Department of Electrical and Computer Engineering and Coordinated Science Laboratory, University of Illinois at Urbana-Champaign, Urbana, Illinois 61801, USA}

\author{Yujie Zhang}
\email{yujie4physics@gmail.com}
\affiliation{Institute for Quantum Computing and Department of Physics \& Astronomy,
University of Waterloo, 200 University Ave W, Waterloo, Ontario, N2L 3G1, Canada}
\affiliation{Perimeter Institute for Theoretical Physics, 31 Caroline Street North, Waterloo, Ontario, Canada N2L 2Y5}


\date{\today}

\begin{abstract}

A quantum channel is called incompatibility-breaking if it always destroys the incompatibility of quantum measurements.
On the other extreme are incompatibility-preserving channels, which never destroy incompatibility.
This work studies both types of channels through the mathematical framework of zonotopes and zonoids, an approach that provides a complete characterization of compatibility for dichotomic measurements and is particularly powerful in the qubit setting.
Using the language of zonoids, we present necessary and sufficient conditions for when a qubit unital channel is incompatibility-breaking, which are also sufficient for nonunital channels.
As one of our main results, we demonstrate the superactivation phenomenon of two incompatibility-breaking qubit channels that are no longer incompatibility-breaking when used in parallel.
Translating this result into the task of quantum steering, it says that two copies of a two-qubit Werner state can be steered at noise threshold $1/2$, even though this is impossible using just one copy. For the problem of incompatibility preservation, we prove that a quantum channel with matching input and output dimensions will preserve general incompatibility if and only if it is a unitary transformation.
When the output dimension is larger, this result no longer holds, and we offer a conjectured form of a general incompatibility-preserving channel.

\end{abstract}

\maketitle


\section{Introduction}

A quantum channel describes a physical process that transmits and transforms quantum information. 
Noise in such a process can degrade the underlying quantumness of the channel output and, beyond some threshold, destroy it completely. 
A familiar example is an entanglement-breaking channel, which
destroys all entanglement between the transmitted system and any external system~\cite{Horodecki-2003a,Chruciski-2006a,Chen-2020a}. 
Analogous resource-destroying channels have been studied for other forms of quantum nonclassicality~\cite{Liu-2017a,Gour-2017a}, including coherence ~\cite{Bu-2016a,Kamin-2020a}, nonclassicality ~\cite{Ivan-2013a,Meznaric-2013a}, Bell nonlocality ~\cite{Pal-2015a,Zhang-2020a,Mukherjee-2025a}, and measurement incompatibility~\cite{Heinosaari-2015a,Kumari-2023a}.

At the opposite extreme are channels that preserve a quantum resource. 
For example, local unitary (LU) evolution preserves the entanglement of every bipartite state, while phase shifts and permutations leave the amount of (speakable) coherence in a system unchanged.  
More generally, one may ask whether a channel preserves a resource universally, or only on a restricted class of states, measurements, or tasks. 
Questions of this type also arise naturally in quantum information theory, where full reversibility of a channel is not necessary (or even possible) for completing a desired information-theoretic task, such as communication over a noisy channel.
Resource-preserving channels are particularly interesting because preservation does not \emph{a priori} require that the underlying process be reversible. 
It is therefore natural to ask whether genuinely noisy channels can nevertheless preserve a given resource, and, if so, how strongly  this preservation constrains the channel.

This work studies the quantum phenomena of measurement incompatibility.
A family of quantum measurements is called compatible if they can all be simulated by a single ``parent’’ measurement followed by classical post-processing.
Historically, compatibility has also been referred to as joint measurability \cite{Heinosaari-2016a, Guhne-2023a}, and a pair of non-commuting observables represents a textbook example of two measurements that are \textit{not} jointly measurable.
More generally, quantum measurements are described by positive operator-valued measures (POVMs), which extend beyond standard observables.
A family of finite-outcome POVMs on a $d$-dimensional system is described by a collection of positive operators $\{M_{a|x}\}$ on $\mathbb{C}^d$, often referred to as a measurement assemblage.
Here, $x$ denotes the measurement choice and $a$ denotes the outcome for measurement $x$, so that $\sum_a M_{a|x}=\mathbb{I}_d$ for all $x$.
The assemblage is compatible if and only if the measurement operators $M_{a|x}$ can be expressed as
\begin{align}
\label{Eq:compatibility-defn}
M_{a|x}=\int_\Lambda p(a|x,\lambda)d\Pi(\lambda) \qquad\forall a,x,
\end{align}
where $\Pi:\mc{F}\to\text{Pos}(\mathbb{C}^d)$ is some parent POVM defined on a measurable space $(\Lambda,\mc{F})$ and $p(a|x,\lambda)$ is a classical post-processing map such that $p(a|x,\lambda)\ge 0$ and $\sum_ap(a|x,\lambda)=1$.

Compatible measurement families are comparatively easier to implement experimentally, as they require only a single measurement apparatus $\Pi$ that acts independent of the measurement choice.
In contrast, an incompatible assemblage cannot be written in the form of Eq. \eqref{Eq:compatibility-defn}, and therefore its physical implementation requires additional quantum resources such as multiple measurement apparatuses, extra control hardware, or quantum memories.
This engineering overhead has motivated the view of measurement incompatibility as a quantum resource, with compatible measurement families regarded as ``free’’ and incompatible ones as resourceful \cite{Heinosaari-2016a, Skrzypczyk-2019a, Designolle-2019a, Buscemi-2020a, Ji-2024a, Mitra-2022a, Bluhm-2025a}.

More fundamentally, measurement incompatibility is a crucial resource in the realization of Bell nonlocality.
The use of incompatible measurements by both parties is necessary to obtain nonlocal measurement correlations in standard bipartite Bell experiments \cite{Terhal-2003a, Wolf-2009a} and sufficient in the multiqubit setting \cite{Plavala-2025a}, although not sufficient in general \cite{Quintino-2016a, Bene-2018a, Hirsch-2018a}. Moreover, measurement incompatibility is sufficient but not necessary to demonstrate the failure of generalized noncontextuality~\cite{Spekkens2005,TavakoliUola2020,Selby2023,
ZhangSchmidYingSpekkens2026,ZhangYingSchmid2025}. Even stronger, measurement incompatibility is the essential ingredient for demonstrating quantum steering~\cite{Quintino2014, Uola2014}.
A bipartite state $\rho_{\A\B}$ is called Einstein-Rosen-Podolsky (EPR) steerable from $\A$ to $\B$ if there exists a measurement assemblage $\{M_{a|x}\}$ on system $\A$ such that the post-measurement state assemblage on system  $\B$ \textit{does not} satisfy a local hidden state (LHS) model \cite{Wiseman-2007a}.
The latter means that there is a density operator-valued measure $\sigma$ on $(\Lambda,\mc{F})$ and stochastic map $p(a|x,\lambda)$ such that 
\begin{align}
\label{Eq:LHS }
\text{Tr}_{\A}[(M_{a|x}\otimes\mathbb{I})\rho_{\A\B}]=\int_{\Lambda}p(a|x,\lambda)d\sigma(\lambda).
\end{align}
Every incompatible measurement assemblage can be used to demonstrate steering for a suitable bipartite state.

But just as entanglement and coherence can be destroyed when a given system is subject to noise, so can measurement incompatibility.
General noise on a quantum system is described by a quantum channel, which from the Schr\"{o}dinger picture is a completely-positive trace-preserving (CPTP) map $\mc{N}$ acting on the density matrices of the system.
However, in the Heisenberg picture, the channel is a completely-positive unital map $\mc{N}^\dagger$ acting on the measurement operators, which transforms one measurement assemblage into another, $\{M_{a|x}\}\mapsto\{\mc{N}^\dagger(M_{a|x})\}$.

In this work, we investigate two fundamental questions that characterize the extreme effects a channel can have on measurement incompatibility: (i) under what conditions does a channel $\mc{N}$ \textit{destroy} the incompatibility in all measurement assemblages $\{M_{a|x}\}$ or in a specific type; and (ii) under what conditions does $\mc{N}$ \textit{preserve} the incompatibility in all measurement assemblages $\{M_{a|x}\}$ or in a specific type.

A particularly important case for both of these questions is when the measurement assemblage consists of projection valued-measures (PVMs), meaning that the $\{M_{a|x}\}$ are projectors that satisfy $M_{a|x}M_{a'|x}=\delta_{aa'}M_{a|x}$ for all $x$.
These are more commonly referred to as projective measurements, and from an axiomatic perspective, they represent the building block for all physical measurements \cite{Peres-1990a}.
This work focuses primarily on PVM assemblages for the question of incompatibility breaking, whereas its study of incompatibility preservation will consider the most general POVM assemblages.
The primary technical tool we use is the mathematical theory of zonoids, which enables us to develop many new structural results on measurement incompatibility.
Our main results are summarized as follows:



\begin{enumerate}
\item We develop a geometric approach to dichotomic measurement compatibility based on zonoids. 
Beyond providing a criterion for compatibility in terms of zonoid inclusion, we present a relatively simple proof that compatibility is a property that holds in the convergent limit of compatible assemblages, a result using much more elaborate argumentation to prove for non-dichotomic measurements \cite{Kim-2026a}.

\item
For qubits, our approach reduces the compatibility problem to the geometry of centrally-symmetric zonoids in $\mbb{R}^3$. 
As a first consequence, we obtain an analytic proof of the conjectured compatibility threshold for arbitrary finite families of noisy planar qubit PVMs~\cite{Andrejic-2020a,
Porto-2026a}.

\item We apply the zonoid geometry to incompatibility-breaking qubit channels. For a general qubit channel, we derive a geometric criterion guaranteeing that a prescribed family of spin measurements becomes compatible after passing through the channel. For unital channels, this condition is also necessary. Applied to the full family of qubit PVMs, it yields an exact PVM-incompatibility-breaking criterion for unital qubit channels and a sufficient criterion for general nonunital channels. 

\item At the unital incompatibility-breaking boundary, we further establish a rigidity result that effectively fixes the parent POVM. 
This implies, for instance, that the LHS model for two-qubit Werner states $\omega_r$ at the compatibility threshold of $r=1/2$ is essentially unique. 
To our knowledge, this is the first uniqueness result for any local hidden state or hidden variable model.  
Although while finalizing this manuscript, a somewhat similar uniqueness result was announced independently in~\cite{zhang2026spectral}.

\item We show that incompatibility breaking is not stable under tensor products. In particular, although the qubit depolarizing channel $\Delta_{1/2}$ is incompatibility breaking, its tensor product
$\Delta_{1/2}\otimes\Delta_{1/2}$ is not. Thus, two copies can activate incompatibility even though each individual channel is incompatibility breaking. Through the channel--state correspondence, the same construction gives a particularly simple two-copy activation of quantum steering for the corresponding two-qubit Werner state $\omega_{1/2}$.

\item Finally, we consider the opposite extreme of incompatibility preservation. Using the zonoid geometry, we first prove geometrically that a qubit channel is incompatibility preserving if and only if it is unitary. 
We then extend this classification to arbitrary finite dimensions. 
Correspondingly, for equal local dimensions, a bipartite state is shown to be incompatibility preserving from Alice to Bob if and only if it is pure and has full Schmidt rank.
\end{enumerate}

The remainder of the paper is organized as follows. 
After briefly introducing notation below, we begin in Section \ref{Sect:Prelim-Prior} with precise definitions of incompatibility-breaking/preserving channels and their connection to EPR steering.  
An overview of some previous work related to the results of this paper is also provided.
Sections \ref{Sect:Zonotopes-Zonoids} and \ref{Sect:Zonoids-Incompatibility} introduce the concepts of zonoids and their connection to measurement incompatibility, drawing on the mathematical framework of Aubrun and Lancien that relates POVMs to zonoids \cite{Aubrun-2015a}.
In Section \ref{Sect:qubit-compatibility}, we focus on qubit POVMs and characterize qubit compatibility in terms of zonoidal inclusion of Bloch vectors.
Section \ref{Sect:Planar} specializes even further to planar measurements, and we derive compatibility conditions as relatively simple geometric corollaries.
Our general framework is applied to the study of incompatibility-breaking channels in Section \ref{Sect:Incompatibility-breaking-qubit}, with examples for both unital and nonunital channels given in Section \ref{Sect:Examples}.
A two-copy activation result for incompatibility non-breaking is established in Section \ref{Sect:superactivation}. 
Section \ref{Sect:incompatibility-preserving-qubits} provides a geometric characterization of all incompatibility-preserving qubit channels, and this classification is then extended to arbitrary finite dimensions in Section \ref{Sect:incompatibility-preserving-arbitrary}.
For completion, in Section \ref{Sect:States-incompatibility-preserving} we map this problem back to the setting of steering and investigate which bipartite states can demonstrate EPR steering using any incompatible measurement assemblage.
Finally, a proof for Theorem \ref{Thm:zonoid-unique} and details for Theorem \ref{Thm:two-copy-boundary-activation} are provided in Appendices \ref{App:boundary-zonoid-uniqueness} and \ref{App:two-copy-boundary-parent}, respectively.
Some general facts about representing POVMs as zonoids is presented in Appendix \ref{App:zonoid-generator}.

\medskip

\noindent\textbf{Notation:}

Throughout this work we let $\mathsf{Herm}(\mc{H})$, $\mathsf{Pos}(\mc{H})$, and $\text{D}(\mc{H})$ denote the sets of hermitian, positive, and density operators on Hilbert space $\mc{H}$, respectively.
We write $[0,\mathbb{I}]$ to denote the family of operators in $\mathsf{Pos}(\mc{H})$ with eigenvalues in the interval $[0,1]$.
For some topological space $X$, we let $\mf{B}(X)$ denote the Borel sets of $X$.
The Hausdorff distance for compact sets $A$ and $B$ is defined as
\begin{align}
d_H(A,B)=\max\{\max_{\vec{a}\in A}\min_{\vec{b}\in B}\Vert\vec{a}-\vec{b}\Vert, \max_{\vec{b}\in B}\min_{\vec{a}\in A}\Vert\vec{a}-\vec{b}\Vert\}.\notag
\end{align}
Many of our results will rely heavily on the Bloch sphere representing qubit state space.
We let $B^3$ denote the unit ball in $\mbb{R}^3$ (i.e. the Bloch sphere) and $S^2$ its surface.
We will throughout this work use $\omega$ to denote the normalized uniform measure on $S^2$, so that $d\mu(\hat{n})=\frac{1}{4\pi}d\Omega=\frac{1}{4\pi}\sin\theta d\theta d\phi$.
In general, we adopt standard notation and use a hat, as in $\hat{s}$, to indicate unit vectors, whereas general vectors will be written as $\vec{m}$.

\section{Preliminaries and Prior Results}

\label{Sect:Prelim-Prior}

To better set the stage for this work, we begin with some definitions and a summary of prior results.
\begin{definition}
A channel $\mc N$ is called 
\begin{itemize}
    \item \textbf{$\mf{M}$-incompatibility breaking} if the output assemblage
$\{\mc{N}^\dagger(M_{a|x})\}$ is compatible for input assemblage $\mf{M}=\{M_{a|x}\}$;
\item \textbf{PVM-incompatibility breaking} if it is $\mf{M}$-incompatibility breaking for every assemblage $\mf{M}$ of projective measurements;
\item \textbf{Incompatibility breaking} if it is $\mf{M}$-incompatibility breaking for every measurement assemblage $\mf{M}$.
\end{itemize}
\end{definition}
\noindent An analogous set of definitions hold for the complementary notion of incompatibility breaking.
\begin{definition}
A channel $\mc N$ is called 
\begin{itemize}
    \item \textbf{$\mf{M}$-incompatibility preserving} if it is not $\mf{M}$-incompatibility breaking whenever $\mf{M}$ is incompatible;
\item \textbf{PVM-incompatibility preserving} if it is $\mf{M}$-incompatibility preserving for every assemblage $\mf{M}$ of incompatible projective measurements;
\item\textbf{Incompatibility preserving} if it is $\mf{M}$-incompatibility preserving for every incompatible measurement assemblage $\mf{M}$.
\end{itemize}
\end{definition}
The notion of incompatibility-breaking channels was first introduced by Heinosaari \textit{et al.} \cite{Heinosaari-2015a} and was previously discussed in ~\cite{Ku2022,Ku2023, Heinosaari-2015a}.
A special form of incompatibility-breaking involves mixing a given assemblage with white noise or some other type of measurement.
The minimum amount of mixing needed to break the incompatibility of the initial assemblage is called its robustness \cite{Heinosaari-2015b, Bavaresco-2017a, Designolle-2019a, Skrzypczyk-2019a, Buscemi-2020a}, a quantity that has analogous meaning in other resource-theoretic settings \cite{Piani-2016a, Oszmaniec-2019a, Chitambar-2019a, Takagi-2019a, Kuroiwa-2024a}
 
The breaking and preserving of incompatibility can be equivalently described in terms of quantum steering \cite{Uola2015}.
To see this connection, let
\begin{align}
J^{(\mc{N})}_{\A\B}:=\mc{N}\otimes \text{id}(\Phi^+_{\B'\B})
\end{align}
denote the Choi state of an arbitrary channel $\mc{N}_{\mathsf{B}'\to\mathsf{A}}$, where $\Phi^+_{\B'\B}=\frac{1}{|\B|}\sum_{i,j=1}^{|\B|}\op{ii}{jj}_{\B\B'}$ and $\B\cong\B'$.
Conversely, up to a local contraction, every bipartite state $\rho_{\A\B}$ is the Choi state of some channel $\mc{N}_{\B'\to \A}$; namely, $\mc{N}$ is the CPTP map whose action on an arbitrary $\tau\in\mathrm{D}(\mc{H}_{\B})$ is given by $\mc{N}(\tau)=\tr_{\B} [(\mathbb{I}\otimes\tau^{\mathtt{T}})\hat{\rho}_{\A\B}]$, where $\hat{\rho}_{\A\B}:=1/|\mathsf B|(\mathbb{I}\otimes \rho_{\B}^{-1/2})\rho_{\A\B}(\mathbb{I}\otimes \rho_{\B}^{-1/2})$. 
The local contraction of $\rho_{AB}$ by $(\mathbb{I}\otimes \rho_{\B}^{-1/2})$ is needed to make $\mc{N}$ trace-preserving, but it does not affect steerability in the sense that a measurement assemblage $\{M_{a|x}\}$ on Alice's side can be used to steer $\rho_{\A\B}$ iff it can also be used to steer $\hat{\rho}_{\A\B}$ \cite{Quintino-2015a}.

Now suppose that Alice and Bob share a Choi state $J_{\A\B}^{(\mc{N})}$ and Alice performs some local measurement described by assemblage $\{M_{a|x}\}$.
Then the post-measurement state assemblage on Bob's side is given by
\begin{align}
    \sigma_{a|x}&=\tr_{\A}[(M_{a|x}\otimes\mbb{I})J^{(\mc{N})}]\notag\\
    &=\tr_{\A}[(\mc{N}^\dagger(M_{a|x})\otimes\mbb{I})\Phi^+_{\B'\B}]=\frac{1}{|\B|}\mc{N}^\dagger(M_{a|x})^\mathtt{T},
\end{align}
where $\mathtt{T}$ denotes the matrix transpose.
From this we see that the state assemblage $\{\sigma_{a|x}\}$ satisfies an LHS model iff the measurement assemblage $\{\mc{N}^\dagger(M_{a|x})\}$ is compatible.
Therefore, the state $J^{(\mc{N})}$ cannot be steered from Alice to Bob using measurements $\{M_{a|x}\}$ iff the assemblage $\{\mc{N}^\dagger(M_{a|x})\}$ is compatible \cite{Uola2014, Quintino2014}.
The question of whether some channel breaks or preserves a given measurement assemblage has thus been translated into the question of whether that assemblage steers the corresponding Choi state.

One of the most well-studied channels is the partially depolarizing channel, $M_{a|x}\mapsto \mc{D}_r^\dagger(M_{a|x}):=r M_{a|x}+ (1-r)\tr[M_{a|X}]\mathbb{I}/d$, and its Choi state is the $d\otimes d$ isotropic bipartite state $\omega_r:=J^{(\mc{D}_r)}=r\Phi^+_{\A\B}+(1-r)\mathbb{I}\otimes\mathbb{I}/d^2$.
Following the seminal work of Werner \cite{Werner-1989a}, an LHS model was originally proposed by Wiseman \textit{et al.} for $\omega_r$ \cite{Wiseman-2007a}. 
Translating this into the context of measurement incompatibility, it says that $\mc{D}_r$ is PVM-incompatibility breaking whenever $r\leq r_{\text{th}}=(\sum_{n=1}^d1/n-1)/(d-1)$. 
In particular, when $d=2$, this gives an PVM-incompatibility-breaking threshold of $1/2$, a threshold that was recently shown to also hold for arbitrary POVMs \cite{Renner-2024a, Zhang-2024a}.
For $d>2$, the incompatibility-breaking threshold for general POVMs remains an open problem, with the best lower bound given in Ref. \cite{Barrett-2002a}.

It should be stressed that when discussing the full sets of PVMs or POVMs, we are dealing with continuous families of measurements.
This leads to very subtle questions regarding the closure of compatible assemblages.
For example, suppose that $\{\mc{N}_\lambda\}$ is some family of channels that vary smoothly as we change the parameter $\lambda$.
When applying this to a fixed measurement assemblage $\{M_{a|x}\}$, we obtain the parameterized assemblage $\{\mc{N}_\lambda^\dagger(M_{a|x})\}$.
Up to some threshold value of $\lambda$, say, $\lambda_{\text{th}}$, the assemblages might be compatible, whereas they are incompatible when $\lambda>\lambda_{\text{th}}$.
The interesting question is what happens at the threshold $\lambda=\lambda_{\text{th}}$; is the assemblage compatible or incompatible?
Equivalently, we can rephrase this question as asking whether a convergent sequence of bipartite states $\{\rho_\lambda\}$ that is unsteerable by some assemblage $\{M_{a|x}\}$ has a limit point $\rho_{\lambda_{\text{th}}}$ that is also unsteerable by $\{M_{a|x}\}$.
When the assemblage $\{M_{a|x}\}$ consists of a finite number of measurements, a standard compactness argument implies that this is always true.
However, for infinite families of measurements, the question is highly non-trivial.

This problem is physically motivated by considering natural families of channels like the partial depolarizing channel $\mc{D}_r$.
In this case, we know that the assemblage $\{\mc{D}_{r_{\text{th}}}(M_{a|x})\}$ is compatible because we have a concrete LHS model at the threshold $r=r_{\text{th}}$ \cite{Wiseman-2007a}.
In fact, for $d=2$ this model requires a parent POVM with an unbounded number of measurement outcomes \cite{Zhang-2025a}, and as we show below (Theorem \ref{Thm:zonoid-unique}), this parent is effectively unique.
But without having an explicit LHS model at the threshold value, can we always conclude that the resulting assemblage is compatible based on the general topological structure of measurement compatibility? 
We prove in this paper (Theorem \ref{Thm:dichotomic-closure}) that indeed the threshold assemblage will always be compatible for families of dichotomic (i.e. two-outcomes) measurements.
For more general families of measurements, the same result holds, but a different method is needed for its proof than what is pursued here \cite{Kim-2026a}.

Beyond the partial depolarizing channel other results from the study of quantum steering can be imported to identify incompatibility-breaking channels.
References \cite{Cavalcanti-2017a, Uola-2020a, Guhne-2023a} provide excellent reviews of such results.
Here, we just highlight a few more specifically for the case $d=2$.
Through the Choi-state mapping, the family of unital qubit channels corresponds to the family of Bell-diagonal two-qubit states and their LU equivalents.
The PVM-steerability of Bell-diagonal states has been fully characterized \cite{Jevtic-2015a, Nguyen-2016a, Nguyen-2019a}, and recently one of us has shown that the same steering criteria also applies for general POVMs \cite{zhang2026parallel}; thus the incompatibility-breaking conditions for unital qubit channels is completely solved (see Section \ref{Sect:Incompatibility-Breaking-Unital} for these conditions). 
Much less is known about nonunital qubit channels, which correspond to two-qubit states with a non-uniform reduced matrix on Alice's side.
This non-uniformity is a necessary condition for realizing of one-way steering, and thus known examples of states demonstrating this phenomenon offer examples of incompatibility-breaking channels \cite{Bowles-2014a, Quintino-2015a, Bowles-2016a}.
A canonical example of a nonunital qubit channel is the amplitude damping channel $\mc{A}_\lambda$, where $\lambda$ denotes the damping parameter.
Using the pure-state steering method introduced in Ref. \cite{Song-2026a}, it can be shown that $\mc{A}_\lambda$ is not PVM-incompatibility breaking for all values $\lambda<1$.
However, as we will show in this paper (Section \ref{Sect:Examples}), as soon as any nonzero amount of depolarizing noise is mixed with $\mc{A}_\lambda$, there will always be some threshold damping value $\lambda_{\text{th}}<1$ such that the resulting channel is PVM-incompatibility breaking.

In parallel to these results on two-qubit steering, there is a long line of research that has directly studied the incompatibility of qubit measurements.
Early work by Busch provided necessary and sufficient conditions for when two unbiased dichotomic POVMs are compatible \cite{Busch-1986a}, a criterion that has been extended to biased pairs of POVMs \cite{Stano-2008a, Busch-2009a, Yu-2010a}, as well as more general measurements with special structure \cite{Liang-2011a, Yu-2013a, Kunjwal-2014a, Uola-2016a, Jae-2019a, Andrejic-2020a, Song-2026a}.
In particular, measurement assemblages restricted to a fixed plane in the Bloch sphere can be analyzed more readily.
It is known that the partially depolarizing channel is $\mf{M}$-incompatibility breaking for all planar measurements whenever $r\leq 2/\pi$ \cite{Uola-2016a, Zhang-2024a}.
We will recover this result for PVMs below (Section \ref{Sect:Planar}).

Various computational methods have been developed for deciding whether a given assemblage is compatible.
The compatibility decision problem, and more generally the incompatibility robustness of a given assemblage, can be expressed as a semi-definite program (SDP) \cite{Wolf-2009a, Cavalcanti-2016a}.
Unfortunately, this SDP scales poorly as the number of measurement choices and size of parent POVM increases \cite{Skrzypczyk-2020a}.
A more efficient method has been developed that involves optimizing the classical processing over some fixed choice of parent POVM \cite{Porto-2026a}, which works particularly well for Bloch sphere approximations in the case of qubits.
Another recent approach, based on the Fourier transform of the classical post-processing map, provides sufficient conditions for the compatibility of general qubit POVMs that can be formulated as a second-order cone program \cite{Grinko-2025a}.

The zonoid methods developed in this work build on a line of investigation initiated in our earlier paper \cite{Zhang-2025a}.
 Previously, we showed how the zonotope characterization of qubit compatibility can be used to obtain upper and lower bounds on the size of any parent POVM that can simulate all measurements on a sufficiently shrunken Bloch sphere.
 Here, we study more general types of Bloch sphere deformations and use the zonoid method to identify conditions under which incompatibility is either preserved or broken.
 In this approach, the problem of compatibility for dichotomic measurements is cast as a zonoid inclusion problem, which can be tackled using its own form of optimization and computational methods that complement the ones listed above.


\section{Zonoids and Measurement Compatibility}
Our approach to understanding compatibility will primarily involve geometrical objects known as zonotopes and zonoids.
Before restricting attention to qubits, we review definitions and observe some general connections between zonoids and measurement compatibility.
Some references on zonoids include \cite{Bolker-1969a, Goodey-1993a}, as well as most textbooks on  convex polytopes.

\subsection{Zonotopes and Zonoids}

\label{Sect:Zonotopes-Zonoids}

A parallelogram in $\mathbb{R}^2$ forms the most basic type of zonotope.
Adding more sides and generalizing to higher dimensions leads to the following definition.

\begin{definition}
\label{Defn:Zonotope}
A \textbf{zonotope} $\mc Z$ in $\mathbb{R}^d$ is a convex body of the form
 \begin{align}
 \label{Eq:zonotope-generated}
  \mc Z=\left\{\sum_{i=1}^kt_i\vec{v}_i\; \bigg| \; t_i\in[0,1]\right\},
\end{align}  
where the $\{\vec{v}_i\}_{i=1}^k$ are called generators of $\mc{Z}$.
\end{definition}

\noindent Here, we adopt the convention that every zonotopoe contains the origin, a condition that could be relaxed by allowing for translations.
We can equivalently think of a zonotope as a $k$-dimensional unit cube under the linear map $\vec{e}_i\mapsto \vec{v}_i$ for each unit vector $\vec{e}_i\in\mathbb{R}^k$.
The center of $\mc{Z}$ is the point $\vec{c}/2$ with $\vec{c}=\sum_{i=1}^k\vec{v}_i$.
For $d=2$ (resp. $d=3$), a zonotope is also referred to as a zonogon (resp. zonohedron).



Zonotopes are generated by finite sets of vectors.
We next consider the topological closure of such objects.
\begin{definition}
\label{Defn:Zonoid}
Any convex body $\mc Z$ that can be approximated arbitrarily well by a sequence of zonotopes is called a \textbf{zonoid}.
In other words, there exists a sequence of zonotopes $\mc Z_1,\mc Z_2,\cdots$ such that $d_H(\mc Z_n,\mc Z)\to 0$.
\end{definition}
\noindent In fact, in this definition it suffices to consider convergent sequences in which each constituent zonotope is an outer approximation of the convex body.
For example, the unit ball in $\mathbb{R}^d$ is a zonoid as well as any centrally-symmetric ellipsoid, since the linear image of any zonoid is also a zonoid.

A standard theorem in the theory of zonoids says that every zonoid $\mc{Z}\subset\mathbb{R}^d$ is the image of some vector-valued non-atomic measure $\pi:\mc{F}\to\mathbb{R}^d$ on a measurable space $(\Lambda,\mc{F})$; i.e. $\mc{Z}=\pi(\mc{F})$ \cite{Bolker-1969a}.
The converse of this statement is also true, and if $\pi$ is an atomic measure, one can just consider its convex hull to obtain a zonoid, $\mc{Z}=\conv(\pi(\mc{F}))$. 
A useful consequence of this is that we can always express $\mc{Z}$ in the form
\begin{align}
\label{Eq:zonoid-generator-representation}
    \mc{Z}=\left\{\int_{S^{d-1}}q(\hat{r})\hat{r}\;d\nu(\hat{r})\;|\; q\in[0,1] \;\text{is measurable}\right\},
\end{align}
where $\nu$ is a positive Borel measure on $S^{d-1}$ known as a generating measure of $\mc{Z}$ (see Appendix \ref{App:zonoid-generator} for details).
By comparing Eqns. \eqref{Eq:zonotope-generated} and \eqref{Eq:zonoid-generator-representation}, observe that we recover a zonotope with generators $\{\vec{v_i}\}_{i=1}^k$ by choosing the measure $\nu=\sum_{i=1}^k\Vert\vec{v}_i\Vert\delta_{\hat{v}_i}$, where $\delta_{\hat{v}_i}$ is the Dirac measure for point $\hat{v}_i$. 
The center of the zonoid is $\vec{c}/2$, where $\vec{c}=\int_{S^{d-1}}\hat{r}d\nu(\hat{r})$. 

The support function $h_{\mc{Z}}$ of a convex set $\mc{Z}$ is given by $h_{\mc Z}(\vec{n})=\max_{\vec{z}\in \mc Z}\langle\vec{n},\vec{z}\rangle$.
For a zonotope $\mc{Z}$ with generators $\{\vec{v}_i\}_{i=1}^k$, we can use Eq. \eqref{Eq:zonotope-generated} to readily obtain
\begin{align}
\label{Eq:support-function-zonotope1}
    h_{\mc{Z}}(\vec{n})=\frac{1}{2}\langle\vec{n},\vec{c}\rangle+\frac{1}{2}\sum_{i=1}^k|\langle\vec{n},\vec{v}_i\rangle|.
\end{align}
Similarly for a zonoid $\mc{Z}$ with generating measure $\nu$, we have
\begin{align}
\label{Eq:support-function-zonoid1}
    h_{\mc{Z}}(\vec{n})=\frac{1}{2}\langle\vec{n},\vec{c}\rangle+\frac{1}{2}\int_{S^{d-1}}|\langle\vec{n},\hat{r}\rangle| d\nu(\hat{r}).
\end{align}
By introducing the reflected measure $\nu_-(B):=\nu(-B)$ and the symmetrized even measure, $\nu_{\text{even}}:=\frac{1}{2}(\nu+\nu_-)$, we can use the fact that $|\langle\vec{n},\hat{r}\rangle|=|\langle\vec{n},-\hat{r}\rangle|$ to write $\int_{S^{d-1}}|\langle\vec{n},\hat{r}\rangle| d\nu(\hat{r})=\int_{S^{d-1}}|\langle\vec{n},\hat{r}\rangle|d\nu_{\text{even}}(\hat{r})$.
Therefore, the support function can be expressed as
\begin{align}
\label{Eq:support-representation}
    h_{\mc{Z}}(\vec{n})=\frac{1}{2}\langle\vec{n},\vec{c}\rangle+\frac{1}{2}\int_{S^{d-1}}|\langle\vec{n},\hat{r}\rangle|d\nu_{\text{even}}(\hat{r}).
\end{align}
In fact, it can be shown that $\nu_{\text{even}}$ is the \text{unique} even measure for which Eq. \eqref{Eq:support-representation} holds \cite{Bolker-1969a, Aubrun-2015a}.



For a zonotope $\mc{Z}$ with generators $\{\vec{v}_i\}_{i=1}^k$, its generator length is defined as the sum
\begin{align}
\ell(\mc Z):=\sum_{i=1}^k\Vert\vec{v}_i\Vert.
\end{align}
The generator length admits a particularly elegant expression using the support function of $\mc Z$.
While the following applies to zonotopes in arbitrary dimension, here we restrict attention to $d=3$ and $d=2$ since most of our main results pertain to qubit system.
For each generator $\vec{v}_i\in\mathbb{R}^3$, we compute the average magnitude of its inner product with a uniformly sampled unit vector,
\begin{align}
    \int_{S^{2}}|\langle \hat{n},\vec{v}_i\rangle|d\mu(\hat{n})=\Vert\vec{v}_i\Vert\int_{S^2}|\langle\hat{n},\hat{z}\rangle|d\mu(\hat{n})=\frac{1}{2} \Vert\vec{v}_i\Vert,
\end{align}
where $d\mu(\hat{n})=\frac{1}{4\pi}\sin\theta d\theta d\phi$.
Since $\int_{S^2}\langle\hat{n},\vec{c}\rangle d\mu(\hat{n})=0$, together with Eq.~\eqref{Eq:support-representation}, we have
\begin{align}
\label{Eq:generator-length-integral} \ell(\mc Z)=\sum_{i=1}^k\Vert\vec{v}_i\Vert=4\int_{S^2}h_{\mc Z}(\hat{n})d\mu(\hat{n}).
\end{align}
Note that since the RHS does not depend on the choice of generators $\{\vec{v}_i\}$ for $\mc Z$ or its center $\vec{c}$; it is a translation-invariant geometrical property of $\mc{Z}$.
For a zonoid $\mc Z\subset\mathbb{R}^3$, we similarly define $\ell(\mc Z):=4\int_{S^2}h_{\mc Z}(\hat{n})d\mu(\hat{n})$, and $\ell(\mc Z_n)\to \ell(\mc Z)$ for any approximating sequence of zonotopes $(\mc Z_n)$ such that $\mc Z_n\to \mc Z$.

A special property of zonotopes in $\mathbb{R}^2$ is that every centrally-symmetric polygon with $2n$ vertices forms a zonogon $\mc Z$ generated by $n$ vectors $\{\vec{v}_i\}_{i=1}^n$.
Furthermore, the perimeter of the zonogon is given by the simple formula
\begin{align}
\text{Per}(\mc Z)=2\cdot\ell(\mc Z)=2\sum_{i=1}^n\Vert\vec{v}_i\Vert.
\end{align}
Since every centrally-symmetric convex body in $\mathbb{R}^2$ can be approximated arbitrarily well by a polygon, the former is also a zonoid whose perimeter also satisfies the generator length formula $\text{Per}(\mc Z)=2\ell(\mc Z)$.
Once we move to $\mbb{R}^3$, there exist centrally-symmetric polyhedra that are not zonohedra, a regular octahedron being one example. 
Indeed, every face of a zonohedron is a zonogon, which is not the case for a tetrahedron.


\subsection{Zonoids and Incompatibility}

\label{Sect:Zonoids-Incompatibility}

Zonotopes and zonoids in $\text{Herm}(\mc{H})$ are defined analogously to Definitions \ref{Defn:Zonotope} and \ref{Defn:Zonoid}.
Alternatively, one can simply apply the isomorphism between $\mathsf{Herm}(\mc{H})$ and $\mathbb{R}^{d^2}$ for every $d$-dimensional Hilbert space $\mc{H}$.
Deep connections between zonoids, POVMs, and distinguishability norms have been drawn by Aubrun and Lancien in Ref. \cite{Aubrun-2015a}.
We review one of their preliminary results here.

\begin{lemma}[\cite{Aubrun-2015a}]
\label{Lem:Aubrun-Lancien}
    Let $\mc Z\subset[-\mathbb{I},\mathbb{I}]$ be a convex subset of $\text{Herm}(\mathbb{C}^d)$.
    Then the following are equivalent.
    \begin{enumerate}
        \item $\mc Z\subset[-\mathbb{I},\mathbb{I}]$ is a centrally-symmetric zonoid containing $\pm\mathbb{I}$.
        \item There exists a measurable space $(\Lambda,\mc{F})$ and POVM $\Pi:\mc{F}\to\text{Pos}(\mbb{C}^d)$ such that $\text{conv}(\Pi(\mc{F}))=\frac{1}{2}(\mathbb{I}+\mc{Z})$.
    \end{enumerate}
    Moreover, when $\mc Z$ is a zonotope, one can restrict to finite-outcome POVMs.
\end{lemma}

Every POVM $\Pi$ on measurable space $(\Lambda,\mc{F})$ can be associated with the zonoid $\mc{K}_{\Pi}:=\conv(\Pi(\mc{F}))$.
Since $M\in \conv(\Pi(\mc{F}))$ iff $\mathbb{I}-M\in \conv(\Pi(\mc{F}))$, the set $\mc{K}_{\Pi}$ is symmetric about the element $\mathbb{I}/2$.
That $\mc{Z}_{\Pi}$ is indeed a zonoid follows from Lemma \ref{Lem:Aubrun-Lancien}, or directly from the fact that the convex hull of the image of any vector-valued measure is a zonoid \cite{Bolker-1969a}.
Note that if $\Pi$ has a finite number of outcomes $\{\lambda_1,\cdots,\lambda_n\}$ with effects $\Pi_\lambda$, then $\mc{K}_{\Pi}$ is the zonotope given by
\begin{align}
\label{Eq:POVM-zonotope}
\mc{K}_{\Pi}=\left\{\sum_{i=1}^nt_{i}\Pi_{\lambda_i} | 0\leq t_i\leq 1\right\}.
\end{align}
This equality can be deduced by noting that $\text{conv}(\Pi(\mc{F}))$ is precisely the image of the vertices of the $n$-dimensional centrally-symmetric unit cube under the affine map $\vec{e}_i\mapsto\frac{1}{2}\Pi_{\lambda_i}+ \frac{1}{2}\mathbb{I}$, which is the RHS of Eq. \eqref{Eq:POVM-zonotope}.

The zonoid $\mc{K}_{\Pi}$ provides some insight into the measurement assemblages that can be simulated by $\Pi$ using classical post-processing. 
Within the study of quantum measurements, it is natural to introduce a preorder between POVMs based on whether one can simulate the other by post-processing \cite{Buscemi-2005a, Heinonen-2005a}.
This leads to the following notion of POVM equivalence.
\begin{definition}
      Two POVMs are called \textbf{operationally equivalent} if they can simulate the same collection of measurement assemblages.
\end{definition}
\noindent 
In other words, $\Pi$ and $\Pi'$ are operationally equivalent if any measurement assemblage $\{M_{a|x}\}$ that can be obtained from $\Pi$ using classical post-processing can also be obtained from $\Pi'$, and vice versa.
A connection to the zonoid $\mc{K}_{\Pi}$ is provided by the following lemma, whose proof is given in Appendix \ref{App:zonoid-generator}.
\begin{lemma}
\label{Lem:POVM-zonoid-generator}
    Let $\mc{K}_{\Pi}=\text{conv}(\Pi(\mc{F}))$ be the zonoid generated by POVM $\Pi:\mc{F}\to\text{Pos}(\mathbb{C}^d)$ on a measurable space $(\Lambda,\mc{F})$.
Then there exists a Borel probability measure $\omega$ on $\text{D}(\mathbb{C}^d)$, called a generating measure of $\mc{K}_{\Pi}$, such that
  \begin{align}
\label{Eq:POVM-zonoid-generator}
    \mc{K}_{\Pi}=\!\left\{d\!\!\int\! q(\rho)\rho\;d\omega(\rho)\;\bigg|\; q\in[0,1] \;\text{is measurable}\!\right\},
\end{align}
  Moreover, the POVM $X\mapsto d\int_X\rho d\omega(\rho)$ is operationally equivalent to $\Pi$.
\end{lemma}
\noindent 
The representation of $\mc{K}_{\Pi}$ given by Eq. \eqref{Eq:POVM-zonoid-generator} reduces to Eq. \eqref{Eq:POVM-zonotope} for a finite outcome POVM $\{\Pi_{\lambda_i}\}_{i=1}^n$ using the point measure $\omega=\frac{1}{d}\sum_{i=1}^n\tr[\Pi_{\lambda_i}]\delta_{\hat\Pi_{\lambda_i}}$, where $\hat{\Pi}_{\lambda_i}=\Pi_{\lambda_i}/\tr[\Pi_{\lambda_i}]$.
Note also that $\mathbb{I}=\Pi(\tD(\mathbb{C}^d))=d\int\rho d\omega(\rho)$. 

An immediate consequence of Lemma \ref{Lem:POVM-zonoid-generator} is the following.
\begin{corollary}
    $\{M_{a|x}\}\subset \mc{K}_{\Pi}$ whenever $\Pi$ can simulate assemblage $\{M_{a|x}\}$, and so $\mc{K}_{\Pi}=\mc{K}_{\Pi'}$ for any two operationally equivalent POVMs $\Pi$ and $\Pi'$.
\end{corollary}
\noindent The converse of this statement is not true in general.
This is because the simulation of a general measurement assemblage $\{M_{a|x}\}$ by $\Pi$ requires more than just a bounded response function $q(a|x)\in[0,1]$ (which alone is sufficient for membership in $\mc{K}_{\Pi}$);  it further requires that the response function be a normalized probability distribution, i.e. $\sum_{a}q(a|x)=1$.
However, when $\mc{K}_{\Pi}=\mc{K}_{\Pi'}$, the two POVMs are always parents to the same assemblages of dichotmic POVMs.
This is an implication of the following proposition, which provides a sufficient condition for the compatibility of dichotomic measurements.

\begin{proposition}
\label{Prop:dichotomic-compatibility}
    Suppose that $\{M_{x}\}\cup\{\mathbb{I}\}\subset\mc K$ for some zonoid $\mc{K}\subset[0,\mathbb{I}]$ that is symmetric about $\mathbb{I}/2$.  Then the assemblage of dichotomic measurements $\{M_x,\mathbb{I}-M_x\}_x$ is compatible.
\end{proposition}
\begin{proof}
    
Under the assumptions of the proposition, the zonoid $\mc{Z}=2\mc{K}-\mathbb{I}\subset[-\mathbb{I},\mathbb{I}]$ is centrally-symmetric and containing $\pm \mathbb{I}$.
     By Lemma~\ref{Lem:Aubrun-Lancien}, there exists a POVM $\Pi:\mc{F}\to\text{Pos}(\mathbb{C}^d)$ such that $\mc{K}=\conv(\Pi(\mc{F}))$.
     Hence, for each measurement choice $x$, there exist measurable sets $\{E_{i|x}\}$ and probabilities $\{q_{i|x}\}$ such that
     \begin{align}
     M_{x} = \sum_i q_{i|x}\Pi(E_{i|x}), \qquad \sum_i q_{i|x}=1.
     \end{align}
     Define $p(0|x,\lambda) := \sum_i q_{i|x}\chi_{E_{i|x}}(\lambda)$ and $p(1|x,\lambda)  := 1-p(0|x,\lambda)$.
     These are valid response functions that satisfy
     \begin{align}
     \int_{\Lambda}p(0|x,\lambda)d\Pi(\lambda) & = \sum_iq_{i|x}\Pi(E_{i|x}) = M_{x}, \notag \\
     \int_{\Lambda}p(1|x,\lambda)d\Pi(\lambda) & = \mathbb{I}-M_{x}.
     \end{align}
     Thus, $\Pi$ is able to simulate $\{M_{x},\mathbb{I}-M_x\}_x$.
\end{proof}

Proposition~\ref{Prop:dichotomic-compatibility} replaces the search over postprocessing functions in the definition of compatible assemblages by the single geometric condition that all child effects lie in a zonoid $\mc{K}$.
This is particularly useful for families containing infinitely many measurements, for which the standard finite-dimensional SDP formulation of joint measurability cannot be directly applied.
In practice, one could first search for a  zonotope $\mc{K}$ that contains a finite subset of effects, and then certify compatibility by checking whether the entire family of effects is also contained in it. 
Thus, an infinite collection of postprocessing constraints is replaced by a single convex-set inclusion problem.

This geometric inclusion can also be checked numerically in many cases of
interest. 
A finite-generator zonotope has only finitely many supporting facets, so containment reduces to checking the support of the target set along these facet directions. 
For example, we showed in Ref.~\cite{Zhang-2025a} that the smallest amount of depolarizing noise needed to simulate all qubit spin measurements by some finite-outcome parent POVM $\Pi$ can be computed by a finite optimization over the facets of a judiciously chosen zonotope.
This provides a practical numerical method for certifying compatibility of an entire continuous family of measurements. 
For the qubit examples considered in the next sections, the geometry is sufficiently simple that the corresponding inclusion problem can, in fact, be solved analytically.

In the method just discussed, one considers the inclusion of measurement effects in zonotopes, which have finite sets of generators.
However, in general there are compatible assemblages whose effects are contained only in the limiting structure of a zonoid.
While zonoids can be approximated arbitrarily well by zonotopes, no single zonotope might contain the entire assemblage; in other words, there is no finite-outcome parent POVM that can simulate the entire assemblage
This begs the question of whether measurement compatibility always holds in the limit of convergent measurement assemblages.
The following closure theorem ensures that this is indeed the case.


\begin{theorem}
\label{Thm:dichotomic-closure}
    Let $(\{M_{x}^{(n)},\mathbb{I}-M_{x}^{(n)}\}_x)_n$ be a sequence of compatible dichotomic measurement assemblages such that each $M_{x}^{(n)}$ converges pointwise to $M_{x}$ (with the set of measurement settings $\mc{X}$ being potentially continuous).
    Then $\{M_{x},\mathbb{I}-M_x\}_x$ is a compatible measurement assemblage.
\end{theorem}
\begin{proof}
    Let $\mc{K}_{\Pi}^{(n)}$ be the zonoid generated by the parent POVM for the compatible measurements $\{M_{x}^{(n)},\mathbb{I}-M_{x}^{(n)}\}_x$.
    Since each $\mc{K}_{\Pi}^{(n)}$ is a convex set subset of $[0,\mathbb{I}]$, the Blaschke Selection Theorem asserts the existence of a subsequence $(\mc{K}_{\Pi}^{(n_i)})_i$ that converges to some zonoid $\mc K$ with respect to the Hausdorff metric \cite{Kelly-1979a}.
    Hence, for each $x$, the subsequence $M_{x}^{(n_i)}$ converges to an element in $\mc{K}$, which must be $M_{x}$ since a convergent sequence has the same limit point as any of its subsequences.
    Thus, $\{M_{x}\}\subset\mc{K}$.
    Clearly we have $\mathbb{I}\in\mc{K}\subset[0,\mathbb{I}]$ with $\mc{K}$ being symmetric about $\mathbb{I}/2$, since these properties hold for each $\mc{K}_{\Pi}^{(n)}$.  
    Therefore, $\mc{K}=\text{conv}(\Pi(\mc{F}))$ for some POVM $\Pi$ by Proposition \ref{Prop:dichotomic-compatibility}, which means that the assemblage $\{M_{x},\mathbb{I}-M_x\}_x$ is compatible.
\end{proof}
\noindent It is worth noting that if $(\mc{K}^{(n)}\subset[0,\mathbb{I}])_n$ is any sequence of zonoids symmetric about $\mathbb{I}/2$ that converges to some zonoid $\mc{K}$ with $\{\mathbb{I}\}\cup\{M_{x}^{(n)}\}\subset\mc{K}^{(n)}$ , then pointwise convergence $(\{M_{x}^{(n)}\})_n\to \{M_{x}\}$ implies uniform convergence.
This is because $d_H(\mc{K}^{(n)},\mc{K})<\epsilon$ means that all elements of $\mc{K}^{(n)}\supset \{M_x^{(n)}\}$ are $\epsilon$ close to their limit points.


The problem of convergence among compatible measurements will be studied more deeply in forthcoming work.
Specifically, it can be shown that the closure of compatibility also holds for non-dichotomic measurements \cite{Kim-2026a}.

\medskip


\subsection{Qubit Compatibility}
\label{Sect:qubit-compatibility}

Elements in $\text{Herm}(\mathbb{C}^2)$ can be represented by vectors in $\mathbb{R}^4$ using the Pauli basis as
\begin{align}
M=\alpha\mathbb{I}+\frac{1}{2}\vec{m}\cdot\vec{\sigma} \leftrightarrow \vec{\tau}=(\alpha,\vec{m}).
\end{align}
The vector $\vec{m}=(m_x,m_y,m_z)\in\mathbb{R}^3$ is called the Bloch vector of $M$, and we include the factor of $1/2$ in the expansion of $M$ so that $m_i=\tr[M\sigma_i]$ for $i=x,y,z$.
A general qubit POVM $\{M_{a|x}\}$ can then be represented by a family of four vectors $\{\vec{\tau}_{a|x}=(\alpha_{a|x},\vec{m}_{a|x})\}$.

An important type of qubit measurement is one that is unbiased and dichotomic, often referred to as a \textit{noisy spin measurement}.
Unbiased refers to the fact that each outcome has equal probability when measuring the maximally mixed state $\mathbb{I}/2$.
Every noisy spin measurement is characterized by a Bloch vector $\vec{m}\in B^3$ and has two POVM elements 
\begin{align}
M_{\vec{m}}&:=\frac{1}{2}\left(\mathbb{I}+ \vec{m}\cdot\vec{\sigma}\right), &\mathbb{I}-M_{\vec{m}}&=\frac{1}{2}\left(\mbb{I}-\vec{m}\cdot\vec{\sigma}\right).
\end{align}
We will henceforth denote an assemblage of noisy spin measurements simply as the collection of effects $\{M_{\vec{m}}\}$, with each $M_{\vec{m}}$ representing the dichotomic measurement $\{M_{\vec{m}},\mathbb{I}-M_{\vec{m}}\}$.
Note that projective measurements are noiseless spin measurements and correspond to unit Bloch vectors, 
\begin{align}
\op{\hat{m}}{\hat{m}} & =M_{\hat{m}}=\frac{1}{2}(\mbb{I}+\hat{m}\cdot\vec{\sigma}),\notag \\
\op{-\hat{m}}{-\hat{m}} & =\mathbb{I}-M_{\hat{m}}=\frac{1}{2}(\mbb{I}-\hat{m}\cdot\vec{\sigma}).
\end{align}

The results of the previous section can be applied to describe qubit compatibility as zonoid inclusion of the four-vectors $\vec{\tau}$.
However, this is not always easy to use in practice.
In this section, we will translate some of the compatibility conditions into conditions on Bloch vectors and zonoids in $\mbb{R}^3$.
Thus, we will be considering two different zonoids: a zonoid in $\text{Pos}(\mathbb{C}^2)$, which we denote by $\mc{K}_{\Pi}$, and a zonid in $\mathbb{R}^3$, which we denote by $\mc{Z}$.
Our goal will be to understand the relationship between these two zonoids.
This will ultimately provide a more transparent geometric condition that can help identify incompatibility-breaking channels, as carried out in Section \ref{Sect:Incompatibility-breaking-qubit}. 

Lemma~\ref{Lem:POVM-zonoid-generator} and its translation to the Bloch sphere (see Corollary \ref{Cor:Zonoid-representation-Appendix-qubit}) says that every family of compatible qubit measurements has effects $\{M_{a|x}\}$ that can be written as
\begin{align}
    M_{a|x}=\int_{B^3} p(a|x,\vec{s})(\mathbb{I}+\vec{s}\cdot\vec{\sigma})d\gamma(\vec{s}),
\end{align}
where $\gamma$ is a probability measure on $B^3$ that generates the zonoid $\mc{K}_{\Pi}$ for some parent POVM $\Pi$.
Let $\mc{Z}:=Bl[\mc{K}_{\Pi}]$ be the zonoid in $\mbb{R}^3$ obtained by projecting the elements of $\mc{K}_{\Pi}$ onto their Bloch vectors.
The completion relation $\mathbb{I}=\int(\mathbb{I}+\vec{s}\cdot\vec{\sigma})d\gamma(\vec{s})$ implies that $\mc{Z}$ is centrally-symmetric.
The positivity condition $\Vert\vec{s}\Vert\leq 1$ imposes a constraint on the generator length $\ell(\mc{Z})$, as we now observe.
Starting with the support function $h_{\mc{Z}}$, we have
\begin{align}
\label{Eq:qubit-support-function-bound}
    h_{\mc{Z}}(\hat{n})\!=\!\sup_{M\in\mc{K}_{\Pi}}\!\!\hat{n}\!\cdot\! Bl[M]&=2 \sup_{q\in[0,1]}\!\int\! q(\vec{s})\hat{n}\cdot\vec{s}d\gamma(\vec{s})\notag\\
    &=2\!\int_{H_{\hat{n}}}\!\!\!\!\hat{n}\!\cdot\!\vec{s}\;d\gamma(\vec{s})=\!\!\int_{B^3}\!\!\!|\hat{n}\cdot\vec{s}|d\gamma(\vec{s}),
\end{align}
where $H_{\hat n}:=\{\vec s\in B^3:\vec s\cdot\hat n>0\}$ and the last equality follows from the central symmetry $\int_{B^3}\vec s d\gamma(\vec s)=0$.
To compute $\ell(\mc{Z})$, we integrate both sides with the uniform measure on $S^2$, which yields
\begin{align}
\label{Eq:qubit-generator-length-bound}
    \ell(\mc{Z})&=4\int_{S^2}d\mu(\hat{n})\int_{B^3}d\gamma(\vec{s}) |\hat{n}\cdot\vec{s}|\notag\\
    &=4\int_{B^3}d\gamma(\vec{s})\int_{S^2} |\hat{n}\cdot\vec{s}|d\mu(\hat{n})=2\int_{B^3}d\gamma(\vec{s})\Vert\vec{s}\Vert\leq 2.
\end{align}
We summarize these findings in the following proposition, which also has a converse for dichotomic measurements.

\begin{proposition}
\label{prop:qubits-zonoid-necessary}
The Bloch vectors $\{\vec{m}_{a|x}\}$ of any compatible assemblage $\{M_{a|x}\}$ must belong to some centrally-symmetric zonoid $\mc Z\subset\mbb{R}^3$ with $\ell(\mc Z)\leq 2$. 
Conversely, for an assemblage of noisy spin measurements $\{M_{\vec{m}}\}$, this condition is also
sufficient for compatibility. 
\end{proposition}
\begin{proof}
    It remains to prove the converse.
    Consider a family of Bloch vectors $\{\vec{m}\}$ that belong to some centrally-symmetric zonoid $\mc{Z}\subset\mathbb{R}^3$ satisfying $\ell(\mc{Z})\leq 2$.
    By simply rescaling $\mc{Z}$ by the factor $2/\ell(\mc{Z})$, we can always assume without loss of generality that $\ell(\mc{Z})=2$.
    As shown in Appendix \ref{App:zonoid-generator}, there exists a measure $\nu$ on $S^2$ such that $\mc{Z}=\{\int_{S^2}q(\hat{s})\hat{s}d\nu(\hat{s})\;|\; q\in[0,1]\}$.
   Its support function is given by $h_{\mc{Z}}(\hat{n})=\frac{1}{2}\int_{S^2}|\langle\hat{n},\hat{s}\rangle |d\nu(\hat{s})$.
   Integrating both sides of this over $S^2$ with the uniform measure $\mu$ yields 
   \begin{equation}
   2=\ell(\mc{Z})=4\int_{S^2}h_{\mc{Z}}(\hat{n})d\mu(\hat{n})=\int_{S^2}d\nu(\hat{s}).
   \end{equation}
   Hence, $\nu/2$ is a probability measure, and we use it to construct the two positive-operator zonoids 
   \begin{equation}
       \mc{K}_{\pm}=\left\{\int_{S^2}q(\hat{s})\frac{\mathbb{I}\pm\hat{s}\cdot\vec{\sigma}}{2}\frac{d\nu(\hat{s})}{2}\;\big|\;q\in[0,1]\right\}
       \end{equation}
       and their Minkowski sum $\mc{K}=\mc{K}_++\mc{K}_-$.
       Clearly $\mathbb{I}\in\mc{K}$, and by construction both $M_{\vec{m}}$ and $\mathbb{I}-M_{\vec{m}}$ are elements of $\mc{K}$ for all $\vec{m}$.
       Indeed, if $\vec{m}=\int q(\hat{s})\hat{s}d\nu(\hat{s})$, then
       \[M_{\vec{m}}=\!\int\! q(\hat{s})\frac{\mathbb{I}+\hat{s}\cdot\vec{\sigma}}{2}\frac{d\nu(\hat{s})}{2}+\!\int (1-q(\hat{s}))\frac{\mathbb{I}-\hat{s}\cdot\vec{\sigma}}{2}\frac{d\nu(\hat{s})}{2},\]
where we use the fact that $\int\hat{s}d\nu(\hat{s})=0$ since $\mc{Z}$ is centrally-symmetric.
   By Proposition \ref{Prop:dichotomic-compatibility}, compatibility of the assemblage $\{M_{\vec{m}}\}$ will follow if $\mc{K}\subset[0,\mathbb{I}]$. 
   To see this, note that every element of $\mc{K}$ has the form $\int \bm{[}a(\hat{s})\op{\hat{s}}{\hat{s}}+b(\hat{s})\op{\!-\!\hat{s}}{-\hat{s}}\bm{]}\frac{d\nu(\hat{s})}{2}$ with $0\leq a,b\leq 1$.
   Each term in the integral belongs to $[0,\mathbb{I}]$, and since $\nu/2$ is a probability measure, so will the full integral.
\end{proof}

While Proposition \ref{prop:qubits-zonoid-necessary} provides an elegant characterization of compatibility for noisy spin measurements, it does not provide a direct method for constructing the parent POVM and response function.
This will be remedied later in Theorem~\ref{Thm:incompatibility-breaking} where we present a more explicit and geometrically-intuitive parent POVM by zonotope approximation.
In the limit of these approximations, one obtains a parent zonoid like the one described in Proposition \ref{prop:qubits-zonoid-necessary}.
A very natural question is whether this limit zonoid $\mc{K}$ is unique.

Remarkably, when the bound $\ell(\mc{Z})\leq 2$ is saturated and we consider a maximal family of noisy spin measurements, the parent zonoid is completely fixed.

\begin{theorem}
\label{Thm:zonoid-unique}
Let $\mathcal{Z}\subset\mathbb{R}^3$ be a centrally-symmetric zonoid with $\ell(\mc{Z})=2$.  
Then there exists a unique probability measure $\nu$ on $S^2$ such that the POVM $\Pi_{\mc{Z}}:\mf{B}(S^2)\to\text{Pos}(\mathbb{C}^2)$ defined by $\Pi_{\mc{Z}}(X)=\int_X(\mathbb{I}+\hat{s}\cdot\vec{\sigma})d\nu(\hat{s})$  simulates the assemblage $\{M_{\vec{m}}\}_{\vec{m}\in\mc{Z}}$.
Moreover, any other POVM that can simulate the assemblage is operationally equivalent to $\Pi_{\mc{Z}}$.
\end{theorem}
\noindent The proof is given in Appendix \ref{App:boundary-zonoid-uniqueness}.
In terms of zonoids, Theorem \ref{Thm:zonoid-unique} implies that every centrally-symmetric zonoid $\mc{Z}\in\mbb{R}^3$ with $\ell(\mc{Z})=2$ can be uniquely associated with a positive-operator zonoid $\mc{K}\subset\text{Pos}(\mathbb{C}^2)$ such that any POVM $\Pi$ that can simulate the assemblage satisfies $\mc{K}_{\Pi}=\mc{K}$.

One example application of Theorem \ref{Thm:zonoid-unique}, consider the full family of noisy spin measurements $\{\frac{1}{2}(\mathbb{I}+\frac{1}{2}\vec{m}\cdot\vec{\sigma})\}_{\vec{m}\in B^3}$ with radius $\leq 1/2$.
One can see that $\ell(B^3/2)=2$, and so there is a unique zonoid $\mc{K}$ in $[0,\mathbb{I}]$ that contains all noisy spin measurements on the Bloch sphere with shrunken radius $1/2$; namely, $\mc{K}=\{\int q(\hat{s})(\mathbb{I}+\hat{s}\cdot\vec{\sigma})d\mu(\hat{s})\;|\;q\in[0,1]\}$. 
Consequently, by adding to $\mc{K}$ just a single measurement $M_{\pm}=\frac{1}{2}(\mathbb{I}\pm\vec{m}\cdot\vec{\sigma})$ with $\Vert\vec{m}\Vert>1/2$, the resulting measurement assemblage becomes incompatibility.
Indeed, by Proposition \ref{prop:qubits-zonoid-necessary}, any parent POVM that can simulate $M_{\pm}$ along with noisy spin measurements with radius $\leq 1/2$ requires that the Bloch vectors of these measurements belong to some centrally-symmetric zonotope $\mc{Z}$ that strictly contains $B^3/2$.
But since zonotopes are compact and convex, their support functions are continuous and so $\ell(\mc{Z})>\ell(B^3/2)=2$.



\subsection{Planar Measurements}
\label{Sect:Planar}

A family of qubit measurements is called planar if all their Bloch vectors lie in the same two-dimensional plane, which we always take to be the $\hat{x}$--$\hat{z}$ plane. 
A first consequence of Proposition~\ref{prop:qubits-zonoid-necessary} is an exact characterization of compatible planar spin measurements. 
For an arbitrary collection of spin directions $\{\vec m\}\subset\mathbb{R}^2$, define the centrally symmetric convex body
\begin{align}
\mc C\{\vec m\}:=\text{conv}[\{\pm\vec m\}].
\end{align}
As noted in Section \ref{Sect:Zonotopes-Zonoids}, every centrally symmetric convex  body in $\mbb R^2$ is a zonoid, and its perimeter satisfies
\begin{align}
\text{Per}(\mc C\{\vec m\})=2\ell(\mc C\{\vec m\}).
\end{align}
This planar geometry is closely related to the compatible-radius and zonotope methods studied in Refs.~\cite{Andrejic-2020a,Zhang-2025a}.

\begin{corollary}
\label{Cor:compatibility-perimeter}
A family of noisy planar spin measurements $\{M_{\vec m}\}$ is compatible if and only if their set of Bloch vectors satisfies
\begin{align}
\text{Per}(\mc C\{\vec m\})\leq 4.
\end{align}
\end{corollary}

\begin{proof}
If the measurements are compatible, Proposition~\ref{prop:qubits-zonoid-necessary} gives a centrally symmetric zonoid of generator length at most two that contains all their Bloch vectors. 
Projecting this zonoid onto the measurement plane does not increase its generator length, and therefore
\begin{align}
\text{Per}(\mc C\{\vec m\})\leq 4.
\end{align}
Conversely, if $\text{Per}(\mc C\{\vec m\})\leq4$, then $\mc C\{\vec m\}$ is itself a zonoid satisfying $\ell(\mc C\{\vec m\})\leq 2$. Proposition~\ref{prop:qubits-zonoid-necessary} therefore implies compatibility.
\end{proof}

Consider now $n$ planar projective measurements with distinct axes
\begin{align}
\hat m_x:=\cos\alpha_x \hat x+\sin\alpha_x \hat z,\quad0=\alpha_1<\cdots<\alpha_n<\pi,
\notag
\end{align}
and set $\alpha_{n+1}:=\alpha_1+\pi$. 
Their noisy versions are
\begin{align}
M^{(r)}_{b|x}:=\frac{1}{2}\left[\mathbb I+(-1)^b r\hat m_x\cdot\vec\sigma\right],
\end{align}
and the corresponding polygon has perimeter
\begin{align}
\text{Per} \left(\text{conv}\{\pm r\hat m_x\}_x\right)=4r\sum_{x=1}^n\sin\left(\frac{\alpha_{x+1}-\alpha_x}{2}\right).
\end{align}
Andrejic and Kunjwal proved that the resulting expression lower bounds the compatibility threshold and conjectured that it is exact in general~\cite{Andrejic-2020a}. 
This conjecture was subsequently supported by extensive linear-programming calculations for large families of planar measurements~\cite{Porto-2026a}. 
The preceding perimeter criterion now resolves it analytically.

\begin{corollary}
\label{Cor:planar-robustness-conjecture}
The measurements $\{M^{(r)}_{b|x}\}_{b,x}$ are compatible if and only if
\begin{align}
r\le \left[\sum_{x=1}^n\sin\left(\frac{\alpha_{x+1}-\alpha_x}{2}\right)\right]^{-1}. \label{Eq:planar-exact-threshold}
\end{align}
Consequently, among all families of $n$ planar projective measurements, the incompatible family most robust to deploarizing noise consists of equally spaced spin directions,
\begin{align}
\hat m_k:=\cos\frac{k\pi}{n} \hat x+\sin\frac{k\pi}{n} \hat z,\qquad k=0,\ldots,n-1, \label{Eq:symmetric-planar}
\end{align}
whose compatibility threshold is $\frac{1}{n\sin\left(\frac{\pi}{2n}\right)}$.

\end{corollary}

\begin{proof}
The exact threshold follows immediately from Corollary~\ref{Cor:compatibility-perimeter} and the perimeter formula above. To optimize over the directions, define $\delta_x:=\alpha_{x+1}-\alpha_x$, so that $\sum_x\delta_x=\pi$. Concavity of $\sin(\delta/2)$ gives
\begin{align}
\sum_{x=1}^n\sin\frac{\delta_x}{2}\leq n\sin\frac{\pi}{2n},
\end{align}
with equality if and only if $\delta_x=\pi/n$ for every $x$. Hence equally spaced axes minimize the compatibility threshold and are therefore the most robustly incompatible.
\end{proof}





\section{Incompatibility-Breaking Qubit Channels}

\label{Sect:Incompatibility-breaking-qubit}

We now consider the structure of compatibility breaking channels in qubit systems.
We describe a qubit channel $\mc{N}$ through the Bloch representation of its Heisenberg adjoint. For a POVM element $M=\alpha\mathbb{I}+\frac{1}{2}\vec{m}\cdot\vec{\sigma}$, suppose that
\begin{align}
\label{Eq:qubit-channel-adjoint} \mc{N}^\dagger(M) = \left(\alpha+\frac{1}{2}\vec{t}\cdot\vec{m}\right)\mathbb{I} + \frac{1}{2}(D\vec{m})\cdot\vec{\sigma},
\end{align}
where $D$ is a real $3\times 3$ matrix and $\vec{t}\in\mbb{R}^3$.
In other words, the four-vector of the POVM element transforms as $(\alpha,\vec{m})\mapsto (\alpha+1/2\vec{t}\cdot\vec{m}, \;D\vec{m})$.
Equivalently, the Schr\"odinger action of $\mc{N}$ on the Bloch vector $\vec{s}$ of an arbitrary qubit state $\rho$ is $\vec{s}\mapsto D^{\mathtt{T}}\vec{s}+\vec{t}$.
These two representations satisfy
\begin{align}
\tr[M \mc{N}(\rho)]=\tr[\mc{N}^\dagger(M)\rho].
\end{align}

\subsection{A Criterion for Incompatibility-Breaking}
\label{Sect:Incompatibility-Breaking-Criterion}
We now derive a sufficient condition for a qubit channel $\mc N$ to break the incompatibility of a family of noisy spin measurements $\mathtt{S}=\{\vec{m}\}$. 
Using Eq.~\eqref{Eq:qubit-channel-adjoint}, we must decide when the measurements $\{\mc{N}^\dagger(M_{b|\vec{m}})\}_{\vec{m}\in\mathtt{S}}$ are compatible.
 
\begin{theorem}
\label{Thm:incompatibility-breaking}
    Let $\mc{N}$ be any qubit channel whose Heisenberg adjoint is described by the matrix $D$ and translation vector $\vec{t}$ as in Eq.~\eqref{Eq:qubit-channel-adjoint}.
    Then $\mc{N}$ is incompatibility breaking for a family of spin directions $\mathtt{S}=\{\vec{m}\}$ whenever
    \begin{align}
    \label{Eq:qubit-breaking-inclusion} \left\{ \frac{D\vec{m}}{1-|\vec{t}\cdot\vec{m}|} \bigg| \vec{m}\in\mathtt{S}\right\} \subset\mc Z
    \end{align}
    for some centrally symmetric zonoid $\mc Z\subset\mathbb{R}^3$ with $\ell(\mc Z)\leq2$. Here we assume $|\vec t|<1$, otherwise, the channel is forced to be a measure-prepare channel and is trivially incompatibility breaking.

\end{theorem}

\begin{remark}
    The proof of Theorem \ref{Thm:incompatibility-breaking} involves constructing a parent POVM for the measurements $\{\mc{N}^\dagger(M_{b|\vec{m}})\}_{\vec{m}\in\mathtt{S}}$.
    The construction is distinct from the one in Proposition \ref{prop:qubits-zonoid-necessary} for two reasons.
    First, the parent POVM above applies only to unbaised measurements, whereas the one constructed below is capable of simulating biased measurements, which arise whenever the channel $\mc{N}$ is nonunital.
    
    Second, the construction here involves building a zonoid through a sequence of approximating zonotopes. 
    This method has a practical advantage from an implementation perspective since every zonotope corresponds to a POVM with a finite number of measurement outcomes, and Theorem \ref{Thm:incompatibility-breaking} shows how to construct this measurement in terms of Bloch sphere geometry.
    In contrast, a general zonoid corresponds to POVM that might have an unbounded number of measurement outcomes and require an ancilla system of unbounded dimension to physically implement \cite{Zhang-2025a}.
\end{remark}

\begin{proof}

Suppose that the inclusion in Eq.~\eqref{Eq:qubit-breaking-inclusion} holds for a centrally symmetric zonoid $\mc Z$ satisfying $\ell(\mc Z)\leq2$. Let $\epsilon>0$ and choose a centrally symmetric zonotope $\mc Z_\epsilon$ such that $d_H(\mc Z_\epsilon,\mc Z)<\epsilon$.
Since $|h_{\mc Z_\epsilon}(\hat n)-h_{\mc Z}(\hat n)|\leq\epsilon$ for every $\hat n\in S^2$, Eq.~\eqref{Eq:generator-length-integral} gives
\begin{align}
\ell(\mc Z_\epsilon)\leq\ell(\mc Z)+4\epsilon\leq2+4\epsilon.
\end{align}
Consequently,
\begin{align}
\min\{2,\ell(\mc Z_\epsilon)\}\leq\ell(\mc Z_\epsilon)
\leq\min\{2,\ell(\mc Z_\epsilon)\}+4\epsilon.
\label{Eq:Generator-length-bound}
\end{align}

Let $\{\vec n_\lambda\}_\lambda$ be generators of $\mc Z_\epsilon$, chosen such that $\sum_\lambda\vec n_\lambda=\vec0$ and $\ell(\mc Z_\epsilon)=\sum_\lambda\Vert\vec n_\lambda\Vert$. 
For every $\vec m\in\mathtt S$, the Hausdorff approximation and Eq.~\eqref{Eq:qubit-breaking-inclusion} imply that there exist coefficients $0\leq t_\lambda(\vec m)\leq1$ such that
\begin{align}
\frac{D\vec m}{1-|\vec t\cdot\vec m|}
\approx_\epsilon\!
\sum_\lambda t_\lambda(\vec m)\vec n_\lambda
=\!
\sum_\lambda[2t_\lambda(\vec m)-1]
\frac{\ell(\mc Z_\epsilon)}{2}
\frac{\vec n_\lambda}{\ell(\mc Z_\epsilon)},
\notag
\end{align}
where the equality follows from $\sum_\lambda\vec n_\lambda=\vec0$. Define
\begin{align}
s_\lambda(\vec m):=[2t_\lambda(\vec m)-1]
\frac{\min\{2,\ell(\mc Z_\epsilon)\}}{2}\in[-1,1].
\end{align}
By Eq.~\eqref{Eq:Generator-length-bound}, this replacement introduces an
additional error of at most $2\epsilon$, and therefore
\begin{align}
\frac{D\vec m}{1-|\vec t\cdot\vec m|}
\approx_{3\epsilon}
\sum_\lambda s_\lambda(\vec m)
\frac{\vec n_\lambda}{\ell(\mc Z_\epsilon)}.
\label{Eq:zonoid-finite-approximation}
\end{align}
Define $p_\lambda:=\frac{\Vert\vec n_\lambda\Vert}{\ell(\mc Z_\epsilon)}$, $\hat n_\lambda:=\frac{\vec n_\lambda}{\Vert\vec n_\lambda\Vert}$ and the parent POVM
\begin{align}
\Pi_{\lambda,c}^{(\epsilon)}:=p_\lambda M_{c|\hat n_\lambda},
\qquad c\in\{0,1\}.
\label{Eq:zonoid-finite-parent}
\end{align}
Indeed, $\sum_{\lambda,c}\Pi_{\lambda,c}^{(\epsilon)}=\mbb I$. 
For each $\vec m\in\mathtt S$, define
\begin{align}
&f_\epsilon(b|\vec m,\lambda,c)  \notag\\
\label{Eq:zonoid-finite-response}
&\!\!\!\!:=\frac{1}{2}\left[1+(-1)^b\left(\vec t\cdot\vec m+(-1)^c(1-|\vec t\cdot\vec m|)s_\lambda(\vec m)\right)\right].
\end{align}
These are valid response functions since 
$$\left|\vec t\cdot\vec m+(-1)^c(1-|\vec t\cdot\vec m|)s_\lambda(\vec m)\right|\leq 1.$$ 
Using Eq.~\eqref{Eq:zonoid-finite-approximation}, we obtain
\begin{align}
\mc N^\dagger(M_{b|\vec m})
\approx_{3\epsilon}
\sum_{\lambda,c}
f_\epsilon(b|\vec m,\lambda,c)\Pi_{\lambda,c}^{(\epsilon)}.
\label{Eq:finite-parent-approximation}
\end{align}
Thus, for every $\epsilon>0$, the noisy measurements are uniformly
approximated by a compatible family with the parent POVM in Eq.~\eqref{Eq:zonoid-finite-parent}. 
Their compatibility
then follows from Theorem~\ref{Thm:dichotomic-closure}.
\end{proof}


\subsection{Unital and Nonunital Channels}

\label{Sect:Incompatibility-Breaking-Unital}

When $\vec{t}=\vec{0}$, the condition in Eq. \eqref{Eq:qubit-breaking-inclusion} reduces to the one in Proposition \ref{prop:qubits-zonoid-necessary} since $\{\mc N^\dagger(M_{b|\vec m})\}$ is again an assemblage of noisy spin measurements with Bloch vectors $\{D\vec m\}$. 
The condition $\vec{t}=\vec{0}$ corresponds to a unital channel $\mc{N}$.
Thus, we have that a unital channel will break the incompability in a family of spins $\mathtt{S}$ iff $\{D\vec m:\vec m\in\mathtt S\}$ is contained in a centrally symmetric zonoid $\mc Z$ satisfying $\ell(\mc Z)\leq 2$. 

The following corollary to Theorem~\ref{Thm:incompatibility-breaking} applies this to the full set of projective measurements and includes a partial extension to nonunital chanels as well.

\begin{corollary}
\label{Cor:projective-measurements-incomp}
    A qubit channel $\mc{N}$ is PVM-incompatibility breaking whenever 
    \begin{align}
    \label{Eq:projective-measurement-cond} {2}\int_{S^2}\Vert D\hat{n}\Vert d\mu(\hat{n})+{\Vert\vec{t}\Vert}\leq 1,
    \end{align}
    where $(\alpha,\vec{m})\mapsto (\alpha+1/2\vec{t}\cdot\vec{m}, \;D\vec{m})$ describes the action of $\mc{N}^\dagger$ on the Bloch sphere, as in Eq.~\eqref{Eq:qubit-channel-adjoint}. 
    Conversely, if $\mc{N}$ is unital, then this condition is also necessary for $\mc{N}$ to be PVM-incompatibility breaking. 
\end{corollary}

\begin{proof}
We apply Theorem~\ref{Thm:incompatibility-breaking} with $\mathtt S=S^2$ consisting of all the spin measurements. Consider the ellipsoid
\begin{align}
\mc E:=\left\{\frac{D\vec m}{1-\norm{\vec t}}\;\middle|\;\vec m\in B^3\right\},
\label{Eq:steering-ellipsoid}
\end{align}
which forms a zonoid in $\mathbb{R}^3$.
Since $|\vec t\cdot\vec m|\leq\norm{\vec t}$ for $\vec m\in B^3$,
\begin{align}
\frac{D\vec m}{1-|\vec t\cdot\vec m|}
=\frac{D}{1-\norm{\vec t}}
\left(\frac{1-\norm{\vec t}}{1-|\vec t\cdot\vec m|}\vec m\right)\in\mc E.
\end{align}
Moreover,
\begin{align}
h_{\mc E}(\hat n)
=\frac{1}{1-\norm{\vec t}}\max_{\vec m\in B^3}(D^{\mathtt T}\hat n)\cdot\vec m
=\frac{\norm{D^{\mathtt T}\hat n}}{1-\norm{\vec t}},
\end{align}
and using $\int_{S^2}\norm{D^{\mathtt T}\hat n}d\mu(\hat n) =\int_{S^2}\norm{D\hat n}d\mu(\hat n)$, we therefore
\begin{align}
\ell(\mc E)&=4\int_{S^2}h_{\mc E}(\hat n)d\mu(\hat n) \notag \\
&=\frac{4}{1-\norm{\vec t}}\int_{S^2}\norm{D\hat n}d\mu(\hat n)\le 2,
\end{align}
with the last line following from Eq.~\eqref{Eq:projective-measurement-cond}.
By Theorem~\ref{Thm:incompatibility-breaking} it follows that $\mathcal{N}$ is PVM-incompatibility breaking.

In the case that $\mc{N}$ is unital, the set $\{D\vec{m} | \vec{m}\in S^2\}$ already forms the surface of an ellipsoid $\mc{E}$.
Thus, if $\mc{N}$ is PVM-incompatibility breaking, Proposition~\ref{prop:qubits-zonoid-necessary} can be applied to conclude that $\ell(\mc{E})\leq 2$.
\end{proof}

\begin{remark}
    As proved in a parallel paper,  the criterion in Eq.~\eqref{Eq:projective-measurement-cond} is also sufficient for a channel $\mc N$ to be incompatibility breaking for all POVMs~\cite{zhang2026parallel}. 
    Also, note that the expression $\frac{1}{1-\Vert\vec{t}\Vert}\int_{S^2}\Vert D^{\mathtt{T}}\hat{n}\Vert d\mu(\hat{n})$ is the mean radius of the ellipsoid $\mc{E}$ given in Eq.~\eqref{Eq:steering-ellipsoid}.
    Hence, Corollary~\ref{Cor:projective-measurements-incomp} combined with Ref. \cite{zhang2026parallel} says that $\mc{N}$ is incompatibility breaking whenever this mean radius is no greater than $1/2$.
    \end{remark}

\subsection{Examples}
\label{Sect:Examples}
We next present two examples of incompatibility breaking channels.
\subsubsection{Pauli Channels}
As a direct application of the exact unital criterion, consider Pauli channels
\begin{align}
\mc{N}(\cdot)=\sum_{b_0,b_1=0}^1p_{b_0b_1}\sigma_z^{b_0}\sigma_x^{b_1}(\cdot)\sigma_x^{b_1}\sigma_z^{b_0}.
\end{align}
The corresponding adjoint Bloch matrix $D$ is diagonal, $D=\text{diag}[c_x,c_y,c_z]$ with $c_x=p_{00}+p_{01}-p_{10}-p_{11}$, $c_y=p_{00}-p_{01}-p_{10}+p_{11}$, and $c_z=p_{00}-p_{01}+p_{10}-p_{11}$.
Corollary~\ref{Cor:projective-measurements-incomp} then gives the incompatibility breaking condition of
\begin{align}
2\int_{S^2}\sqrt{c_x^2n_x^2+c_y^2n_y^2+c_z^2n_z^2}d\mu(\hat{n})\le 1.
\end{align}
This recovers previously known results on steerability of Bell diagonal states~\cite{Nguyen-2016a, Nguyen-2019a}.

\subsubsection{Depolarized Amplitude Damping Channels}

The nonunital case is more subtle because the set described in Eq. \eqref{Eq:qubit-breaking-inclusion} of Theorem~\ref{Thm:incompatibility-breaking} is generally not an ellipsoid. 
However, for the depolarized amplitude-damping channel, its convex hull can be determined exactly.

Let $\mc A_\gamma$ denote the amplitude-damping channel, with
\begin{align}
D_\gamma=\operatorname{diag}\left(\sqrt{1-\gamma},\sqrt{1-\gamma},1-\gamma\right),\quad \vec t_\gamma=\gamma\hat z.
\end{align}
Consider the combined channel $\mc N_{\gamma,r}:=\Delta_r\circ\mc A_\gamma$, where depolarizing noise is applied after the amplitude damping. 
We have $D_{\gamma,r}=rD_\gamma$ and $\vec t_{\gamma,r}=r\gamma\hat z$ for $\mc N_{\gamma,r}$. 
The set described in Eq. \eqref{Eq:qubit-breaking-inclusion} is therefore
\begin{align}
\mc K_{\gamma,r}:=\left\{\frac{rD_\gamma\hat m}{1-r\gamma |m_z|}\;\middle|\;\hat m\in S^2\right\},
\end{align}
where $m_z$ denotes the $z$-component of $\hat{m}$.
Write $u:=|m_z|\in[0,1]$ so that the boundary of $\mc K_{\gamma,r}$ in the region $m_z\ge 0$ can be expressed in cylindrical coordinates as
\begin{align}
\rho(u)=\frac{r\sqrt{1-\gamma}\sqrt{1-u^2}}{1-r\gamma u},\qquad z(u)=\frac{r(1-\gamma)u}{1-r\gamma u}.
\end{align}
Defining
\begin{align*}
R_{\gamma,r}=\!\frac{r\sqrt{1-\gamma}}{\sqrt{1-r^2\gamma^2}},~C_{\gamma,r}=\!\frac{r(1-\gamma)}{1-r^2\gamma^2},~s_{\gamma,r}=\!\frac{r^2\gamma(1-\gamma)}{1-r^2\gamma^2},
\end{align*}
one can verify directly that
\begin{align}
\frac{\rho(u)^2}{R_{\gamma,r}^2}+\frac{(z(u)-s_{\gamma,r})^2}{C_{\gamma,r}^2}=1.
\end{align}
Combined with the reflected body under $z\mapsto-z$, we see that the convex hull of $\mc{K}_{\gamma,r}$ is the Minkowski sum
\begin{align}
\conv(\mc K_{\gamma,r})&=\operatorname{diag}\left(R_{\gamma,r},R_{\gamma,r},C_{\gamma,r}\right)B^3+[-s_{\gamma,r}\hat z,s_{\gamma,r}\hat z]\notag  \\
&=:\mc Z_{\gamma,r}^{\star}
\end{align}
Since the Minkowski sum of an ellipsoid and line segment is a zonoid, $\mc Z_{\gamma,r}^{\star}$ is the the smallest enclosing zonoid of $\mc K_{\gamma,r}$, and its generator length
\begin{align}
\ell(\mc Z_{\gamma,r}^{\star})=\frac{2r(1-\gamma)}{1-r\gamma}+\frac{2r\sqrt{1-\gamma}}{\sqrt{\gamma(1-r^2\gamma)}}\arctan\sqrt{\frac{\gamma(1-r^2\gamma)}{1-\gamma}}.
\end{align}

When $r=1$, we reduce to the amplitude damping channel with no depolarizing noise.
In this case,
\begin{align}
\ell(\mc Z_{\gamma,1}^{\star})=2+\frac{2}{\sqrt{\gamma}}\arctan\sqrt{\gamma}>2
\end{align}
for every $0<\gamma<1$. 
Hence, we are unable to identify any damping parameter in the region $0<\gamma<1$ for which $\mc{A}_\gamma$ is PVM-incompatibility breaking.
This is consistent with recent work implying that the Choi state of the amplitude damping channel can demonstrate EPR steering \cite{Song-2026a}.
This alone does not imply that the amplitude damping channel is not incompatibility in the region

On the other hand, when some depolarization is allowed, Theorem~\ref{Thm:incompatibility-breaking} certifies PVM-incompatibility breaking whenever
\begin{align}
\ell(\mc Z_{\gamma,r}^{\star})\leq2.
\end{align}
As one can check, for any $r<1$ there always exists a $\gamma$ such that $\ell(\mc{Z}_{\gamma,r}^{\star})\leq 2$, thereby establishing that $\mc{Z}_{\gamma,r}^{\star}$ is PVM-incompatibility breaking.  

\subsection{Two-Copy Activation of Incompatibility-breaking channels}
\label{Sect:superactivation}

We next move to the bipartite setting with two systems that are held by Alice ($\A$) and Bob ($\B$), respectively.
A special type of measurement on the joint system $\A\B$ involves Alice and Bob measurements independently performing local POVMs followed by joint classical post-processing.
We refer to such measurements as being \textit{freely localizable} \cite{Beckman-2001a, Pauwels-2025a}, and every freely localizable POVM $\Pi_{\A\B}$ has the form
\begin{align}
    E\mapsto \Pi_{\A\B}(E)=\int_{\Lambda_{\A}\times\Lambda_{\B}}q_E(\lambda_\A,\lambda_\B)d\Pi_{\A}(\lambda_\A) d\Pi_{\B}(\lambda_\B),\notag
\end{align}
where $\Pi_{\A}$ and $\Pi_{\B}$ are local POVMs.
Mathematically, the family of dichotomic bipartite measurements that can be simulated by a freely localizable POVM forms a tensor product zonoid $\mc{K}_{\Pi_{\A}}\otimes^{\mc{Z}}\mc{K}_{\Pi_{\B}}$ \cite{Aubrun-2015a}.

It is not difficult to show that two assemblages $\mf{M}_{\A}=\{M_{a|x}\}$ and $\mf{M}_{\B}=\{M_{b|y}\}$ are compatible iff their joint assemblage $\mf{M}_{\A}\otimes\mf{M}_{\B}=\{M_{a|x}\otimes M_{b|y}\}$ is compatible.
This is true despite the fact that joint measurements across both systems are possible, thereby allowing for more powerful parent POVMs than freely localizable ones.
Hence, the additional power of non-localizable measurements for simulating POVMs becomes manifest only when considering assemblages of non-product POVMs.
Indeed, there exists compatible assemblages $\mf{M}_{\A\B}$ consisting of non-product POVMs that cannot be simulated by a freely localizable parent POVM.
However, our first result shows that such a scenario is not possible whenever $\mf{M}_{\A\B}$ contains noisy spin measurements at the compatibility threshold on each system.
\begin{theorem}
\label{Thm:two-copy-unique}
    Consider arbitrary noisy spin assemblages $\mathfrak{M}_{\A}=\{M_{\vec{m}}\}_{\vec{m}\in\mc{Z}_\A}$ and $\mathfrak{M}_{\B}=\{M_{\vec{m}'}\}_{\vec{m}'\in\mc{Z}_\B}$ such that $\ell(\mc{Z}_{\A})=\ell(\mc{Z}_{\B})=2$.
    Let $\Pi_{\A}$ and $\Pi_{\B}$ be POVMs that can simulate $\mathfrak{M}_{\A}$ and $\mathfrak{M}_{\B}$, respectively. 
    Then any bipartite POVM $\Pi_{\A\B}$ that can simulate the product assemblage $\mathfrak{M}_{\A}\otimes \mathfrak{M}_{\B}$ must be operationally equivalent to the product POVM $\Pi_{\A}\otimes\Pi_{\B}$.
\end{theorem}
\noindent The proof of this theorem is given in Appendix \ref{App:two-copy-boundary-parent}, and it can be seen as an extension of Theorem \ref{Thm:zonoid-unique} to the bipartite setting.

We now apply Theorem \ref{Thm:two-copy-unique} directly to the problem of activating incompatibility-breaking channels.
This problem asks whether the property of incompatibility-breaking is stable under parallel composition.
That is, if $\mc N$ is incompatibility breaking, must
$\mc N\otimes\mc N$ also be incompatibility breaking?
In terms of quantum steering, this is equivalent to the question of whether two unsteerable states can become EPR steerable when used in parallel.
Multicopy superactivation of steering is known~\cite{Quintino2016a}, and these results translate into the statement that incompatibility-breaking channels exist that fail to be incompatibility-breaking when a sufficiently large number of copies are used.
In particular, while the partially depolarizing channel 
\begin{align}
\Delta_r(X)=rX+(1-r)\frac{\tr X}{2}\mbb I
\end{align}
is incompatibility-breaking at the threshold value $r=1/2$,
this is no longer the case when using $n=6$ copies of $\Delta_{1/2}$ in parallel.
Here, we demonstrate that such a phenomenon occurs in a particularly sharp form: $\Delta_{1/2}^{\otimes n}$ fails to be incompatibility-breaking already at $n=2$.

\begin{theorem}
\label{Thm:two-copy-boundary-activation}
The channel $\Delta_{1/2}\otimes\Delta_{1/2}$ is not incompatibility
breaking.
\end{theorem}

\begin{proof}
Consider the measurement family consisting of all noisy product PVMs
\begin{align}
M_{\hat m_1,\hat m_2}:=
\Delta_{1/2}(M_{\hat m_1})
\otimes
\Delta_{1/2}(M_{\hat m_2})
\end{align}
together with the noisy Bell measurement
\begin{align}
\Delta_{1/2}^{\otimes2}(\Phi_k),
\qquad k=0,\ldots,3,
\label{eq:noisy-Bell}
\end{align}
where $\{\Phi_k\}_{k=0}^3$ is the standard Bell PVM. Suppose, for contradiction, that this family is compatible.
A parent POVM that simulates $\{\Delta_{1/2}(M_{\hat m})\}_{\hat{m}\in S^2}$ is given by $X\mapsto \int_X(\mathbb{I}+\hat{s}\cdot\vec{\sigma})d\mu(\hat{s})$, where $\mu$ is the generating measure on $S^2$ \cite{Heinosaari-2015a}.
Since $\ell(B^3/2)=2$, Theorem \ref{Thm:two-copy-unique} implies that the two-qubit parent POVM generated by $\mu\times\mu$ on $S^2\times S^2$ must be able to simulate $\{M_{\hat{m}_1,\hat{m}_2}\}$.
Let
\begin{align}
\vec v^{(0)}&=(1,-1,1),&
\vec v^{(1)}&=(1,1,-1),& \notag \\
\vec v^{(2)}&=(-1,1,1),&
\vec v^{(3)}&=(-1,-1,-1)
\end{align}
be the Bell-tetrahedron vectors. Define
\begin{align}
B_k^{(\eta)}:=\Delta_{\eta}^{\otimes2}(\Phi_k)=\frac{1}{4}\left[\mbb I\otimes\mbb I+\eta^2\sum_{i=1}^3v_i^{(k)}\sigma_i\otimes\sigma_i\right].
\label{Eq:noisy-Bell-effects}
\end{align}
The Bell measurement in Eq.~\eqref{eq:noisy-Bell} corresponds to
$B_k^{(1/2)}$.

Suppose that $\{B_k^{(\eta)}\}_k$ can also be simulated by the parent POVM,
so that
\begin{align}
&B_k^{(\eta)} \label{Eq:Bell-product-Haar-simulation} \\
&=\int_{S^2\times S^2}\!\!\!\!\!\!\!\!\!\!f_k(\hat n_1,\hat n_2)\left(\mbb I+\hat n_1\cdot\vec\sigma\right)\otimes\left(\mbb I+\hat n_2\cdot\vec\sigma\right)d\mu(\hat n_1)d\mu(\hat n_2),\notag 
\end{align}
where $f_k\geq0$ and $\sum_k f_k=1$. Comparing the correlation
coefficients in Eq.~\eqref{Eq:noisy-Bell-effects} and Eq.~\eqref{Eq:Bell-product-Haar-simulation} gives
\begin{align}
\int f_k(\hat n_1,\hat n_2)n_{1i}n_{2j}d\mu(\hat n_1)d\mu(\hat n_2)
=\frac{\eta^2}{4}v_i^{(k)}\delta_{ij}.
\label{Eq:Bell-correlation-constraint}
\end{align}
Using
$\sum_{k=0}^3v_i^{(k)}v_j^{(k)}=4\delta_{ij}$, we obtain
\begin{align}
\eta^2&=\frac{1}{3}\int (\sum_{k=0}^3f_k(\hat n_1,\hat n_2)\vec v^{(k)})\cdot(\hat n_1\circ\hat n_2)d\mu(\hat n_1)d\mu(\hat n_2)\notag\\
&\leq\frac{1}{3}\int\max_k\vec v^{(k)}\cdot(\hat n_1\circ\hat n_2)d\mu(\hat n_1)d\mu(\hat n_2).
\label{Eq:Bell-tetrahedron-bound}
\end{align}
Here $\circ$ denotes the Hadamard product. Using
\begin{align*}
\max_k\vec v^{(k)}\cdot\vec z=\sum_{i=1}^3|z_i|-2\min_i|z_i|\mathbf \chi_{\{z_1z_2z_3>0\}},
\end{align*}
with $\chi_{\{z_1z_2z_3>0\}}$ being the indicator function of the
condition $z_1z_2z_3>0$, and
\begin{align*}
\int\sum_{i=1}^3|n_{1i}n_{2i}|d\mu(\hat n_1)d\mu(\hat n_2)=\sum_{i=1}^3\left(\int|n_i|d\mu(\hat n)\right)^2=\frac{3}{4},
\end{align*}
we have
\begin{align}
\eta^2\leq\frac{1}{4}-\frac{2}{3}I,\label{Eq:Bell-strict-bound}
\end{align}
where 
\begin{align*}
I&:=\int_{S^2\times S^2}\min_i|n_{1i}n_{2i}|\mathbf \chi_{\{\prod_i n_{1i}n_{2i}>0\}}d\mu(\hat n_1)d\mu(\hat n_2).
\end{align*}

It remains to show that $I>0$. Define the subsets
\begin{align}
\mc C
:=
\left\{
\hat n\in S^2:
n_i\geq\frac{1}{2}
\text{ for }i=1,2,3
\right\},
\end{align}
which are clearly nonempty ($\mu(\mc C)\approx0.003>0$). Restricting the integral defining $I$ to $\mc C\times\mc C$ therefore gives
\begin{align}
I
\geq
\int_{\mc C\times\mc C}
\min_i|n_{1i}n_{2i}|
d\mu(\hat n_1)d\mu(\hat n_2)
\geq
\frac{1}{4}\mu(\mc C)^2
>0.
\end{align}
Consequently,
\begin{align}
\eta^2\leq\frac{1}{4}-\frac{2}{3}I<\frac{1}{4}.
\end{align}
This contradicts the value $\eta=1/2$ of the noisy Bell measurement
$\{\Delta_{1/2}^{\otimes2}(\Phi_k)\}_k$. Hence the noisy Bell measurement cannot be simulated by the unique parent POVM of the noisy product PVMs, and so $\Delta_{1/2}\otimes\Delta_{1/2}$ is not incompatibility breaking.
\end{proof}

\begin{remark}
Through the channel--state correspondence, this also gives a two-copy steering-activation result for the corresponding Werner states. 
Moreover, a robust extension of Theorem~\ref{Thm:two-copy-boundary-activation} demonstrating superactivation of $\Delta_{\eta}\otimes\Delta_{\eta}$ for $\eta\in[1/2-\epsilon ,1/2]$ will be presented in future work.
\end{remark}

\section{Incompatibility-Preserving Channels and States}

We now turn to the opposite extreme of incompatibility-breaking channels. 
Recall that a channel $\mc N$ is incompatibility-preserving if $\{\mc N^\dagger(M_{a|x})\}_{a,x}$ is incompatible for every incompatible measurement assemblage $\{M_{a|x}\}_{a,x}$.

\subsection{Qubit Channels}
\label{Sect:incompatibility-preserving-qubits}
The planar compatibility criterion gives a direct characterization of incompatibility-preserving qubit channels.

\begin{theorem}
\label{Thm:incompatibility-preserving-qubit}
A qubit channel is incompatibility preserving if and only if it is unitary.
\end{theorem}

\begin{proof}
If $\mc N$ is unitary, then $\mc N^\dagger$ is reversible and therefore preserves compatibility and incompatibility.

Conversely, suppose that $\mc N$ is not unitary. Let $D$ and $\vec t$ be its Heisenberg parameters as in Eq.~\eqref{Eq:qubit-channel-adjoint}. Choose two orthogonal vectors $\vec m_0,\vec m_1\in\vec t^\perp$ having the same length
\begin{align}
\Vert\vec m_0\Vert=\Vert\vec m_1\Vert=q>\frac{1}{\sqrt2},
\end{align}
when $1/\sqrt2\le q\leq1$, the spin measurements
\begin{align}
M_{\vec m_k}=\frac{1}{2}\left[\mbb I+\vec m_k\cdot\vec\sigma\right], \qquad k=0,1,
\end{align}
are incompatible; this can also be seen from the fact that the polygon generated by $\{\pm\vec m_0,\pm\vec m_1\}$ has perimeter
\begin{align}
2\Vert\vec m_0-\vec m_1\Vert+2\Vert\vec m_0+\vec m_1\Vert=4\sqrt2q>4.
\end{align}
Since $\vec t\cdot\vec m_k=0$, we have:
\begin{align}
\mc N^\dagger(M_{\vec m_k}) = \frac{1}{2} \left[ \mbb I+(D\vec m_k)\cdot\vec\sigma \right].
\end{align}
Define $\hat b_\pm:=\frac{\vec m_0\pm\vec m_1}{\sqrt{\Vert\vec m_0\Vert^2+\Vert\vec m_1\Vert^2}}=\frac{\vec m_0\pm\vec m_1}{\sqrt2q}$.
The vectors $\hat b_+$ and $\hat b_-$ are orthonormal. And let $\lambda_1\geq\lambda_2\geq\lambda_3$ be the eigenvalues of
$D^{\mathtt T}D$, the polygon generated by the output Bloch vectors $D\vec m_0$ and $D\vec m_1$ has perimeter
\begin{align}
& 2\Vert D(\vec m_0-\vec m_1)\Vert +2\Vert D(\vec m_0+\vec m_1)\Vert\notag \\
& =2\sqrt2q\left(\Vert D\hat b_-\Vert+\Vert D\hat b_+\Vert\right) \leq4q\sqrt{\Vert D\hat b_-\Vert^2+\Vert D\hat b_+\Vert^2}\notag \\
& \leq4q\sqrt{\lambda_{1}+\lambda_{2}}, \label{Eq:channel-output-parallelogram-bound}
\end{align}
Choose $1/\sqrt2<q\le 1/\sqrt{\lambda_1+\lambda_2}$, the right-hand side of Eq.~\eqref{Eq:channel-output-parallelogram-bound} can fail to be strictly smaller than four for every such choice only if $\lambda_1=\lambda_2=1$.

Let $\xi_i$ be the signed singular values of $D$, so that $\xi_i^2=\lambda_i$. The generalized Fujiwara--Algoet conditions~\cite{Braun2014} imply
\begin{align}
& (\xi_i\pm\xi_j)^2\leq(1\pm\xi_k)^2. \\
& \|\vec t\|^2\leq1-\sum_{i=1}^3\xi_i^2+2\xi_1\xi_2\xi_3
\end{align}
Since $|\xi_1|=|\xi_2|=1$, these inequalities force
$|\xi_3|=1$ and $\vec t=0$. Thus $\mc N$ is unitary, contrary to the assumption.
\end{proof}

\subsection{Arbitrary Dimensions}
\label{Sect:incompatibility-preserving-arbitrary}

The qubit proof above is geometric and constructive. We now show that the same classification holds in every finite dimension. The general channel
result follows from a known rigidity theorem for maps preserving the coexistency relation between quantum effects.

For an effect $E$, let $\msf M^E:=\{E,\mbb I-E\}$ denote the corresponding binary POVM. Two effects $E$ and $F$ are coexistent if and only if
$\msf M^E$ and $\msf M^F$ are compatible.

We use the following consequence of
Ref.~\cite[Theorem~1.2]{MoriSemrl2021}.

\begin{lemma}
\label{Lem:linear-coexistency-preserver}
Let $\Lambda:\mathcal L(\mc H_{\A})\longrightarrow\mathcal L(\mc H_{\B})$
be a linear, completely positive, and unital map, where $\dim\mc H_{\A}=\dim\mc H_{\B}=d<\infty$. Suppose that, for all effects
$E,F$,
\begin{align}
& \msf M^E\text{ is compatible with }\msf M^F \notag \\
\Longleftrightarrow\quad &
\msf M^{\Lambda(E)}\text{ is compatible with }
\msf M^{\Lambda(F)}.
\label{Eq:binary-compatibility-reflection}
\end{align}
Then there exists a unitary $U:\mc H_{\A}\to\mc H_{\B}$ such that
\begin{align}
\Lambda(X)=UXU^\dagger
\end{align}
for every $X\in\mathcal L(\mc H_{\A})$.
\end{lemma}

\begin{proof}
By Ref.~\cite[Theorem~1.2]{MoriSemrl2021}, a continuous map on quantum effects that preserves coexistency in both directions must be given by
either a unitary or antiunitary conjugation, possibly composed with orthocomplementation. The orthocomplementation cases are excluded by
unitality, since orthocomplementation maps $\mbb I$ to $\mbb 0$. An antiunitary conjugation can be represented, in a fixed basis, as transposition followed by unitary conjugation. Since transposition is not completely positive for $d\geq2$, this possibility is excluded by the complete positivity of $\Lambda$. Hence only unitary conjugation remains.
\end{proof}

We can now obtain the channel classification immediately.

\begin{theorem}
\label{Thm:incompatibility-preserving-general-channel}
Let $\mc N:\mathcal L(\mc H_{\B})\longrightarrow\mathcal L(\mc H_{\A})$
be a quantum channel, where
$\dim\mc H_{\A}=\dim\mc H_{\B}=d<\infty$. Then $\mc N$ is incompatibility preserving if and only if it is unitary.
\end{theorem}

\begin{proof}
If $\mc N$ is unitary, then $\mc N^\dagger$ is reversible and therefore preserves both compatibility and incompatibility.

Conversely, suppose that $\mc N$ is incompatibility preserving. Since $\mc N^\dagger$ is positive and unital, it automatically maps every compatible measurement family to a compatible measurement family.
On the other hand, incompatibility preservation implies, by contraposition, that if $\mc N^\dagger(\msf M)$ is compatible, then $\msf M$ must have been compatible. Hence compatibility is preserved in both directions. In particular, for all effects $E,F$,
\begin{align}
& \msf M^E\text{ is compatible with }\msf M^F \notag \\
\Longleftrightarrow\quad &
\msf M^{\mc N^\dagger(E)}\text{ is compatible with }
\msf M^{\mc N^\dagger(F)}.
\end{align}
Since $\mc N^\dagger$ is linear, completely positive, and unital,
Lemma~\ref{Lem:linear-coexistency-preserver} implies
\begin{align}
\mc N^\dagger(X)=UXU^\dagger
\end{align}
for some unitary $U$. Hence $\mc N$ is unitary.
\end{proof}
\subsection{Bipartite States}
\label{Sect:States-incompatibility-preserving}
We next translate Theorem~\ref{Thm:incompatibility-preserving-qubit} into a steering statement. We say that a bipartite state $\rho_{\A\B}$ is incompatibility-preserving from Alice to Bob if every incompatible measurement assemblage $\{M_{a|x}\}$ on Alice generates a steerable assemblage
\begin{align}
\sigma_{a|x}:=\tr_{\A}\left[(M_{a|x}\otimes\mbb I)\rho_{\A\B}\right].
\end{align}
The reverse implication always holds: compatible measurements generate an unsteerable assemblage for every bipartite state.

From the steering--joint-measurability correspondence gives~~\cite{Quintino2014,Uola2014,Uola2015,Kiukas2017}, we thus have, for two-qubit states,

\begin{corollary}
\label{Cor:incompatibility-steering-two-qubit}
A two-qubit state is incompatibility preserving if and only if it is pure and entangled.
\end{corollary}

Although Corollary~\ref{Cor:incompatibility-steering-two-qubit} will follow from the general theorem below, the qubit proof is retained because it is elementary, uses the planar geometry developed above, and gives an explicit two-measurement witness.
We next derive the corresponding classification of bipartite states by associating a quantum channel with the steering-equivalent-observable map.

\begin{theorem}
\label{Thm:incompatibility-preserving-general-state}
Let $\rho_{\A\B}\in\mathcal D(\mc H_{\A}\otimes\mc H_{\B})$, where
$\dim\mc H_{\A}=\dim\mc H_{\B}=d$. Then $\rho_{\A\B}$ is incompatibility
preserving from Alice to Bob if and only if it is a pure state of Schmidt
rank $d$.
\end{theorem}

\begin{proof}
Suppose first that $\rho_{\A\B}$ is incompatibility preserving. Let
$K_{\B}:=\supp(\rho_{\B})$ and define
\begin{align}
\Phi_\rho(E):=
\rho_{\B}^{-1/2}
\tr_{\A}\left[(E\otimes\mbb I)\rho_{\A\B}\right]
\rho_{\B}^{-1/2},
\end{align}
where the inverse is taken on $K_{\B}$. Incompatibility preservation forces
$\Phi_\rho$ to be injective. Indeed, if
$\Phi_\rho(\Delta)=0$ for some nonzero Hermitian $\Delta$, then for
sufficiently small $\epsilon>0$,
$E=\mbb I/2+\epsilon\Delta$ is a nonscalar effect satisfying
$\Phi_\rho(E)=\mbb I/2$. Choosing a projection $P$ that does not commute
with $E$, the binary POVMs $\msf M^E$ and $\msf M^P$ are incompatible,
whereas $\msf M^{\Phi_\rho(E)}$ is compatible with every POVM, contradicting
incompatibility preservation.

Since $\Phi_\rho:\text{Herm}(\mc H_{\A})\to
\text{Herm}(K_{\B})$ is injective $d^2\leq(\dim K_{\B})^2$. 
But $\dim K_{\B}\leq d$, and therefore $\dim K_{\B}=d$. Hence $\rho_{\B}$
has full rank.

Define the state-associated channel through
\begin{align}
\mc N_\rho^\dagger(X):=
\rho_{\B}^{-1/2}
\tr_{\A}\left[(X^{\mathtt T}\otimes\mbb I)\rho_{\A\B}\right]
\rho_{\B}^{-1/2}.
\label{Eq:state-associated-channel}
\end{align}
This map is unital and completely positive. By the
steering--joint-measurability correspondence, together with the fact that
transposition preserves compatibility, $\mc N_\rho$ is incompatibility
preserving. Hence, by
Theorem~\ref{Thm:incompatibility-preserving-general-channel},
$\mc N_\rho$ is unitary. Its Choi matrix
\begin{align}
J(\mc N_\rho^\dagger)
=
(\mbb I\otimes\rho_{\B}^{-1/2})
\rho_{\A\B}
(\mbb I\otimes\rho_{\B}^{-1/2})
\end{align}
therefore has rank one. Since $\rho_{\B}$ is full rank, the local filter is
invertible, and hence $\rho_{\A\B}$ itself has rank one and is pure. Since $\text{rank}(\rho_{\B})=d$, $\rho_{\A\B}$ has Schmidt rank
$d$.

Conversely, suppose that $\rho_{\A\B}=\op{\psi}{\psi}$ is pure with Schmidt
rank $d$. We may write
\begin{align}
\ket{\psi}=\sum_{i=1}^d\sqrt{\lambda_i}\ket{i}_{\A}\ket{i}_{\B},
\qquad \lambda_i>0,
\end{align}
so that $\rho_{\B}=\sum_i\lambda_i\op{i}{i}$, and
\begin{align}
\rho_{\B}^{-1/2}\tr_{\A}\left[(X\otimes\mbb I)\rho_{\A\B}\right]\rho_{\B}^{-1/2}=X^{\mathtt T},
\end{align}
and hence $\Phi_\rho(X)=X^{\mathtt T}$. Equivalently, the state-associated
channel in Eq.~\eqref{Eq:state-associated-channel} satisfies
\begin{align}
\mc N_\rho^\dagger(X)=\Phi_\rho(X^{\mathtt T})=X.
\end{align}
Thus $\mc N_\rho$ is unitary. Since transposition
and unitary conjugation preserve compatibility in both directions, every
incompatible measurement assemblage on Alice gives an incompatible
steering-equivalent measurement assemblage on Bob. 
\end{proof}

\label{Sect:incompatibility-preserving-unequal}

The assumption of equal input and output dimensions in
Theorem~\ref{Thm:incompatibility-preserving-general-channel} is essential.
When the output dimension is larger than the input dimension,
nonunitary incompatibility-preserving maps can exist. A simple class of
such maps is given by the following construction.

\begin{proposition}
\label{Prop:unequal-dimensional-preserver}
Let $\Lambda:\mathcal L(\mc H_{\A})\to\mathcal L(\mc H_{\B})$ be a
completely positive unital map. Suppose that, up to a unitary on
$\mc H_{\B}$, it has the form
\begin{align}
\Lambda(X)=U\left[X\otimes\mbb I_m\oplus\mc E(X)\right]U^\dagger,
\label{Eq:general-preserver-sufficient}
\end{align}
where $\mc E$ is a completely positive unital map acting on the remaining
orthogonal subspace. Then $\Lambda$ is incompatibility preserving.
\end{proposition}

\begin{proof}
Suppose that the image of a measurement assemblage
$\msf M=\{M_{a|x}\}_{a,x}$ under $\Lambda$ is compatible. Compressing the
corresponding parent POVM onto any one of the copies of $\mc H_{\A}$
contained in the block $X\otimes\mbb I_m$ gives a parent POVM for the
original assemblage $\msf M$. Hence, compatibility of $\Lambda(\msf M)$
implies compatibility of $\msf M$. Therefore, every incompatible
assemblage remains incompatible under $\Lambda$, and $\Lambda$ is
incompatibility preserving.
\end{proof}

\noindent This construction motivates the following conjecture.

\begin{conjecture}
\label{Conj:unequal-dimensional-preserver}
Let $\Lambda:\mathcal L(\mc H_{\A})\to\mathcal L(\mc H_{\B})$ be a
completely positive unital map with $\dim\mc H_{\B}>\dim\mc H_{\A}$.
Then $\Lambda$ is incompatibility preserving if and only if, up to a
unitary and basis choices, it has the form
\begin{align}
\Lambda(X)=U\left[X\otimes\mbb I_m\oplus\mc E(X)\right]U^\dagger,
\label{Eq:conjectured-general-preserver}
\end{align}
where $\mc E$ is a completely positive unital map on the remaining
orthogonal subspace.
\end{conjecture}

We finally translate the dimensional restriction back to bipartite states.
If Alice is measured and Bob is steered, incompatibility preservation
requires the map $\Phi_\rho$ defined above to be injective, and therefore
\begin{align}
d_{\A}\leq\operatorname{rank}(\rho_{\B})\leq d_{\B}.
\end{align}
Hence, no incompatibility-preserving state exists when $d_{\B}<d_{\A}$.

When $d_{\B}>d_{\A}$, incompatibility preservation is no longer restricted
to pure states, in accordance with the additional possibilities suggested
by Conjecture~\ref{Conj:unequal-dimensional-preserver}. If preservation is
required in both steering directions, injectivity gives
$d_{\A}\leq d_{\B}$ and $d_{\B}\leq d_{\A}$, and hence $d_{\A}=d_{\B}$.
Theorem~\ref{Thm:incompatibility-preserving-general-channel} then implies
that the two-way incompatibility-preserving states are precisely the pure
states of full Schmidt rank.

\section{Conclusion and Outlook} 

In this work, we studied two opposite ways in which a quantum channel can act on incompatible measurement assemblages: destroying incompatibility and preserving it. Our main technical tool was the geometry of zonotopes and zonoids, which provides a particularly useful description of compatibility for dichotomic qubit measurements. Using this framework, we characterized the compatibility of noisy spin measurements in terms of centrally symmetric zonoids. As a direct application, we obtained an exact compatibility criterion for arbitrary finite families of planar noisy qubit PVMs, solving a conjecture discussed in~\cite{Andrejic-2020a,Grinko-2025a, Porto-2026a}. We then applied the same geometric framework to PVM-incompatibility-breaking qubit channels, obtaining a general sufficient criterion for nonunital channels that becomes necessary and sufficient in the unital case. At the incompatibility-breaking boundary, we further established a rigidity property that essentially fixes the parent POVM. 
We believe this result can be of independent interest in the study of quantum networks and network steering \cite{Jones-2021a}.

For this work, the boundary rigidity allowed us to study how incompatibility breaking behaves under parallel composition. We showed that the class of incompatibility-breaking channels is not closed under tensor products: although the qubit depolarizing channel $\Delta_{1/2}$ is incompatibility-breaking, $\Delta_{1/2}\otimes\Delta_{1/2}$ is not. Through the channel--state correspondence, this provides a simple two-copy activation of quantum steering.

At the opposite extreme, we found a much stronger form of rigidity for incompatibility preservation. For equal input and output dimensions, a channel preserves every incompatible measurement assemblage if and only if it is unitary. Correspondingly, for equal local dimensions, the bipartite states that preserve incompatibility from one party to the other are precisely the pure states of full Schmidt rank. 

Several questions remain open. Most notably, while the incompatibility-breaking condition is completely understood for unital qubit channels, an exact characterization for general nonunital qubit channels is still lacking. For incompatibility-preserving maps, when the output dimension exceeds the input dimension, nonunitary incompatibility-preserving maps can exist. The classification of such maps, including the conjectured block structure discussed above, remains open.
\blk

\begin{acknowledgments}
EC would like to thank Marco T\'{u}lio Quintino for suggesting the problem of incompatibility-preserving bipartite states.
EC acknowledges support from the National Science Foundation (NSF) Grant \# 119348.
A preliminary version of the manuscript was presented at the Bad Honnef WE-Heraeus-Seminar - 90 Years of Einstein‐Podolsky‐Rosen Paradox: Foundations and Applications. Y.Z. acknowledges the support of the Institute for Quantum Computing, the University of Waterloo, and the Perimeter Institute. Research at the Institute for Quantum Computing is supported by Innovation, Science and Economic Development Canada.
\end{acknowledgments}

\bibliography{compatibility}

\begin{thebibliography}{95}%
\makeatletter
\providecommand \@ifxundefined [1]{%
 \@ifx{#1\undefined}
}%
\providecommand \@ifnum [1]{%
 \ifnum #1\expandafter \@firstoftwo
 \else \expandafter \@secondoftwo
 \fi
}%
\providecommand \@ifx [1]{%
 \ifx #1\expandafter \@firstoftwo
 \else \expandafter \@secondoftwo
 \fi
}%
\providecommand \natexlab [1]{#1}%
\providecommand \enquote  [1]{``#1''}%
\providecommand \bibnamefont  [1]{#1}%
\providecommand \bibfnamefont [1]{#1}%
\providecommand \citenamefont [1]{#1}%
\providecommand \href@noop [0]{\@secondoftwo}%
\providecommand \href [0]{\begingroup \@sanitize@url \@href}%
\providecommand \@href[1]{\@@startlink{#1}\@@href}%
\providecommand \@@href[1]{\endgroup#1\@@endlink}%
\providecommand \@sanitize@url [0]{\catcode `\\12\catcode `\$12\catcode `\&12\catcode `\#12\catcode `\^12\catcode `\_12\catcode `\%12\relax}%
\providecommand \@@startlink[1]{}%
\providecommand \@@endlink[0]{}%
\providecommand \url  [0]{\begingroup\@sanitize@url \@url }%
\providecommand \@url [1]{\endgroup\@href {#1}{\urlprefix }}%
\providecommand \urlprefix  [0]{URL }%
\providecommand \Eprint [0]{\href }%
\providecommand \doibase [0]{https://doi.org/}%
\providecommand \selectlanguage [0]{\@gobble}%
\providecommand \bibinfo  [0]{\@secondoftwo}%
\providecommand \bibfield  [0]{\@secondoftwo}%
\providecommand \translation [1]{[#1]}%
\providecommand \BibitemOpen [0]{}%
\providecommand \bibitemStop [0]{}%
\providecommand \bibitemNoStop [0]{.\EOS\space}%
\providecommand \EOS [0]{\spacefactor3000\relax}%
\providecommand \BibitemShut  [1]{\csname bibitem#1\endcsname}%
\let\auto@bib@innerbib\@empty
\bibitem [{\citenamefont {Horodecki}\ \emph {et~al.}(2003)\citenamefont {Horodecki}, \citenamefont {Shor},\ and\ \citenamefont {Ruskai}}]{Horodecki-2003a}%
  \BibitemOpen
  \bibfield  {author} {\bibinfo {author} {\bibfnamefont {M.}~\bibnamefont {Horodecki}}, \bibinfo {author} {\bibfnamefont {P.~W.}\ \bibnamefont {Shor}},\ and\ \bibinfo {author} {\bibfnamefont {M.~B.}\ \bibnamefont {Ruskai}},\ }\bibfield  {title} {\bibinfo {title} {Entanglement breaking channels},\ }\href {https://doi.org/10.1142/S0129055X03001709} {\bibfield  {journal} {\bibinfo  {journal} {Reviews in Mathematical Physics}\ }\textbf {\bibinfo {volume} {15}},\ \bibinfo {pages} {629} (\bibinfo {year} {2003})}\BibitemShut {NoStop}%
\bibitem [{\citenamefont {Chruściński}\ and\ \citenamefont {Kossakowski}(2006)}]{Chruciski-2006a}%
  \BibitemOpen
  \bibfield  {author} {\bibinfo {author} {\bibfnamefont {D.}~\bibnamefont {Chruściński}}\ and\ \bibinfo {author} {\bibfnamefont {A.}~\bibnamefont {Kossakowski}},\ }\bibfield  {title} {\bibinfo {title} {On partially entanglement breaking channels},\ }\href {https://doi.org/10.1007/s11080-006-7264-7} {\bibfield  {journal} {\bibinfo  {journal} {Open Systems \& amp; Information Dynamics}\ }\textbf {\bibinfo {volume} {13}},\ \bibinfo {pages} {17–26} (\bibinfo {year} {2006})}\BibitemShut {NoStop}%
\bibitem [{\citenamefont {Chen}\ and\ \citenamefont {Chitambar}(2020)}]{Chen-2020a}%
  \BibitemOpen
  \bibfield  {author} {\bibinfo {author} {\bibfnamefont {S.}~\bibnamefont {Chen}}\ and\ \bibinfo {author} {\bibfnamefont {E.}~\bibnamefont {Chitambar}},\ }\bibfield  {title} {\bibinfo {title} {Entanglement-breaking superchannels},\ }\href {https://doi.org/10.22331/q-2020-07-16-299} {\bibfield  {journal} {\bibinfo  {journal} {Quantum}\ }\textbf {\bibinfo {volume} {4}},\ \bibinfo {pages} {299} (\bibinfo {year} {2020})}\BibitemShut {NoStop}%
\bibitem [{\citenamefont {Liu}\ \emph {et~al.}(2017)\citenamefont {Liu}, \citenamefont {Hu},\ and\ \citenamefont {Lloyd}}]{Liu-2017a}%
  \BibitemOpen
  \bibfield  {author} {\bibinfo {author} {\bibfnamefont {Z.-W.}\ \bibnamefont {Liu}}, \bibinfo {author} {\bibfnamefont {X.}~\bibnamefont {Hu}},\ and\ \bibinfo {author} {\bibfnamefont {S.}~\bibnamefont {Lloyd}},\ }\bibfield  {title} {\bibinfo {title} {Resource destroying maps},\ }\href {https://doi.org/10.1103/PhysRevLett.118.060502} {\bibfield  {journal} {\bibinfo  {journal} {Phys. Rev. Lett.}\ }\textbf {\bibinfo {volume} {118}},\ \bibinfo {pages} {060502} (\bibinfo {year} {2017})}\BibitemShut {NoStop}%
\bibitem [{\citenamefont {Gour}(2017)}]{Gour-2017a}%
  \BibitemOpen
  \bibfield  {author} {\bibinfo {author} {\bibfnamefont {G.}~\bibnamefont {Gour}},\ }\bibfield  {title} {\bibinfo {title} {Quantum resource theories in the single-shot regime},\ }\href {https://doi.org/10.1103/PhysRevA.95.062314} {\bibfield  {journal} {\bibinfo  {journal} {Phys. Rev. A}\ }\textbf {\bibinfo {volume} {95}},\ \bibinfo {pages} {062314} (\bibinfo {year} {2017})}\BibitemShut {NoStop}%
\bibitem [{\citenamefont {Bu}\ \emph {et~al.}(2016)\citenamefont {Bu}, \citenamefont {Swati}, \citenamefont {Singh},\ and\ \citenamefont {Wu}}]{Bu-2016a}%
  \BibitemOpen
  \bibfield  {author} {\bibinfo {author} {\bibfnamefont {K.}~\bibnamefont {Bu}}, \bibinfo {author} {\bibnamefont {Swati}}, \bibinfo {author} {\bibfnamefont {U.}~\bibnamefont {Singh}},\ and\ \bibinfo {author} {\bibfnamefont {J.}~\bibnamefont {Wu}},\ }\bibfield  {title} {\bibinfo {title} {Coherence-breaking channels and coherence sudden death},\ }\href {https://doi.org/10.1103/PhysRevA.94.052335} {\bibfield  {journal} {\bibinfo  {journal} {Phys. Rev. A}\ }\textbf {\bibinfo {volume} {94}},\ \bibinfo {pages} {052335} (\bibinfo {year} {2016})}\BibitemShut {NoStop}%
\bibitem [{\citenamefont {Kamin}\ \emph {et~al.}(2020)\citenamefont {Kamin}, \citenamefont {Tabesh}, \citenamefont {Salimi},\ and\ \citenamefont {Kheirandish}}]{Kamin-2020a}%
  \BibitemOpen
  \bibfield  {author} {\bibinfo {author} {\bibfnamefont {F.~H.}\ \bibnamefont {Kamin}}, \bibinfo {author} {\bibfnamefont {F.~T.}\ \bibnamefont {Tabesh}}, \bibinfo {author} {\bibfnamefont {S.}~\bibnamefont {Salimi}},\ and\ \bibinfo {author} {\bibfnamefont {F.}~\bibnamefont {Kheirandish}},\ }\bibfield  {title} {\bibinfo {title} {The resource theory of coherence for quantum channels},\ }\bibfield  {journal} {\bibinfo  {journal} {Quantum Information Processing}\ }\textbf {\bibinfo {volume} {19}},\ \href {https://doi.org/10.1007/s11128-020-02702-9} {10.1007/s11128-020-02702-9} (\bibinfo {year} {2020})\BibitemShut {NoStop}%
\bibitem [{\citenamefont {Ivan}\ \emph {et~al.}(2013)\citenamefont {Ivan}, \citenamefont {Sabapathy},\ and\ \citenamefont {Simon}}]{Ivan-2013a}%
  \BibitemOpen
  \bibfield  {author} {\bibinfo {author} {\bibfnamefont {J.~S.}\ \bibnamefont {Ivan}}, \bibinfo {author} {\bibfnamefont {K.~K.}\ \bibnamefont {Sabapathy}},\ and\ \bibinfo {author} {\bibfnamefont {R.}~\bibnamefont {Simon}},\ }\bibfield  {title} {\bibinfo {title} {Nonclassicality breaking is the same as entanglement breaking for bosonic gaussian channels},\ }\href {https://doi.org/10.1103/PhysRevA.88.032302} {\bibfield  {journal} {\bibinfo  {journal} {Phys. Rev. A}\ }\textbf {\bibinfo {volume} {88}},\ \bibinfo {pages} {032302} (\bibinfo {year} {2013})}\BibitemShut {NoStop}%
\bibitem [{\citenamefont {Meznaric}\ \emph {et~al.}(2013)\citenamefont {Meznaric}, \citenamefont {Clark},\ and\ \citenamefont {Datta}}]{Meznaric-2013a}%
  \BibitemOpen
  \bibfield  {author} {\bibinfo {author} {\bibfnamefont {S.}~\bibnamefont {Meznaric}}, \bibinfo {author} {\bibfnamefont {S.~R.}\ \bibnamefont {Clark}},\ and\ \bibinfo {author} {\bibfnamefont {A.}~\bibnamefont {Datta}},\ }\bibfield  {title} {\bibinfo {title} {Quantifying the nonclassicality of operations},\ }\href {https://doi.org/10.1103/PhysRevLett.110.070502} {\bibfield  {journal} {\bibinfo  {journal} {Phys. Rev. Lett.}\ }\textbf {\bibinfo {volume} {110}},\ \bibinfo {pages} {070502} (\bibinfo {year} {2013})}\BibitemShut {NoStop}%
\bibitem [{\citenamefont {Pal}\ and\ \citenamefont {Ghosh}(2015)}]{Pal-2015a}%
  \BibitemOpen
  \bibfield  {author} {\bibinfo {author} {\bibfnamefont {R.}~\bibnamefont {Pal}}\ and\ \bibinfo {author} {\bibfnamefont {S.}~\bibnamefont {Ghosh}},\ }\bibfield  {title} {\bibinfo {title} {Non-locality breaking qubit channels: the case for chsh inequality},\ }\href {https://doi.org/10.1088/1751-8113/48/15/155302} {\bibfield  {journal} {\bibinfo  {journal} {Journal of Physics A: Mathematical and Theoretical}\ }\textbf {\bibinfo {volume} {48}},\ \bibinfo {pages} {155302} (\bibinfo {year} {2015})}\BibitemShut {NoStop}%
\bibitem [{\citenamefont {Zhang}\ \emph {et~al.}(2020)\citenamefont {Zhang}, \citenamefont {Bravo}, \citenamefont {Lorenz},\ and\ \citenamefont {Chitambar}}]{Zhang-2020a}%
  \BibitemOpen
  \bibfield  {author} {\bibinfo {author} {\bibfnamefont {Y.}~\bibnamefont {Zhang}}, \bibinfo {author} {\bibfnamefont {R.~A.}\ \bibnamefont {Bravo}}, \bibinfo {author} {\bibfnamefont {V.~O.}\ \bibnamefont {Lorenz}},\ and\ \bibinfo {author} {\bibfnamefont {E.}~\bibnamefont {Chitambar}},\ }\bibfield  {title} {\bibinfo {title} {Channel activation of chsh nonlocality},\ }\href {https://doi.org/10.1088/1367-2630/ab7bef} {\bibfield  {journal} {\bibinfo  {journal} {New Journal of Physics}\ }\textbf {\bibinfo {volume} {22}},\ \bibinfo {pages} {043003} (\bibinfo {year} {2020})}\BibitemShut {NoStop}%
\bibitem [{\citenamefont {Mukherjee}\ and\ \citenamefont {Ganguly}(2025)}]{Mukherjee-2025a}%
  \BibitemOpen
  \bibfield  {author} {\bibinfo {author} {\bibfnamefont {K.}~\bibnamefont {Mukherjee}}\ and\ \bibinfo {author} {\bibfnamefont {N.}~\bibnamefont {Ganguly}},\ }\href@noop {} {\bibinfo {title} {Network nonlocality breaking channels}} (\bibinfo {year} {2025}),\ \Eprint {https://arxiv.org/abs/arXiv:2510.26417} {arXiv:2510.26417} \BibitemShut {NoStop}%
\bibitem [{\citenamefont {Heinosaari}\ \emph {et~al.}(2015{\natexlab{a}})\citenamefont {Heinosaari}, \citenamefont {Kiukas}, \citenamefont {Reitzner},\ and\ \citenamefont {Schultz}}]{Heinosaari-2015a}%
  \BibitemOpen
  \bibfield  {author} {\bibinfo {author} {\bibfnamefont {T.}~\bibnamefont {Heinosaari}}, \bibinfo {author} {\bibfnamefont {J.}~\bibnamefont {Kiukas}}, \bibinfo {author} {\bibfnamefont {D.}~\bibnamefont {Reitzner}},\ and\ \bibinfo {author} {\bibfnamefont {J.}~\bibnamefont {Schultz}},\ }\bibfield  {title} {\bibinfo {title} {Incompatibility breaking quantum channels},\ }\href {https://doi.org/10.1088/1751-8113/48/43/435301} {\bibfield  {journal} {\bibinfo  {journal} {Journal of Physics A: Mathematical and Theoretical}\ }\textbf {\bibinfo {volume} {48}},\ \bibinfo {pages} {435301} (\bibinfo {year} {2015}{\natexlab{a}})}\BibitemShut {NoStop}%
\bibitem [{\citenamefont {Kumari}\ \emph {et~al.}(2023)\citenamefont {Kumari}, \citenamefont {Naikoo}, \citenamefont {Ghosh},\ and\ \citenamefont {Pan}}]{Kumari-2023a}%
  \BibitemOpen
  \bibfield  {author} {\bibinfo {author} {\bibfnamefont {S.}~\bibnamefont {Kumari}}, \bibinfo {author} {\bibfnamefont {J.}~\bibnamefont {Naikoo}}, \bibinfo {author} {\bibfnamefont {S.}~\bibnamefont {Ghosh}},\ and\ \bibinfo {author} {\bibfnamefont {A.~K.}\ \bibnamefont {Pan}},\ }\bibfield  {title} {\bibinfo {title} {Interplay of nonlocality and incompatibility breaking qubit channels},\ }\href {https://doi.org/10.1103/PhysRevA.107.022201} {\bibfield  {journal} {\bibinfo  {journal} {Phys. Rev. A}\ }\textbf {\bibinfo {volume} {107}},\ \bibinfo {pages} {022201} (\bibinfo {year} {2023})}\BibitemShut {NoStop}%
\bibitem [{\citenamefont {Heinosaari}\ \emph {et~al.}(2016)\citenamefont {Heinosaari}, \citenamefont {Miyadera},\ and\ \citenamefont {Ziman}}]{Heinosaari-2016a}%
  \BibitemOpen
  \bibfield  {author} {\bibinfo {author} {\bibfnamefont {T.}~\bibnamefont {Heinosaari}}, \bibinfo {author} {\bibfnamefont {T.}~\bibnamefont {Miyadera}},\ and\ \bibinfo {author} {\bibfnamefont {M.}~\bibnamefont {Ziman}},\ }\bibfield  {title} {\bibinfo {title} {An invitation to quantum incompatibility},\ }\href {https://doi.org/10.1088/1751-8113/49/12/123001} {\bibfield  {journal} {\bibinfo  {journal} {Journal of Physics A: Mathematical and Theoretical}\ }\textbf {\bibinfo {volume} {49}},\ \bibinfo {pages} {123001} (\bibinfo {year} {2016})}\BibitemShut {NoStop}%
\bibitem [{\citenamefont {G\"uhne}\ \emph {et~al.}(2023)\citenamefont {G\"uhne}, \citenamefont {Haapasalo}, \citenamefont {Kraft}, \citenamefont {Pellonp\"a\"a},\ and\ \citenamefont {Uola}}]{Guhne-2023a}%
  \BibitemOpen
  \bibfield  {author} {\bibinfo {author} {\bibfnamefont {O.}~\bibnamefont {G\"uhne}}, \bibinfo {author} {\bibfnamefont {E.}~\bibnamefont {Haapasalo}}, \bibinfo {author} {\bibfnamefont {T.}~\bibnamefont {Kraft}}, \bibinfo {author} {\bibfnamefont {J.-P.}\ \bibnamefont {Pellonp\"a\"a}},\ and\ \bibinfo {author} {\bibfnamefont {R.}~\bibnamefont {Uola}},\ }\bibfield  {title} {\bibinfo {title} {Colloquium: Incompatible measurements in quantum information science},\ }\href {https://doi.org/10.1103/RevModPhys.95.011003} {\bibfield  {journal} {\bibinfo  {journal} {Rev. Mod. Phys.}\ }\textbf {\bibinfo {volume} {95}},\ \bibinfo {pages} {011003} (\bibinfo {year} {2023})}\BibitemShut {NoStop}%
\bibitem [{\citenamefont {Skrzypczyk}\ and\ \citenamefont {Linden}(2019)}]{Skrzypczyk-2019a}%
  \BibitemOpen
  \bibfield  {author} {\bibinfo {author} {\bibfnamefont {P.}~\bibnamefont {Skrzypczyk}}\ and\ \bibinfo {author} {\bibfnamefont {N.}~\bibnamefont {Linden}},\ }\bibfield  {title} {\bibinfo {title} {Robustness of measurement, discrimination games, and accessible information},\ }\href {https://doi.org/10.1103/PhysRevLett.122.140403} {\bibfield  {journal} {\bibinfo  {journal} {Phys. Rev. Lett.}\ }\textbf {\bibinfo {volume} {122}},\ \bibinfo {pages} {140403} (\bibinfo {year} {2019})}\BibitemShut {NoStop}%
\bibitem [{\citenamefont {Designolle}\ \emph {et~al.}(2019)\citenamefont {Designolle}, \citenamefont {Farkas},\ and\ \citenamefont {Kaniewski}}]{Designolle-2019a}%
  \BibitemOpen
  \bibfield  {author} {\bibinfo {author} {\bibfnamefont {S.}~\bibnamefont {Designolle}}, \bibinfo {author} {\bibfnamefont {M.}~\bibnamefont {Farkas}},\ and\ \bibinfo {author} {\bibfnamefont {J.}~\bibnamefont {Kaniewski}},\ }\bibfield  {title} {\bibinfo {title} {Incompatibility robustness of quantum measurements: a unified framework},\ }\href {https://doi.org/10.1088/1367-2630/ab5020} {\bibfield  {journal} {\bibinfo  {journal} {New Journal of Physics}\ }\textbf {\bibinfo {volume} {21}},\ \bibinfo {pages} {113053} (\bibinfo {year} {2019})}\BibitemShut {NoStop}%
\bibitem [{\citenamefont {Buscemi}\ \emph {et~al.}(2020)\citenamefont {Buscemi}, \citenamefont {Chitambar},\ and\ \citenamefont {Zhou}}]{Buscemi-2020a}%
  \BibitemOpen
  \bibfield  {author} {\bibinfo {author} {\bibfnamefont {F.}~\bibnamefont {Buscemi}}, \bibinfo {author} {\bibfnamefont {E.}~\bibnamefont {Chitambar}},\ and\ \bibinfo {author} {\bibfnamefont {W.}~\bibnamefont {Zhou}},\ }\bibfield  {title} {\bibinfo {title} {Complete resource theory of quantum incompatibility as quantum programmability},\ }\href {https://doi.org/10.1103/PhysRevLett.124.120401} {\bibfield  {journal} {\bibinfo  {journal} {Phys. Rev. Lett.}\ }\textbf {\bibinfo {volume} {124}},\ \bibinfo {pages} {120401} (\bibinfo {year} {2020})}\BibitemShut {NoStop}%
\bibitem [{\citenamefont {Ji}\ and\ \citenamefont {Chitambar}(2024)}]{Ji-2024a}%
  \BibitemOpen
  \bibfield  {author} {\bibinfo {author} {\bibfnamefont {K.}~\bibnamefont {Ji}}\ and\ \bibinfo {author} {\bibfnamefont {E.}~\bibnamefont {Chitambar}},\ }\bibfield  {title} {\bibinfo {title} {Incompatibility as a resource for programmable quantum instruments},\ }\href {https://doi.org/10.1103/PRXQuantum.5.010340} {\bibfield  {journal} {\bibinfo  {journal} {PRX Quantum}\ }\textbf {\bibinfo {volume} {5}},\ \bibinfo {pages} {010340} (\bibinfo {year} {2024})}\BibitemShut {NoStop}%
\bibitem [{\citenamefont {Mitra}\ and\ \citenamefont {Farkas}(2022)}]{Mitra-2022a}%
  \BibitemOpen
  \bibfield  {author} {\bibinfo {author} {\bibfnamefont {A.}~\bibnamefont {Mitra}}\ and\ \bibinfo {author} {\bibfnamefont {M.}~\bibnamefont {Farkas}},\ }\bibfield  {title} {\bibinfo {title} {Compatibility of quantum instruments},\ }\href {https://doi.org/10.1103/PhysRevA.105.052202} {\bibfield  {journal} {\bibinfo  {journal} {Phys. Rev. A}\ }\textbf {\bibinfo {volume} {105}},\ \bibinfo {pages} {052202} (\bibinfo {year} {2022})}\BibitemShut {NoStop}%
\bibitem [{\citenamefont {Bluhm}\ \emph {et~al.}(2025)\citenamefont {Bluhm}, \citenamefont {Lepp\"{a}j\"{a}rvi},\ and\ \citenamefont {Nechita}}]{Bluhm-2025a}%
  \BibitemOpen
  \bibfield  {author} {\bibinfo {author} {\bibfnamefont {A.}~\bibnamefont {Bluhm}}, \bibinfo {author} {\bibfnamefont {L.}~\bibnamefont {Lepp\"{a}j\"{a}rvi}},\ and\ \bibinfo {author} {\bibfnamefont {I.}~\bibnamefont {Nechita}},\ }\bibfield  {title} {\bibinfo {title} {On the simulation of quantum multimeters},\ }\href {https://doi.org/10.22331/q-2025-01-27-1608} {\bibfield  {journal} {\bibinfo  {journal} {Quantum}\ }\textbf {\bibinfo {volume} {9}},\ \bibinfo {pages} {1608} (\bibinfo {year} {2025})}\BibitemShut {NoStop}%
\bibitem [{\citenamefont {Terhal}\ \emph {et~al.}(2003)\citenamefont {Terhal}, \citenamefont {Doherty},\ and\ \citenamefont {Schwab}}]{Terhal-2003a}%
  \BibitemOpen
  \bibfield  {author} {\bibinfo {author} {\bibfnamefont {B.~M.}\ \bibnamefont {Terhal}}, \bibinfo {author} {\bibfnamefont {A.~C.}\ \bibnamefont {Doherty}},\ and\ \bibinfo {author} {\bibfnamefont {D.}~\bibnamefont {Schwab}},\ }\bibfield  {title} {\bibinfo {title} {Symmetric extensions of quantum states and local hidden variable theories},\ }\href {https://doi.org/10.1103/PhysRevLett.90.157903} {\bibfield  {journal} {\bibinfo  {journal} {Phys. Rev. Lett.}\ }\textbf {\bibinfo {volume} {90}},\ \bibinfo {pages} {157903} (\bibinfo {year} {2003})}\BibitemShut {NoStop}%
\bibitem [{\citenamefont {Wolf}\ \emph {et~al.}(2009)\citenamefont {Wolf}, \citenamefont {Perez-Garcia},\ and\ \citenamefont {Fernandez}}]{Wolf-2009a}%
  \BibitemOpen
  \bibfield  {author} {\bibinfo {author} {\bibfnamefont {M.~M.}\ \bibnamefont {Wolf}}, \bibinfo {author} {\bibfnamefont {D.}~\bibnamefont {Perez-Garcia}},\ and\ \bibinfo {author} {\bibfnamefont {C.}~\bibnamefont {Fernandez}},\ }\bibfield  {title} {\bibinfo {title} {Measurements incompatible in quantum theory cannot be measured jointly in any other no-signaling theory},\ }\href {https://doi.org/10.1103/PhysRevLett.103.230402} {\bibfield  {journal} {\bibinfo  {journal} {Phys. Rev. Lett.}\ }\textbf {\bibinfo {volume} {103}},\ \bibinfo {pages} {230402} (\bibinfo {year} {2009})}\BibitemShut {NoStop}%
\bibitem [{\citenamefont {Pl\'avala}\ \emph {et~al.}(2025)\citenamefont {Pl\'avala}, \citenamefont {G\"uhne},\ and\ \citenamefont {Quintino}}]{Plavala-2025a}%
  \BibitemOpen
  \bibfield  {author} {\bibinfo {author} {\bibfnamefont {M.}~\bibnamefont {Pl\'avala}}, \bibinfo {author} {\bibfnamefont {O.}~\bibnamefont {G\"uhne}},\ and\ \bibinfo {author} {\bibfnamefont {M.~T.}\ \bibnamefont {Quintino}},\ }\bibfield  {title} {\bibinfo {title} {All incompatible measurements on qubits lead to multiparticle bell nonlocality},\ }\href {https://doi.org/10.1103/PhysRevLett.134.200201} {\bibfield  {journal} {\bibinfo  {journal} {Phys. Rev. Lett.}\ }\textbf {\bibinfo {volume} {134}},\ \bibinfo {pages} {200201} (\bibinfo {year} {2025})}\BibitemShut {NoStop}%
\bibitem [{\citenamefont {Quintino}\ \emph {et~al.}(2016{\natexlab{a}})\citenamefont {Quintino}, \citenamefont {Bowles}, \citenamefont {Hirsch},\ and\ \citenamefont {Brunner}}]{Quintino-2016a}%
  \BibitemOpen
  \bibfield  {author} {\bibinfo {author} {\bibfnamefont {M.~T.}\ \bibnamefont {Quintino}}, \bibinfo {author} {\bibfnamefont {J.}~\bibnamefont {Bowles}}, \bibinfo {author} {\bibfnamefont {F.}~\bibnamefont {Hirsch}},\ and\ \bibinfo {author} {\bibfnamefont {N.}~\bibnamefont {Brunner}},\ }\bibfield  {title} {\bibinfo {title} {Incompatible quantum measurements admitting a local-hidden-variable model},\ }\href {https://doi.org/10.1103/PhysRevA.93.052115} {\bibfield  {journal} {\bibinfo  {journal} {Phys. Rev. A}\ }\textbf {\bibinfo {volume} {93}},\ \bibinfo {pages} {052115} (\bibinfo {year} {2016}{\natexlab{a}})}\BibitemShut {NoStop}%
\bibitem [{\citenamefont {Bene}\ and\ \citenamefont {Vértesi}(2018)}]{Bene-2018a}%
  \BibitemOpen
  \bibfield  {author} {\bibinfo {author} {\bibfnamefont {E.}~\bibnamefont {Bene}}\ and\ \bibinfo {author} {\bibfnamefont {T.}~\bibnamefont {Vértesi}},\ }\bibfield  {title} {\bibinfo {title} {Measurement incompatibility does not give rise to bell violation in general},\ }\href {https://doi.org/10.1088/1367-2630/aa9ca3} {\bibfield  {journal} {\bibinfo  {journal} {New Journal of Physics}\ }\textbf {\bibinfo {volume} {20}},\ \bibinfo {pages} {013021} (\bibinfo {year} {2018})}\BibitemShut {NoStop}%
\bibitem [{\citenamefont {Hirsch}\ \emph {et~al.}(2018)\citenamefont {Hirsch}, \citenamefont {Quintino},\ and\ \citenamefont {Brunner}}]{Hirsch-2018a}%
  \BibitemOpen
  \bibfield  {author} {\bibinfo {author} {\bibfnamefont {F.}~\bibnamefont {Hirsch}}, \bibinfo {author} {\bibfnamefont {M.~T.}\ \bibnamefont {Quintino}},\ and\ \bibinfo {author} {\bibfnamefont {N.}~\bibnamefont {Brunner}},\ }\bibfield  {title} {\bibinfo {title} {Quantum measurement incompatibility does not imply bell nonlocality},\ }\href {https://doi.org/10.1103/PhysRevA.97.012129} {\bibfield  {journal} {\bibinfo  {journal} {Phys. Rev. A}\ }\textbf {\bibinfo {volume} {97}},\ \bibinfo {pages} {012129} (\bibinfo {year} {2018})}\BibitemShut {NoStop}%
\bibitem [{\citenamefont {Spekkens}(2005)}]{Spekkens2005}%
  \BibitemOpen
  \bibfield  {author} {\bibinfo {author} {\bibfnamefont {R.~W.}\ \bibnamefont {Spekkens}},\ }\bibfield  {title} {\bibinfo {title} {Contextuality for preparations, transformations, and unsharp measurements},\ }\href {https://doi.org/10.1103/PhysRevA.71.052108} {\bibfield  {journal} {\bibinfo  {journal} {Phys. Rev. A}\ }\textbf {\bibinfo {volume} {71}},\ \bibinfo {pages} {052108} (\bibinfo {year} {2005})}\BibitemShut {NoStop}%
\bibitem [{\citenamefont {Tavakoli}\ and\ \citenamefont {Uola}(2020)}]{TavakoliUola2020}%
  \BibitemOpen
  \bibfield  {author} {\bibinfo {author} {\bibfnamefont {A.}~\bibnamefont {Tavakoli}}\ and\ \bibinfo {author} {\bibfnamefont {R.}~\bibnamefont {Uola}},\ }\bibfield  {title} {\bibinfo {title} {Measurement incompatibility and steering are necessary and sufficient for operational contextuality},\ }\href {https://doi.org/10.1103/PhysRevResearch.2.013011} {\bibfield  {journal} {\bibinfo  {journal} {Phys. Rev. Research}\ }\textbf {\bibinfo {volume} {2}},\ \bibinfo {pages} {013011} (\bibinfo {year} {2020})}\BibitemShut {NoStop}%
\bibitem [{\citenamefont {Selby}\ \emph {et~al.}(2023)\citenamefont {Selby}, \citenamefont {Schmid}, \citenamefont {Wolfe}, \citenamefont {Sainz}, \citenamefont {Kunjwal},\ and\ \citenamefont {Spekkens}}]{Selby2023}%
  \BibitemOpen
  \bibfield  {author} {\bibinfo {author} {\bibfnamefont {J.~H.}\ \bibnamefont {Selby}}, \bibinfo {author} {\bibfnamefont {D.}~\bibnamefont {Schmid}}, \bibinfo {author} {\bibfnamefont {E.}~\bibnamefont {Wolfe}}, \bibinfo {author} {\bibfnamefont {A.~B.}\ \bibnamefont {Sainz}}, \bibinfo {author} {\bibfnamefont {R.}~\bibnamefont {Kunjwal}},\ and\ \bibinfo {author} {\bibfnamefont {R.~W.}\ \bibnamefont {Spekkens}},\ }\bibfield  {title} {\bibinfo {title} {Contextuality without incompatibility},\ }\href {https://doi.org/10.1103/PhysRevLett.130.230201} {\bibfield  {journal} {\bibinfo  {journal} {Phys. Rev. Lett.}\ }\textbf {\bibinfo {volume} {130}},\ \bibinfo {pages} {230201} (\bibinfo {year} {2023})}\BibitemShut {NoStop}%
\bibitem [{\citenamefont {Zhang}\ \emph {et~al.}(2026{\natexlab{a}})\citenamefont {Zhang}, \citenamefont {Schmid}, \citenamefont {Y{\=i}ng},\ and\ \citenamefont {Spekkens}}]{ZhangSchmidYingSpekkens2026}%
  \BibitemOpen
  \bibfield  {author} {\bibinfo {author} {\bibfnamefont {Y.}~\bibnamefont {Zhang}}, \bibinfo {author} {\bibfnamefont {D.}~\bibnamefont {Schmid}}, \bibinfo {author} {\bibfnamefont {Y.}~\bibnamefont {Y{\=i}ng}},\ and\ \bibinfo {author} {\bibfnamefont {R.~W.}\ \bibnamefont {Spekkens}},\ }\bibfield  {title} {\bibinfo {title} {Reassessing the boundary between classical and nonclassical for individual quantum processes},\ }\href {https://doi.org/10.1103/vqfz-wzjg} {\bibfield  {journal} {\bibinfo  {journal} {Phys. Rev. X}\ }\textbf {\bibinfo {volume} {16}},\ \bibinfo {pages} {021050} (\bibinfo {year} {2026}{\natexlab{a}})}\BibitemShut {NoStop}%
\bibitem [{\citenamefont {Zhang}\ \emph {et~al.}(2026{\natexlab{b}})\citenamefont {Zhang}, \citenamefont {Yīng},\ and\ \citenamefont {Schmid}}]{ZhangYingSchmid2025}%
  \BibitemOpen
  \bibfield  {author} {\bibinfo {author} {\bibfnamefont {Y.}~\bibnamefont {Zhang}}, \bibinfo {author} {\bibfnamefont {Y.}~\bibnamefont {Yīng}},\ and\ \bibinfo {author} {\bibfnamefont {D.}~\bibnamefont {Schmid}},\ }\bibfield  {title} {\bibinfo {title} {Quantifiers and witnesses for the nonclassicality of measurements and of states},\ }\href {https://doi.org/10.22331/q-2026-07-30-2180} {\bibfield  {journal} {\bibinfo  {journal} {{Quantum}}\ }\textbf {\bibinfo {volume} {10}},\ \bibinfo {pages} {2180} (\bibinfo {year} {2026}{\natexlab{b}})}\BibitemShut {NoStop}%
\bibitem [{\citenamefont {Quintino}\ \emph {et~al.}(2014)\citenamefont {Quintino}, \citenamefont {V\'ertesi},\ and\ \citenamefont {Brunner}}]{Quintino2014}%
  \BibitemOpen
  \bibfield  {author} {\bibinfo {author} {\bibfnamefont {M.~T.}\ \bibnamefont {Quintino}}, \bibinfo {author} {\bibfnamefont {T.}~\bibnamefont {V\'ertesi}},\ and\ \bibinfo {author} {\bibfnamefont {N.}~\bibnamefont {Brunner}},\ }\bibfield  {title} {\bibinfo {title} {Joint measurability, einstein-podolsky-rosen steering, and bell nonlocality},\ }\href {https://doi.org/10.1103/PhysRevLett.113.160402} {\bibfield  {journal} {\bibinfo  {journal} {Phys. Rev. Lett.}\ }\textbf {\bibinfo {volume} {113}},\ \bibinfo {pages} {160402} (\bibinfo {year} {2014})}\BibitemShut {NoStop}%
\bibitem [{\citenamefont {Uola}\ \emph {et~al.}(2014)\citenamefont {Uola}, \citenamefont {Moroder},\ and\ \citenamefont {G\"uhne}}]{Uola2014}%
  \BibitemOpen
  \bibfield  {author} {\bibinfo {author} {\bibfnamefont {R.}~\bibnamefont {Uola}}, \bibinfo {author} {\bibfnamefont {T.}~\bibnamefont {Moroder}},\ and\ \bibinfo {author} {\bibfnamefont {O.}~\bibnamefont {G\"uhne}},\ }\bibfield  {title} {\bibinfo {title} {Joint measurability of generalized measurements implies classicality},\ }\href {https://doi.org/10.1103/PhysRevLett.113.160403} {\bibfield  {journal} {\bibinfo  {journal} {Phys. Rev. Lett.}\ }\textbf {\bibinfo {volume} {113}},\ \bibinfo {pages} {160403} (\bibinfo {year} {2014})}\BibitemShut {NoStop}%
\bibitem [{\citenamefont {Wiseman}\ \emph {et~al.}(2007)\citenamefont {Wiseman}, \citenamefont {Jones},\ and\ \citenamefont {Doherty}}]{Wiseman-2007a}%
  \BibitemOpen
  \bibfield  {author} {\bibinfo {author} {\bibfnamefont {H.~M.}\ \bibnamefont {Wiseman}}, \bibinfo {author} {\bibfnamefont {S.~J.}\ \bibnamefont {Jones}},\ and\ \bibinfo {author} {\bibfnamefont {A.~C.}\ \bibnamefont {Doherty}},\ }\bibfield  {title} {\bibinfo {title} {Steering, entanglement, nonlocality, and the einstein-podolsky-rosen paradox},\ }\href {https://doi.org/10.1103/PhysRevLett.98.140402} {\bibfield  {journal} {\bibinfo  {journal} {Phys. Rev. Lett.}\ }\textbf {\bibinfo {volume} {98}},\ \bibinfo {pages} {140402} (\bibinfo {year} {2007})}\BibitemShut {NoStop}%
\bibitem [{\citenamefont {Peres}(1990)}]{Peres-1990a}%
  \BibitemOpen
  \bibfield  {author} {\bibinfo {author} {\bibfnamefont {A.}~\bibnamefont {Peres}},\ }\bibfield  {title} {\bibinfo {title} {Neumark’s theorem and quantum inseparability},\ }\href {https://doi.org/10.1007/bf01883517} {\bibfield  {journal} {\bibinfo  {journal} {Foundations of Physics}\ }\textbf {\bibinfo {volume} {20}},\ \bibinfo {pages} {1441–1453} (\bibinfo {year} {1990})}\BibitemShut {NoStop}%
\bibitem [{\citenamefont {Kim}\ \emph {et~al.}(2026)\citenamefont {Kim}, \citenamefont {Zhang}, \citenamefont {Chitambar}, \citenamefont {Junge},\ and\ \citenamefont {Leditzky}}]{Kim-2026a}%
  \BibitemOpen
  \bibfield  {author} {\bibinfo {author} {\bibfnamefont {C.}~\bibnamefont {Kim}}, \bibinfo {author} {\bibfnamefont {Y.}~\bibnamefont {Zhang}}, \bibinfo {author} {\bibfnamefont {E.}~\bibnamefont {Chitambar}}, \bibinfo {author} {\bibfnamefont {M.}~\bibnamefont {Junge}},\ and\ \bibinfo {author} {\bibfnamefont {F.}~\bibnamefont {Leditzky}},\ }\href@noop {} {\bibinfo {title} {Compatibility of quantum instruments}} (\bibinfo {year} {2026}),\ \Eprint {https://arxiv.org/abs/arXiv:2609.40266} {arXiv:2609.40266} \BibitemShut {NoStop}%
\bibitem [{\citenamefont {Andrejic}\ and\ \citenamefont {Kunjwal}(2020)}]{Andrejic-2020a}%
  \BibitemOpen
  \bibfield  {author} {\bibinfo {author} {\bibfnamefont {N.}~\bibnamefont {Andrejic}}\ and\ \bibinfo {author} {\bibfnamefont {R.}~\bibnamefont {Kunjwal}},\ }\bibfield  {title} {\bibinfo {title} {Joint measurability structures realizable with qubit measurements: Incompatibility via marginal surgery},\ }\href {https://doi.org/10.1103/PhysRevResearch.2.043147} {\bibfield  {journal} {\bibinfo  {journal} {Phys. Rev. Res.}\ }\textbf {\bibinfo {volume} {2}},\ \bibinfo {pages} {043147} (\bibinfo {year} {2020})}\BibitemShut {NoStop}%
\bibitem [{\citenamefont {Porto}\ \emph {et~al.}(2026)\citenamefont {Porto}, \citenamefont {Designolle}, \citenamefont {Pokutta},\ and\ \citenamefont {Quintino}}]{Porto-2026a}%
  \BibitemOpen
  \bibfield  {author} {\bibinfo {author} {\bibfnamefont {L.~E.~A.}\ \bibnamefont {Porto}}, \bibinfo {author} {\bibfnamefont {S.}~\bibnamefont {Designolle}}, \bibinfo {author} {\bibfnamefont {S.}~\bibnamefont {Pokutta}},\ and\ \bibinfo {author} {\bibfnamefont {M.~T.}\ \bibnamefont {Quintino}},\ }\bibfield  {title} {\bibinfo {title} {Measurement incompatibility and quantum steering via linear programming},\ }\href {https://doi.org/10.22331/q-2026-06-19-2141} {\bibfield  {journal} {\bibinfo  {journal} {Quantum}\ }\textbf {\bibinfo {volume} {10}},\ \bibinfo {pages} {2141} (\bibinfo {year} {2026})}\BibitemShut {NoStop}%
\bibitem [{\citenamefont {Zhang}\ \emph {et~al.}(2026{\natexlab{c}})\citenamefont {Zhang}, \citenamefont {Kwek},\ and\ \citenamefont {Chen}}]{zhang2026spectral}%
  \BibitemOpen
  \bibfield  {author} {\bibinfo {author} {\bibfnamefont {Y.-X.}\ \bibnamefont {Zhang}}, \bibinfo {author} {\bibfnamefont {L.-C.}\ \bibnamefont {Kwek}},\ and\ \bibinfo {author} {\bibfnamefont {J.-L.}\ \bibnamefont {Chen}},\ }\href {https://arxiv.org/abs/2609.14272} {\bibinfo {title} {Epr-steering boundaries from a universal spectral equation}} (\bibinfo {year} {2026}{\natexlab{c}}),\ \Eprint {https://arxiv.org/abs/2609.14272} {arXiv:2609.14272 [quant-ph]} \BibitemShut {NoStop}%
\bibitem [{\citenamefont {Aubrun}\ and\ \citenamefont {Lancien}(2015)}]{Aubrun-2015a}%
  \BibitemOpen
  \bibfield  {author} {\bibinfo {author} {\bibfnamefont {G.}~\bibnamefont {Aubrun}}\ and\ \bibinfo {author} {\bibfnamefont {C.}~\bibnamefont {Lancien}},\ }\bibfield  {title} {\bibinfo {title} {Zonoids and sparsification of quantum measurements},\ }\href {https://doi.org/10.1007/s11117-015-0337-5} {\bibfield  {journal} {\bibinfo  {journal} {Positivity}\ }\textbf {\bibinfo {volume} {20}},\ \bibinfo {pages} {1–23} (\bibinfo {year} {2015})}\BibitemShut {NoStop}%
\bibitem [{\citenamefont {Ku}\ \emph {et~al.}(2022)\citenamefont {Ku}, \citenamefont {Kadlec}, \citenamefont {\ifmmode~\check{C}\else \v{C}\fi{}ernoch}, \citenamefont {Quintino}, \citenamefont {Zhou}, \citenamefont {Lemr}, \citenamefont {Lambert}, \citenamefont {Miranowicz}, \citenamefont {Chen}, \citenamefont {Nori},\ and\ \citenamefont {Chen}}]{Ku2022}%
  \BibitemOpen
  \bibfield  {author} {\bibinfo {author} {\bibfnamefont {H.-Y.}\ \bibnamefont {Ku}}, \bibinfo {author} {\bibfnamefont {J.}~\bibnamefont {Kadlec}}, \bibinfo {author} {\bibfnamefont {A.}~\bibnamefont {\ifmmode~\check{C}\else \v{C}\fi{}ernoch}}, \bibinfo {author} {\bibfnamefont {M.~T.}\ \bibnamefont {Quintino}}, \bibinfo {author} {\bibfnamefont {W.}~\bibnamefont {Zhou}}, \bibinfo {author} {\bibfnamefont {K.}~\bibnamefont {Lemr}}, \bibinfo {author} {\bibfnamefont {N.}~\bibnamefont {Lambert}}, \bibinfo {author} {\bibfnamefont {A.}~\bibnamefont {Miranowicz}}, \bibinfo {author} {\bibfnamefont {S.-L.}\ \bibnamefont {Chen}}, \bibinfo {author} {\bibfnamefont {F.}~\bibnamefont {Nori}},\ and\ \bibinfo {author} {\bibfnamefont {Y.-N.}\ \bibnamefont {Chen}},\ }\bibfield  {title} {\bibinfo {title} {Quantifying quantumness of channels without entanglement},\ }\href {https://doi.org/10.1103/PRXQuantum.3.020338} {\bibfield  {journal} {\bibinfo  {journal} {PRX Quantum}\ }\textbf {\bibinfo {volume} {3}},\ \bibinfo {pages}
  {020338} (\bibinfo {year} {2022})}\BibitemShut {NoStop}%
\bibitem [{\citenamefont {Ku}\ \emph {et~al.}(2023)\citenamefont {Ku}, \citenamefont {Lee}, \citenamefont {Lai}, \citenamefont {Lin},\ and\ \citenamefont {Chen}}]{Ku2023}%
  \BibitemOpen
  \bibfield  {author} {\bibinfo {author} {\bibfnamefont {H.-Y.}\ \bibnamefont {Ku}}, \bibinfo {author} {\bibfnamefont {K.-Y.}\ \bibnamefont {Lee}}, \bibinfo {author} {\bibfnamefont {P.-R.}\ \bibnamefont {Lai}}, \bibinfo {author} {\bibfnamefont {J.-D.}\ \bibnamefont {Lin}},\ and\ \bibinfo {author} {\bibfnamefont {Y.-N.}\ \bibnamefont {Chen}},\ }\bibfield  {title} {\bibinfo {title} {Coherent activation of a steerability-breaking channel},\ }\href {https://doi.org/10.1103/PhysRevA.107.042415} {\bibfield  {journal} {\bibinfo  {journal} {Phys. Rev. A}\ }\textbf {\bibinfo {volume} {107}},\ \bibinfo {pages} {042415} (\bibinfo {year} {2023})}\BibitemShut {NoStop}%
\bibitem [{\citenamefont {Heinosaari}\ \emph {et~al.}(2015{\natexlab{b}})\citenamefont {Heinosaari}, \citenamefont {Kiukas},\ and\ \citenamefont {Reitzner}}]{Heinosaari-2015b}%
  \BibitemOpen
  \bibfield  {author} {\bibinfo {author} {\bibfnamefont {T.}~\bibnamefont {Heinosaari}}, \bibinfo {author} {\bibfnamefont {J.}~\bibnamefont {Kiukas}},\ and\ \bibinfo {author} {\bibfnamefont {D.}~\bibnamefont {Reitzner}},\ }\bibfield  {title} {\bibinfo {title} {Noise robustness of the incompatibility of quantum measurements},\ }\href {https://doi.org/10.1103/PhysRevA.92.022115} {\bibfield  {journal} {\bibinfo  {journal} {Phys. Rev. A}\ }\textbf {\bibinfo {volume} {92}},\ \bibinfo {pages} {022115} (\bibinfo {year} {2015}{\natexlab{b}})}\BibitemShut {NoStop}%
\bibitem [{\citenamefont {Bavaresco}\ \emph {et~al.}(2017)\citenamefont {Bavaresco}, \citenamefont {Quintino}, \citenamefont {Guerini}, \citenamefont {Maciel}, \citenamefont {Cavalcanti},\ and\ \citenamefont {Cunha}}]{Bavaresco-2017a}%
  \BibitemOpen
  \bibfield  {author} {\bibinfo {author} {\bibfnamefont {J.}~\bibnamefont {Bavaresco}}, \bibinfo {author} {\bibfnamefont {M.~T.}\ \bibnamefont {Quintino}}, \bibinfo {author} {\bibfnamefont {L.}~\bibnamefont {Guerini}}, \bibinfo {author} {\bibfnamefont {T.~O.}\ \bibnamefont {Maciel}}, \bibinfo {author} {\bibfnamefont {D.}~\bibnamefont {Cavalcanti}},\ and\ \bibinfo {author} {\bibfnamefont {M.~T.}\ \bibnamefont {Cunha}},\ }\bibfield  {title} {\bibinfo {title} {Most incompatible measurements for robust steering tests},\ }\href {https://doi.org/10.1103/PhysRevA.96.022110} {\bibfield  {journal} {\bibinfo  {journal} {Phys. Rev. A}\ }\textbf {\bibinfo {volume} {96}},\ \bibinfo {pages} {022110} (\bibinfo {year} {2017})}\BibitemShut {NoStop}%
\bibitem [{\citenamefont {Piani}\ \emph {et~al.}(2016)\citenamefont {Piani}, \citenamefont {Cianciaruso}, \citenamefont {Bromley}, \citenamefont {Napoli}, \citenamefont {Johnston},\ and\ \citenamefont {Adesso}}]{Piani-2016a}%
  \BibitemOpen
  \bibfield  {author} {\bibinfo {author} {\bibfnamefont {M.}~\bibnamefont {Piani}}, \bibinfo {author} {\bibfnamefont {M.}~\bibnamefont {Cianciaruso}}, \bibinfo {author} {\bibfnamefont {T.~R.}\ \bibnamefont {Bromley}}, \bibinfo {author} {\bibfnamefont {C.}~\bibnamefont {Napoli}}, \bibinfo {author} {\bibfnamefont {N.}~\bibnamefont {Johnston}},\ and\ \bibinfo {author} {\bibfnamefont {G.}~\bibnamefont {Adesso}},\ }\bibfield  {title} {\bibinfo {title} {Robustness of asymmetry and coherence of quantum states},\ }\href {https://doi.org/10.1103/PhysRevA.93.042107} {\bibfield  {journal} {\bibinfo  {journal} {Phys. Rev. A}\ }\textbf {\bibinfo {volume} {93}},\ \bibinfo {pages} {042107} (\bibinfo {year} {2016})}\BibitemShut {NoStop}%
\bibitem [{\citenamefont {Oszmaniec}\ and\ \citenamefont {Biswas}(2019)}]{Oszmaniec-2019a}%
  \BibitemOpen
  \bibfield  {author} {\bibinfo {author} {\bibfnamefont {M.}~\bibnamefont {Oszmaniec}}\ and\ \bibinfo {author} {\bibfnamefont {T.}~\bibnamefont {Biswas}},\ }\bibfield  {title} {\bibinfo {title} {Operational relevance of resource theories of quantum measurements},\ }\href {https://doi.org/10.22331/q-2019-04-26-133} {\bibfield  {journal} {\bibinfo  {journal} {Quantum}\ }\textbf {\bibinfo {volume} {3}},\ \bibinfo {pages} {133} (\bibinfo {year} {2019})}\BibitemShut {NoStop}%
\bibitem [{\citenamefont {Chitambar}\ and\ \citenamefont {Gour}(2019)}]{Chitambar-2019a}%
  \BibitemOpen
  \bibfield  {author} {\bibinfo {author} {\bibfnamefont {E.}~\bibnamefont {Chitambar}}\ and\ \bibinfo {author} {\bibfnamefont {G.}~\bibnamefont {Gour}},\ }\bibfield  {title} {\bibinfo {title} {Quantum resource theories},\ }\href {https://doi.org/10.1103/RevModPhys.91.025001} {\bibfield  {journal} {\bibinfo  {journal} {Rev. Mod. Phys.}\ }\textbf {\bibinfo {volume} {91}},\ \bibinfo {pages} {025001} (\bibinfo {year} {2019})}\BibitemShut {NoStop}%
\bibitem [{\citenamefont {Takagi}\ and\ \citenamefont {Regula}(2019)}]{Takagi-2019a}%
  \BibitemOpen
  \bibfield  {author} {\bibinfo {author} {\bibfnamefont {R.}~\bibnamefont {Takagi}}\ and\ \bibinfo {author} {\bibfnamefont {B.}~\bibnamefont {Regula}},\ }\bibfield  {title} {\bibinfo {title} {General resource theories in quantum mechanics and beyond: Operational characterization via discrimination tasks},\ }\href {https://doi.org/10.1103/PhysRevX.9.031053} {\bibfield  {journal} {\bibinfo  {journal} {Phys. Rev. X}\ }\textbf {\bibinfo {volume} {9}},\ \bibinfo {pages} {031053} (\bibinfo {year} {2019})}\BibitemShut {NoStop}%
\bibitem [{\citenamefont {Kuroiwa}\ \emph {et~al.}(2024)\citenamefont {Kuroiwa}, \citenamefont {Takagi}, \citenamefont {Adesso},\ and\ \citenamefont {Yamasaki}}]{Kuroiwa-2024a}%
  \BibitemOpen
  \bibfield  {author} {\bibinfo {author} {\bibfnamefont {K.}~\bibnamefont {Kuroiwa}}, \bibinfo {author} {\bibfnamefont {R.}~\bibnamefont {Takagi}}, \bibinfo {author} {\bibfnamefont {G.}~\bibnamefont {Adesso}},\ and\ \bibinfo {author} {\bibfnamefont {H.}~\bibnamefont {Yamasaki}},\ }\bibfield  {title} {\bibinfo {title} {Robustness- and weight-based resource measures without convexity restriction: Multicopy witness and operational advantage in static and dynamical quantum resource theories},\ }\href {https://doi.org/10.1103/PhysRevA.109.042403} {\bibfield  {journal} {\bibinfo  {journal} {Phys. Rev. A}\ }\textbf {\bibinfo {volume} {109}},\ \bibinfo {pages} {042403} (\bibinfo {year} {2024})}\BibitemShut {NoStop}%
\bibitem [{\citenamefont {Uola}\ \emph {et~al.}(2015)\citenamefont {Uola}, \citenamefont {Budroni}, \citenamefont {G\"uhne},\ and\ \citenamefont {Pellonp\"a\"a}}]{Uola2015}%
  \BibitemOpen
  \bibfield  {author} {\bibinfo {author} {\bibfnamefont {R.}~\bibnamefont {Uola}}, \bibinfo {author} {\bibfnamefont {C.}~\bibnamefont {Budroni}}, \bibinfo {author} {\bibfnamefont {O.}~\bibnamefont {G\"uhne}},\ and\ \bibinfo {author} {\bibfnamefont {J.-P.}\ \bibnamefont {Pellonp\"a\"a}},\ }\bibfield  {title} {\bibinfo {title} {One-to-one mapping between steering and joint measurability problems},\ }\href {https://doi.org/10.1103/PhysRevLett.115.230402} {\bibfield  {journal} {\bibinfo  {journal} {Phys. Rev. Lett.}\ }\textbf {\bibinfo {volume} {115}},\ \bibinfo {pages} {230402} (\bibinfo {year} {2015})}\BibitemShut {NoStop}%
\bibitem [{\citenamefont {Quintino}\ \emph {et~al.}(2015)\citenamefont {Quintino}, \citenamefont {V\'ertesi}, \citenamefont {Cavalcanti}, \citenamefont {Augusiak}, \citenamefont {Demianowicz}, \citenamefont {Ac\'{\i}n},\ and\ \citenamefont {Brunner}}]{Quintino-2015a}%
  \BibitemOpen
  \bibfield  {author} {\bibinfo {author} {\bibfnamefont {M.~T.}\ \bibnamefont {Quintino}}, \bibinfo {author} {\bibfnamefont {T.}~\bibnamefont {V\'ertesi}}, \bibinfo {author} {\bibfnamefont {D.}~\bibnamefont {Cavalcanti}}, \bibinfo {author} {\bibfnamefont {R.}~\bibnamefont {Augusiak}}, \bibinfo {author} {\bibfnamefont {M.}~\bibnamefont {Demianowicz}}, \bibinfo {author} {\bibfnamefont {A.}~\bibnamefont {Ac\'{\i}n}},\ and\ \bibinfo {author} {\bibfnamefont {N.}~\bibnamefont {Brunner}},\ }\bibfield  {title} {\bibinfo {title} {Inequivalence of entanglement, steering, and bell nonlocality for general measurements},\ }\href {https://doi.org/10.1103/PhysRevA.92.032107} {\bibfield  {journal} {\bibinfo  {journal} {Phys. Rev. A}\ }\textbf {\bibinfo {volume} {92}},\ \bibinfo {pages} {032107} (\bibinfo {year} {2015})}\BibitemShut {NoStop}%
\bibitem [{\citenamefont {Werner}(1989)}]{Werner-1989a}%
  \BibitemOpen
  \bibfield  {author} {\bibinfo {author} {\bibfnamefont {R.~F.}\ \bibnamefont {Werner}},\ }\bibfield  {title} {\bibinfo {title} {Quantum states with einstein-podolsky-rosen correlations admitting a hidden-variable model},\ }\href {https://doi.org/10.1103/PhysRevA.40.4277} {\bibfield  {journal} {\bibinfo  {journal} {Phys. Rev. A}\ }\textbf {\bibinfo {volume} {40}},\ \bibinfo {pages} {4277} (\bibinfo {year} {1989})}\BibitemShut {NoStop}%
\bibitem [{\citenamefont {Renner}(2024)}]{Renner-2024a}%
  \BibitemOpen
  \bibfield  {author} {\bibinfo {author} {\bibfnamefont {M.~J.}\ \bibnamefont {Renner}},\ }\bibfield  {title} {\bibinfo {title} {Compatibility of generalized noisy qubit measurements},\ }\href {https://doi.org/10.1103/PhysRevLett.132.250202} {\bibfield  {journal} {\bibinfo  {journal} {Phys. Rev. Lett.}\ }\textbf {\bibinfo {volume} {132}},\ \bibinfo {pages} {250202} (\bibinfo {year} {2024})}\BibitemShut {NoStop}%
\bibitem [{\citenamefont {Zhang}\ and\ \citenamefont {Chitambar}(2024)}]{Zhang-2024a}%
  \BibitemOpen
  \bibfield  {author} {\bibinfo {author} {\bibfnamefont {Y.}~\bibnamefont {Zhang}}\ and\ \bibinfo {author} {\bibfnamefont {E.}~\bibnamefont {Chitambar}},\ }\bibfield  {title} {\bibinfo {title} {Exact steering bound for two-qubit werner states},\ }\href {https://doi.org/10.1103/PhysRevLett.132.250201} {\bibfield  {journal} {\bibinfo  {journal} {Phys. Rev. Lett.}\ }\textbf {\bibinfo {volume} {132}},\ \bibinfo {pages} {250201} (\bibinfo {year} {2024})}\BibitemShut {NoStop}%
\bibitem [{\citenamefont {Barrett}(2002)}]{Barrett-2002a}%
  \BibitemOpen
  \bibfield  {author} {\bibinfo {author} {\bibfnamefont {J.}~\bibnamefont {Barrett}},\ }\bibfield  {title} {\bibinfo {title} {Nonsequential positive-operator-valued measurements on entangled mixed states do not always violate a bell inequality},\ }\href {https://doi.org/10.1103/PhysRevA.65.042302} {\bibfield  {journal} {\bibinfo  {journal} {Phys. Rev. A}\ }\textbf {\bibinfo {volume} {65}},\ \bibinfo {pages} {042302} (\bibinfo {year} {2002})}\BibitemShut {NoStop}%
\bibitem [{\citenamefont {Zhang}\ \emph {et~al.}(2025)\citenamefont {Zhang}, \citenamefont {Zhang},\ and\ \citenamefont {Chitambar}}]{Zhang-2025a}%
  \BibitemOpen
  \bibfield  {author} {\bibinfo {author} {\bibfnamefont {Y.}~\bibnamefont {Zhang}}, \bibinfo {author} {\bibfnamefont {J.}~\bibnamefont {Zhang}},\ and\ \bibinfo {author} {\bibfnamefont {E.}~\bibnamefont {Chitambar}},\ }\bibfield  {title} {\bibinfo {title} {Cost of simulating entanglement in steering scenarios},\ }\href {https://doi.org/10.22331/q-2025-10-31-1902} {\bibfield  {journal} {\bibinfo  {journal} {Quantum}\ }\textbf {\bibinfo {volume} {9}},\ \bibinfo {pages} {1902} (\bibinfo {year} {2025})}\BibitemShut {NoStop}%
\bibitem [{\citenamefont {Cavalcanti}\ and\ \citenamefont {Skrzypczyk}(2016{\natexlab{a}})}]{Cavalcanti-2017a}%
  \BibitemOpen
  \bibfield  {author} {\bibinfo {author} {\bibfnamefont {D.}~\bibnamefont {Cavalcanti}}\ and\ \bibinfo {author} {\bibfnamefont {P.}~\bibnamefont {Skrzypczyk}},\ }\bibfield  {title} {\bibinfo {title} {Quantum steering: a review with focus on semidefinite programming},\ }\href {https://doi.org/10.1088/1361-6633/80/2/024001} {\bibfield  {journal} {\bibinfo  {journal} {Reports on Progress in Physics}\ }\textbf {\bibinfo {volume} {80}},\ \bibinfo {pages} {024001} (\bibinfo {year} {2016}{\natexlab{a}})}\BibitemShut {NoStop}%
\bibitem [{\citenamefont {Uola}\ \emph {et~al.}(2020)\citenamefont {Uola}, \citenamefont {Costa}, \citenamefont {Nguyen},\ and\ \citenamefont {G\"uhne}}]{Uola-2020a}%
  \BibitemOpen
  \bibfield  {author} {\bibinfo {author} {\bibfnamefont {R.}~\bibnamefont {Uola}}, \bibinfo {author} {\bibfnamefont {A.~C.~S.}\ \bibnamefont {Costa}}, \bibinfo {author} {\bibfnamefont {H.~C.}\ \bibnamefont {Nguyen}},\ and\ \bibinfo {author} {\bibfnamefont {O.}~\bibnamefont {G\"uhne}},\ }\bibfield  {title} {\bibinfo {title} {Quantum steering},\ }\href {https://doi.org/10.1103/RevModPhys.92.015001} {\bibfield  {journal} {\bibinfo  {journal} {Rev. Mod. Phys.}\ }\textbf {\bibinfo {volume} {92}},\ \bibinfo {pages} {015001} (\bibinfo {year} {2020})}\BibitemShut {NoStop}%
\bibitem [{\citenamefont {Jevtic}\ \emph {et~al.}(2015)\citenamefont {Jevtic}, \citenamefont {Hall}, \citenamefont {Anderson}, \citenamefont {Zwierz},\ and\ \citenamefont {Wiseman}}]{Jevtic-2015a}%
  \BibitemOpen
  \bibfield  {author} {\bibinfo {author} {\bibfnamefont {S.}~\bibnamefont {Jevtic}}, \bibinfo {author} {\bibfnamefont {M.~J.~W.}\ \bibnamefont {Hall}}, \bibinfo {author} {\bibfnamefont {M.~R.}\ \bibnamefont {Anderson}}, \bibinfo {author} {\bibfnamefont {M.}~\bibnamefont {Zwierz}},\ and\ \bibinfo {author} {\bibfnamefont {H.~M.}\ \bibnamefont {Wiseman}},\ }\bibfield  {title} {\bibinfo {title} {Einstein–podolsky–rosen steering and the steering ellipsoid},\ }\href {https://doi.org/10.1364/josab.32.000a40} {\bibfield  {journal} {\bibinfo  {journal} {Journal of the Optical Society of America B}\ }\textbf {\bibinfo {volume} {32}},\ \bibinfo {pages} {A40} (\bibinfo {year} {2015})}\BibitemShut {NoStop}%
\bibitem [{\citenamefont {Nguyen}\ and\ \citenamefont {Vu}(2016)}]{Nguyen-2016a}%
  \BibitemOpen
  \bibfield  {author} {\bibinfo {author} {\bibfnamefont {H.~C.}\ \bibnamefont {Nguyen}}\ and\ \bibinfo {author} {\bibfnamefont {T.}~\bibnamefont {Vu}},\ }\bibfield  {title} {\bibinfo {title} {Necessary and sufficient condition for steerability of two-qubit states by the geometry of steering outcomes},\ }\href {https://doi.org/10.1209/0295-5075/115/10003} {\bibfield  {journal} {\bibinfo  {journal} {EPL (Europhysics Letters)}\ }\textbf {\bibinfo {volume} {115}},\ \bibinfo {pages} {10003} (\bibinfo {year} {2016})}\BibitemShut {NoStop}%
\bibitem [{\citenamefont {Nguyen}\ \emph {et~al.}(2019)\citenamefont {Nguyen}, \citenamefont {Nguyen},\ and\ \citenamefont {G\"uhne}}]{Nguyen-2019a}%
  \BibitemOpen
  \bibfield  {author} {\bibinfo {author} {\bibfnamefont {H.~C.}\ \bibnamefont {Nguyen}}, \bibinfo {author} {\bibfnamefont {H.-V.}\ \bibnamefont {Nguyen}},\ and\ \bibinfo {author} {\bibfnamefont {O.}~\bibnamefont {G\"uhne}},\ }\bibfield  {title} {\bibinfo {title} {Geometry of einstein-podolsky-rosen correlations},\ }\href {https://doi.org/10.1103/PhysRevLett.122.240401} {\bibfield  {journal} {\bibinfo  {journal} {Phys. Rev. Lett.}\ }\textbf {\bibinfo {volume} {122}},\ \bibinfo {pages} {240401} (\bibinfo {year} {2019})}\BibitemShut {NoStop}%
\bibitem [{\citenamefont {Zhang}(2026)}]{zhang2026parallel}%
  \BibitemOpen
  \bibfield  {author} {\bibinfo {author} {\bibfnamefont {Y.}~\bibnamefont {Zhang}},\ }\bibfield  {title} {\bibinfo {title} {Exact incompatibility-breaking criterion for unital qubit channels}} (\bibinfo {year} {2026})\BibitemShut {NoStop}%
\bibitem [{\citenamefont {Bowles}\ \emph {et~al.}(2014)\citenamefont {Bowles}, \citenamefont {V\'ertesi}, \citenamefont {Quintino},\ and\ \citenamefont {Brunner}}]{Bowles-2014a}%
  \BibitemOpen
  \bibfield  {author} {\bibinfo {author} {\bibfnamefont {J.}~\bibnamefont {Bowles}}, \bibinfo {author} {\bibfnamefont {T.}~\bibnamefont {V\'ertesi}}, \bibinfo {author} {\bibfnamefont {M.~T.}\ \bibnamefont {Quintino}},\ and\ \bibinfo {author} {\bibfnamefont {N.}~\bibnamefont {Brunner}},\ }\bibfield  {title} {\bibinfo {title} {One-way einstein-podolsky-rosen steering},\ }\href {https://doi.org/10.1103/PhysRevLett.112.200402} {\bibfield  {journal} {\bibinfo  {journal} {Phys. Rev. Lett.}\ }\textbf {\bibinfo {volume} {112}},\ \bibinfo {pages} {200402} (\bibinfo {year} {2014})}\BibitemShut {NoStop}%
\bibitem [{\citenamefont {Bowles}\ \emph {et~al.}(2016)\citenamefont {Bowles}, \citenamefont {Hirsch}, \citenamefont {Quintino},\ and\ \citenamefont {Brunner}}]{Bowles-2016a}%
  \BibitemOpen
  \bibfield  {author} {\bibinfo {author} {\bibfnamefont {J.}~\bibnamefont {Bowles}}, \bibinfo {author} {\bibfnamefont {F.}~\bibnamefont {Hirsch}}, \bibinfo {author} {\bibfnamefont {M.~T.}\ \bibnamefont {Quintino}},\ and\ \bibinfo {author} {\bibfnamefont {N.}~\bibnamefont {Brunner}},\ }\bibfield  {title} {\bibinfo {title} {Sufficient criterion for guaranteeing that a two-qubit state is unsteerable},\ }\href {https://doi.org/10.1103/PhysRevA.93.022121} {\bibfield  {journal} {\bibinfo  {journal} {Phys. Rev. A}\ }\textbf {\bibinfo {volume} {93}},\ \bibinfo {pages} {022121} (\bibinfo {year} {2016})}\BibitemShut {NoStop}%
\bibitem [{\citenamefont {Song}\ and\ \citenamefont {Bae}(2026)}]{Song-2026a}%
  \BibitemOpen
  \bibfield  {author} {\bibinfo {author} {\bibfnamefont {Q.-C.}\ \bibnamefont {Song}}\ and\ \bibinfo {author} {\bibfnamefont {J.}~\bibnamefont {Bae}},\ }\href@noop {} {\bibinfo {title} {One pure steered state implies einstein-podolsky-rosen steering}} (\bibinfo {year} {2026}),\ \Eprint {https://arxiv.org/abs/arXiv:2605.18243} {arXiv:2605.18243} \BibitemShut {NoStop}%
\bibitem [{\citenamefont {Busch}(1986)}]{Busch-1986a}%
  \BibitemOpen
  \bibfield  {author} {\bibinfo {author} {\bibfnamefont {P.}~\bibnamefont {Busch}},\ }\bibfield  {title} {\bibinfo {title} {Unsharp reality and joint measurements for spin observables},\ }\href {https://doi.org/10.1103/PhysRevD.33.2253} {\bibfield  {journal} {\bibinfo  {journal} {Phys. Rev. D}\ }\textbf {\bibinfo {volume} {33}},\ \bibinfo {pages} {2253} (\bibinfo {year} {1986})}\BibitemShut {NoStop}%
\bibitem [{\citenamefont {Stano}\ \emph {et~al.}(2008)\citenamefont {Stano}, \citenamefont {Reitzner},\ and\ \citenamefont {Heinosaari}}]{Stano-2008a}%
  \BibitemOpen
  \bibfield  {author} {\bibinfo {author} {\bibfnamefont {P.}~\bibnamefont {Stano}}, \bibinfo {author} {\bibfnamefont {D.}~\bibnamefont {Reitzner}},\ and\ \bibinfo {author} {\bibfnamefont {T.}~\bibnamefont {Heinosaari}},\ }\bibfield  {title} {\bibinfo {title} {Coexistence of qubit effects},\ }\href {https://doi.org/10.1103/PhysRevA.78.012315} {\bibfield  {journal} {\bibinfo  {journal} {Phys. Rev. A}\ }\textbf {\bibinfo {volume} {78}},\ \bibinfo {pages} {012315} (\bibinfo {year} {2008})}\BibitemShut {NoStop}%
\bibitem [{\citenamefont {Busch}\ and\ \citenamefont {Schmidt}(2009)}]{Busch-2009a}%
  \BibitemOpen
  \bibfield  {author} {\bibinfo {author} {\bibfnamefont {P.}~\bibnamefont {Busch}}\ and\ \bibinfo {author} {\bibfnamefont {H.-J.}\ \bibnamefont {Schmidt}},\ }\bibfield  {title} {\bibinfo {title} {Coexistence of qubit effects},\ }\href {https://doi.org/10.1007/s11128-009-0109-x} {\bibfield  {journal} {\bibinfo  {journal} {Quantum Information Processing}\ }\textbf {\bibinfo {volume} {9}},\ \bibinfo {pages} {143–169} (\bibinfo {year} {2009})}\BibitemShut {NoStop}%
\bibitem [{\citenamefont {Yu}\ \emph {et~al.}(2010)\citenamefont {Yu}, \citenamefont {Liu}, \citenamefont {Li},\ and\ \citenamefont {Oh}}]{Yu-2010a}%
  \BibitemOpen
  \bibfield  {author} {\bibinfo {author} {\bibfnamefont {S.}~\bibnamefont {Yu}}, \bibinfo {author} {\bibfnamefont {N.-l.}\ \bibnamefont {Liu}}, \bibinfo {author} {\bibfnamefont {L.}~\bibnamefont {Li}},\ and\ \bibinfo {author} {\bibfnamefont {C.~H.}\ \bibnamefont {Oh}},\ }\bibfield  {title} {\bibinfo {title} {Joint measurement of two unsharp observables of a qubit},\ }\href {https://doi.org/10.1103/PhysRevA.81.062116} {\bibfield  {journal} {\bibinfo  {journal} {Phys. Rev. A}\ }\textbf {\bibinfo {volume} {81}},\ \bibinfo {pages} {062116} (\bibinfo {year} {2010})}\BibitemShut {NoStop}%
\bibitem [{\citenamefont {Liang}\ \emph {et~al.}(2011)\citenamefont {Liang}, \citenamefont {Spekkens},\ and\ \citenamefont {Wiseman}}]{Liang-2011a}%
  \BibitemOpen
  \bibfield  {author} {\bibinfo {author} {\bibfnamefont {Y.-C.}\ \bibnamefont {Liang}}, \bibinfo {author} {\bibfnamefont {R.~W.}\ \bibnamefont {Spekkens}},\ and\ \bibinfo {author} {\bibfnamefont {H.~M.}\ \bibnamefont {Wiseman}},\ }\bibfield  {title} {\bibinfo {title} {Specker’s parable of the overprotective seer: A road to contextuality, nonlocality and complementarity},\ }\href {https://doi.org/10.1016/j.physrep.2011.05.001} {\bibfield  {journal} {\bibinfo  {journal} {Physics Reports}\ }\textbf {\bibinfo {volume} {506}},\ \bibinfo {pages} {1–39} (\bibinfo {year} {2011})}\BibitemShut {NoStop}%
\bibitem [{\citenamefont {Yu}\ and\ \citenamefont {Oh}(2013)}]{Yu-2013a}%
  \BibitemOpen
  \bibfield  {author} {\bibinfo {author} {\bibfnamefont {S.}~\bibnamefont {Yu}}\ and\ \bibinfo {author} {\bibfnamefont {C.~H.}\ \bibnamefont {Oh}},\ }\href@noop {} {\bibinfo {title} {Quantum contextuality and joint measurement of three observables of a qubit}} (\bibinfo {year} {2013}),\ \Eprint {https://arxiv.org/abs/arXiv:1312.6470} {arXiv:1312.6470} \BibitemShut {NoStop}%
\bibitem [{\citenamefont {Kunjwal}\ \emph {et~al.}(2014)\citenamefont {Kunjwal}, \citenamefont {Heunen},\ and\ \citenamefont {Fritz}}]{Kunjwal-2014a}%
  \BibitemOpen
  \bibfield  {author} {\bibinfo {author} {\bibfnamefont {R.}~\bibnamefont {Kunjwal}}, \bibinfo {author} {\bibfnamefont {C.}~\bibnamefont {Heunen}},\ and\ \bibinfo {author} {\bibfnamefont {T.}~\bibnamefont {Fritz}},\ }\bibfield  {title} {\bibinfo {title} {Quantum realization of arbitrary joint measurability structures},\ }\href {https://doi.org/10.1103/PhysRevA.89.052126} {\bibfield  {journal} {\bibinfo  {journal} {Phys. Rev. A}\ }\textbf {\bibinfo {volume} {89}},\ \bibinfo {pages} {052126} (\bibinfo {year} {2014})}\BibitemShut {NoStop}%
\bibitem [{\citenamefont {Uola}\ \emph {et~al.}(2016)\citenamefont {Uola}, \citenamefont {Luoma}, \citenamefont {Moroder},\ and\ \citenamefont {Heinosaari}}]{Uola-2016a}%
  \BibitemOpen
  \bibfield  {author} {\bibinfo {author} {\bibfnamefont {R.}~\bibnamefont {Uola}}, \bibinfo {author} {\bibfnamefont {K.}~\bibnamefont {Luoma}}, \bibinfo {author} {\bibfnamefont {T.}~\bibnamefont {Moroder}},\ and\ \bibinfo {author} {\bibfnamefont {T.}~\bibnamefont {Heinosaari}},\ }\bibfield  {title} {\bibinfo {title} {Adaptive strategy for joint measurements},\ }\href {https://doi.org/10.1103/PhysRevA.94.022109} {\bibfield  {journal} {\bibinfo  {journal} {Phys. Rev. A}\ }\textbf {\bibinfo {volume} {94}},\ \bibinfo {pages} {022109} (\bibinfo {year} {2016})}\BibitemShut {NoStop}%
\bibitem [{\citenamefont {Jae}\ \emph {et~al.}(2019)\citenamefont {Jae}, \citenamefont {Baek}, \citenamefont {Ryu},\ and\ \citenamefont {Lee}}]{Jae-2019a}%
  \BibitemOpen
  \bibfield  {author} {\bibinfo {author} {\bibfnamefont {J.}~\bibnamefont {Jae}}, \bibinfo {author} {\bibfnamefont {K.}~\bibnamefont {Baek}}, \bibinfo {author} {\bibfnamefont {J.}~\bibnamefont {Ryu}},\ and\ \bibinfo {author} {\bibfnamefont {J.}~\bibnamefont {Lee}},\ }\bibfield  {title} {\bibinfo {title} {Necessary and sufficient condition for joint measurability},\ }\href {https://doi.org/10.1103/PhysRevA.100.032113} {\bibfield  {journal} {\bibinfo  {journal} {Phys. Rev. A}\ }\textbf {\bibinfo {volume} {100}},\ \bibinfo {pages} {032113} (\bibinfo {year} {2019})}\BibitemShut {NoStop}%
\bibitem [{\citenamefont {Cavalcanti}\ and\ \citenamefont {Skrzypczyk}(2016{\natexlab{b}})}]{Cavalcanti-2016a}%
  \BibitemOpen
  \bibfield  {author} {\bibinfo {author} {\bibfnamefont {D.}~\bibnamefont {Cavalcanti}}\ and\ \bibinfo {author} {\bibfnamefont {P.}~\bibnamefont {Skrzypczyk}},\ }\bibfield  {title} {\bibinfo {title} {Quantitative relations between measurement incompatibility, quantum steering, and nonlocality},\ }\href {https://doi.org/10.1103/PhysRevA.93.052112} {\bibfield  {journal} {\bibinfo  {journal} {Phys. Rev. A}\ }\textbf {\bibinfo {volume} {93}},\ \bibinfo {pages} {052112} (\bibinfo {year} {2016}{\natexlab{b}})}\BibitemShut {NoStop}%
\bibitem [{\citenamefont {Skrzypczyk}\ \emph {et~al.}(2020)\citenamefont {Skrzypczyk}, \citenamefont {Hoban}, \citenamefont {Sainz},\ and\ \citenamefont {Linden}}]{Skrzypczyk-2020a}%
  \BibitemOpen
  \bibfield  {author} {\bibinfo {author} {\bibfnamefont {P.}~\bibnamefont {Skrzypczyk}}, \bibinfo {author} {\bibfnamefont {M.~J.}\ \bibnamefont {Hoban}}, \bibinfo {author} {\bibfnamefont {A.~B.}\ \bibnamefont {Sainz}},\ and\ \bibinfo {author} {\bibfnamefont {N.}~\bibnamefont {Linden}},\ }\bibfield  {title} {\bibinfo {title} {Complexity of compatible measurements},\ }\href {https://doi.org/10.1103/PhysRevResearch.2.023292} {\bibfield  {journal} {\bibinfo  {journal} {Phys. Rev. Res.}\ }\textbf {\bibinfo {volume} {2}},\ \bibinfo {pages} {023292} (\bibinfo {year} {2020})}\BibitemShut {NoStop}%
\bibitem [{\citenamefont {Grinko}\ and\ \citenamefont {Uola}(2025)}]{Grinko-2025a}%
  \BibitemOpen
  \bibfield  {author} {\bibinfo {author} {\bibfnamefont {D.}~\bibnamefont {Grinko}}\ and\ \bibinfo {author} {\bibfnamefont {R.}~\bibnamefont {Uola}},\ }\bibfield  {title} {\bibinfo {title} {Compatibility of binary qubit measurements},\ }\href {https://doi.org/10.1103/vv1h-5mf9} {\bibfield  {journal} {\bibinfo  {journal} {Phys. Rev. Lett.}\ }\textbf {\bibinfo {volume} {135}},\ \bibinfo {pages} {200201} (\bibinfo {year} {2025})}\BibitemShut {NoStop}%
\bibitem [{\citenamefont {Bolker}(1969)}]{Bolker-1969a}%
  \BibitemOpen
  \bibfield  {author} {\bibinfo {author} {\bibfnamefont {E.~D.}\ \bibnamefont {Bolker}},\ }\bibfield  {title} {\bibinfo {title} {A class of convex bodies},\ }\href {https://doi.org/10.1090/s0002-9947-1969-0256265-x} {\bibfield  {journal} {\bibinfo  {journal} {Transactions of the American Mathematical Society}\ }\textbf {\bibinfo {volume} {145}},\ \bibinfo {pages} {323–345} (\bibinfo {year} {1969})}\BibitemShut {NoStop}%
\bibitem [{\citenamefont {Goodey}\ and\ \citenamefont {Weil}(1993)}]{Goodey-1993a}%
  \BibitemOpen
  \bibfield  {author} {\bibinfo {author} {\bibfnamefont {P.}~\bibnamefont {Goodey}}\ and\ \bibinfo {author} {\bibfnamefont {W.}~\bibnamefont {Weil}},\ }\bibinfo {title} {Zonoids and generalisations},\ in\ \href {https://doi.org/10.1016/b978-0-444-89597-4.50020-2} {\emph {\bibinfo {booktitle} {Handbook of Convex Geometry}}}\ (\bibinfo  {publisher} {Elsevier},\ \bibinfo {year} {1993})\ p.\ \bibinfo {pages} {1297–1326}\BibitemShut {NoStop}%
\bibitem [{\citenamefont {Buscemi}\ \emph {et~al.}(2005)\citenamefont {Buscemi}, \citenamefont {Keyl}, \citenamefont {D’Ariano}, \citenamefont {Perinotti},\ and\ \citenamefont {Werner}}]{Buscemi-2005a}%
  \BibitemOpen
  \bibfield  {author} {\bibinfo {author} {\bibfnamefont {F.}~\bibnamefont {Buscemi}}, \bibinfo {author} {\bibfnamefont {M.}~\bibnamefont {Keyl}}, \bibinfo {author} {\bibfnamefont {G.~M.}\ \bibnamefont {D’Ariano}}, \bibinfo {author} {\bibfnamefont {P.}~\bibnamefont {Perinotti}},\ and\ \bibinfo {author} {\bibfnamefont {R.~F.}\ \bibnamefont {Werner}},\ }\bibfield  {title} {\bibinfo {title} {Clean positive operator valued measures},\ }\bibfield  {journal} {\bibinfo  {journal} {Journal of Mathematical Physics}\ }\textbf {\bibinfo {volume} {46}},\ \href {https://doi.org/10.1063/1.2008996} {10.1063/1.2008996} (\bibinfo {year} {2005})\BibitemShut {NoStop}%
\bibitem [{\citenamefont {Heinonen}(2005)}]{Heinonen-2005a}%
  \BibitemOpen
  \bibfield  {author} {\bibinfo {author} {\bibfnamefont {T.}~\bibnamefont {Heinonen}},\ }\bibfield  {title} {\bibinfo {title} {Optimal measurements in quantum mechanics},\ }\href {https://doi.org/10.1016/j.physleta.2005.08.003} {\bibfield  {journal} {\bibinfo  {journal} {Physics Letters A}\ }\textbf {\bibinfo {volume} {346}},\ \bibinfo {pages} {77–86} (\bibinfo {year} {2005})}\BibitemShut {NoStop}%
\bibitem [{\citenamefont {Kelly}\ and\ \citenamefont {Weiss}(1979)}]{Kelly-1979a}%
  \BibitemOpen
  \bibfield  {author} {\bibinfo {author} {\bibfnamefont {P.}~\bibnamefont {Kelly}}\ and\ \bibinfo {author} {\bibfnamefont {M.}~\bibnamefont {Weiss}},\ }\href@noop {} {\emph {\bibinfo {title} {Geometry and Convexity: A Study in Mathematical Methods}}}\ (\bibinfo  {publisher} {Wiley},\ \bibinfo {year} {1979})\BibitemShut {NoStop}%
\bibitem [{\citenamefont {Beckman}\ \emph {et~al.}(2001)\citenamefont {Beckman}, \citenamefont {Gottesman}, \citenamefont {Nielsen},\ and\ \citenamefont {Preskill}}]{Beckman-2001a}%
  \BibitemOpen
  \bibfield  {author} {\bibinfo {author} {\bibfnamefont {D.}~\bibnamefont {Beckman}}, \bibinfo {author} {\bibfnamefont {D.}~\bibnamefont {Gottesman}}, \bibinfo {author} {\bibfnamefont {M.~A.}\ \bibnamefont {Nielsen}},\ and\ \bibinfo {author} {\bibfnamefont {J.}~\bibnamefont {Preskill}},\ }\bibfield  {title} {\bibinfo {title} {Causal and localizable quantum operations},\ }\href {https://doi.org/10.1103/PhysRevA.64.052309} {\bibfield  {journal} {\bibinfo  {journal} {Phys. Rev. A}\ }\textbf {\bibinfo {volume} {64}},\ \bibinfo {pages} {052309} (\bibinfo {year} {2001})}\BibitemShut {NoStop}%
\bibitem [{\citenamefont {Pauwels}\ \emph {et~al.}(2025)\citenamefont {Pauwels}, \citenamefont {Pozas-Kerstjens}, \citenamefont {Del~Santo},\ and\ \citenamefont {Gisin}}]{Pauwels-2025a}%
  \BibitemOpen
  \bibfield  {author} {\bibinfo {author} {\bibfnamefont {J.}~\bibnamefont {Pauwels}}, \bibinfo {author} {\bibfnamefont {A.}~\bibnamefont {Pozas-Kerstjens}}, \bibinfo {author} {\bibfnamefont {F.}~\bibnamefont {Del~Santo}},\ and\ \bibinfo {author} {\bibfnamefont {N.}~\bibnamefont {Gisin}},\ }\bibfield  {title} {\bibinfo {title} {Classification of joint quantum measurements based on entanglement cost of localization},\ }\href {https://doi.org/10.1103/PhysRevX.15.021013} {\bibfield  {journal} {\bibinfo  {journal} {Phys. Rev. X}\ }\textbf {\bibinfo {volume} {15}},\ \bibinfo {pages} {021013} (\bibinfo {year} {2025})}\BibitemShut {NoStop}%
\bibitem [{\citenamefont {Quintino}\ \emph {et~al.}(2016{\natexlab{b}})\citenamefont {Quintino}, \citenamefont {Brunner},\ and\ \citenamefont {Huber}}]{Quintino2016a}%
  \BibitemOpen
  \bibfield  {author} {\bibinfo {author} {\bibfnamefont {M.~T.}\ \bibnamefont {Quintino}}, \bibinfo {author} {\bibfnamefont {N.}~\bibnamefont {Brunner}},\ and\ \bibinfo {author} {\bibfnamefont {M.}~\bibnamefont {Huber}},\ }\bibfield  {title} {\bibinfo {title} {Superactivation of quantum steering},\ }\href {https://doi.org/10.1103/PhysRevA.94.062123} {\bibfield  {journal} {\bibinfo  {journal} {Phys. Rev. A}\ }\textbf {\bibinfo {volume} {94}},\ \bibinfo {pages} {062123} (\bibinfo {year} {2016}{\natexlab{b}})}\BibitemShut {NoStop}%
\bibitem [{\citenamefont {Braun}\ \emph {et~al.}(2014)\citenamefont {Braun}, \citenamefont {Giraud}, \citenamefont {Nechita}, \citenamefont {Pellegrini},\ and\ \citenamefont {{\v Z}nidari{\v c}}}]{Braun2014}%
  \BibitemOpen
  \bibfield  {author} {\bibinfo {author} {\bibfnamefont {D.}~\bibnamefont {Braun}}, \bibinfo {author} {\bibfnamefont {O.}~\bibnamefont {Giraud}}, \bibinfo {author} {\bibfnamefont {I.}~\bibnamefont {Nechita}}, \bibinfo {author} {\bibfnamefont {C.}~\bibnamefont {Pellegrini}},\ and\ \bibinfo {author} {\bibfnamefont {M.}~\bibnamefont {{\v Z}nidari{\v c}}},\ }\bibfield  {title} {\bibinfo {title} {A universal set of qubit quantum channels},\ }\href {https://doi.org/10.1088/1751-8113/47/13/135302} {\bibfield  {journal} {\bibinfo  {journal} {Journal of Physics A: Mathematical and Theoretical}\ }\textbf {\bibinfo {volume} {47}},\ \bibinfo {pages} {135302} (\bibinfo {year} {2014})},\ \Eprint {https://arxiv.org/abs/1306.0495} {arXiv:1306.0495 [quant-ph]} \BibitemShut {NoStop}%
\bibitem [{\citenamefont {Mori}\ and\ \citenamefont {{\v{S}}emrl}(2021)}]{MoriSemrl2021}%
  \BibitemOpen
  \bibfield  {author} {\bibinfo {author} {\bibfnamefont {M.}~\bibnamefont {Mori}}\ and\ \bibinfo {author} {\bibfnamefont {P.}~\bibnamefont {{\v{S}}emrl}},\ }\bibfield  {title} {\bibinfo {title} {Continuous coexistency preservers on effect algebras},\ }\href {https://doi.org/10.1088/1751-8121/abcb44} {\bibfield  {journal} {\bibinfo  {journal} {Journal of Physics A: Mathematical and Theoretical}\ }\textbf {\bibinfo {volume} {54}},\ \bibinfo {pages} {015303} (\bibinfo {year} {2021})},\ \Eprint {https://arxiv.org/abs/1911.09490} {arXiv:1911.09490 [math-ph]} \BibitemShut {NoStop}%
\bibitem [{\citenamefont {Kiukas}\ \emph {et~al.}(2017)\citenamefont {Kiukas}, \citenamefont {Budroni}, \citenamefont {Uola},\ and\ \citenamefont {Pellonp\"a\"a}}]{Kiukas2017}%
  \BibitemOpen
  \bibfield  {author} {\bibinfo {author} {\bibfnamefont {J.}~\bibnamefont {Kiukas}}, \bibinfo {author} {\bibfnamefont {C.}~\bibnamefont {Budroni}}, \bibinfo {author} {\bibfnamefont {R.}~\bibnamefont {Uola}},\ and\ \bibinfo {author} {\bibfnamefont {J.-P.}\ \bibnamefont {Pellonp\"a\"a}},\ }\bibfield  {title} {\bibinfo {title} {Continuous-variable steering and incompatibility via state-channel duality},\ }\href {https://doi.org/10.1103/PhysRevA.96.042331} {\bibfield  {journal} {\bibinfo  {journal} {Phys. Rev. A}\ }\textbf {\bibinfo {volume} {96}},\ \bibinfo {pages} {042331} (\bibinfo {year} {2017})}\BibitemShut {NoStop}%
\bibitem [{\citenamefont {Jones}\ \emph {et~al.}(2021)\citenamefont {Jones}, \citenamefont {{\ifmmode \check{S}\else \v{S}\fi{}}upi{\ifmmode \acute{c}\else \'{c}\fi{}}}, \citenamefont {Uola}, \citenamefont {Brunner},\ and\ \citenamefont {Skrzypczyk}}]{Jones-2021a}%
  \BibitemOpen
  \bibfield  {author} {\bibinfo {author} {\bibfnamefont {B.~D.~M.}\ \bibnamefont {Jones}}, \bibinfo {author} {\bibfnamefont {I.}~\bibnamefont {{\ifmmode \check{S}\else \v{S}\fi{}}upi{\ifmmode \acute{c}\else \'{c}\fi{}}}}, \bibinfo {author} {\bibfnamefont {R.}~\bibnamefont {Uola}}, \bibinfo {author} {\bibfnamefont {N.}~\bibnamefont {Brunner}},\ and\ \bibinfo {author} {\bibfnamefont {P.}~\bibnamefont {Skrzypczyk}},\ }\bibfield  {title} {\bibinfo {title} {Network quantum steering},\ }\href {https://doi.org/10.1103/PhysRevLett.127.170405} {\bibfield  {journal} {\bibinfo  {journal} {Phys. Rev. Lett.}\ }\textbf {\bibinfo {volume} {127}},\ \bibinfo {pages} {170405} (\bibinfo {year} {2021})}\BibitemShut {NoStop}%
\bibitem [{\citenamefont {Rudin}(1991)}]{Rudin-1991a}%
  \BibitemOpen
  \bibfield  {author} {\bibinfo {author} {\bibfnamefont {W.}~\bibnamefont {Rudin}},\ }\href {https://books.google.com/books?id=Sh_vAAAAMAAJ} {\emph {\bibinfo {title} {Functional Analysis}}},\ International series in pure and applied mathematics\ (\bibinfo  {publisher} {McGraw-Hill},\ \bibinfo {year} {1991})\BibitemShut {NoStop}%
\bibitem [{\citenamefont {Hug}\ and\ \citenamefont {Weil}(2020)}]{HugWeil2020LecturesConvexGeometry}%
  \BibitemOpen
  \bibfield  {author} {\bibinfo {author} {\bibfnamefont {D.}~\bibnamefont {Hug}}\ and\ \bibinfo {author} {\bibfnamefont {W.}~\bibnamefont {Weil}},\ }\href {https://doi.org/10.1007/978-3-030-50180-8} {\emph {\bibinfo {title} {Lectures on Convex Geometry}}},\ \bibinfo {edition} {1st}\ ed.,\ \bibinfo {series} {Graduate Texts in Mathematics}, Vol.\ \bibinfo {volume} {286}\ (\bibinfo  {publisher} {Springer},\ \bibinfo {address} {Cham},\ \bibinfo {year} {2020})\BibitemShut {NoStop}%
\bibitem [{\citenamefont {Rubin}(2014)}]{Rubin-2014a}%
  \BibitemOpen
  \bibfield  {author} {\bibinfo {author} {\bibfnamefont {B.}~\bibnamefont {Rubin}},\ }\bibfield  {title} {\bibinfo {title} {The $\lambda$-cosine transforms with odd kernel and the hemispherical transform},\ }\href {https://doi.org/10.2478/s13540-014-0198-9} {\bibfield  {journal} {\bibinfo  {journal} {Fractional Calculus and Applied Analysis}\ }\textbf {\bibinfo {volume} {17}},\ \bibinfo {pages} {765–806} (\bibinfo {year} {2014})}\BibitemShut {NoStop}%
\bibitem [{\citenamefont {Bogachev}(2007)}]{Bogachev-2007a}%
  \BibitemOpen
  \bibfield  {author} {\bibinfo {author} {\bibfnamefont {V.~I.}\ \bibnamefont {Bogachev}},\ }\href {https://doi.org/10.1007/978-3-540-34514-5} {\emph {\bibinfo {title} {Measure Theory}}}\ (\bibinfo  {publisher} {Springer Berlin Heidelberg},\ \bibinfo {year} {2007})\BibitemShut {NoStop}%
\end{thebibliography}%

\appendix
\clearpage
\onecolumngrid

\section{Generator Representations of a Zonoid}

\label{App:zonoid-generator}

In this appendix, we provide more details on the generator representation of a zonoid.
Specifically, we derive Eq. \eqref{Eq:zonoid-generator-representation}, Lemma \ref{Lem:POVM-zonoid-generator}, and expressions for the support functions in Eqns. \eqref{Eq:support-function-zonotope1} and \eqref{Eq:support-function-zonoid1}.

\medskip

\begin{applemma}
\label{Lem:Zonoid-representation-Appendix}
Let $\mc{Z}=\pi(\mc{F})$ be the image of some non-atomic vector-valued measure $\pi:\mc{F}\to\mathbb{R}^d$ on a measurable space $(\Lambda,\mc{F})$, i.e. a zonoid.
Then there exists a Borel measure $\nu$ on $S^{d-1}$, called a generating measure, such that
    \begin{align}
\label{Eq:zonoid-generator-representation-appendix}
    \mc{Z}=\left\{\int_{S^{d-1}}q(\hat{r})\hat{r}\;d\nu(\hat{r})\;|\; q\in[0,1] \;\text{is measurable}\right\}=\left\{\int_X \hat{r}d\nu(\hat{r})\:|\; X\in\mf{B}(S^{d-1})\right\}.
\end{align}
\end{applemma}
\begin{proof}

We begin by proving the first equality in Eq. \eqref{Eq:zonoid-generator-representation-appendix}.
Let $\tau:=|\pi| $ denote the total variation of $\pi$ w.r.t. the Euclidean norm on $\mathbb{R}^d$.
By the Radon-Nikodym theorem, there exists a measurable function $\phi:\Lambda\to S^{d-1}$ such that $d\pi=\phi d\tau$ \cite{Rudin-1991a}.
Hence, for every $\vec{w}\in\mc{Z}$, we can find a measurable $E\in\mc{F}$ such that
\begin{align}
    \vec{w}=\pi(E)=\int_\Lambda \chi_E(\lambda) d\pi(\lambda)&=\int_\Lambda \chi_E(\lambda) \phi(\lambda)d\tau(\lambda).
\end{align}
Let $\nu:\mf{B}(S^{d-1})\to[0,\infty)$ denote the pushforward measure of $\tau$, respectively, under the map $\phi$.
Then by the Disintegration Theorem \cite{}, there exists a family of measures $\{\mu_{\hat{r}}\}$ on $\mc{F}$ such that 
\begin{align}
\label{Eq:measure-change-formula}
    \int_\Lambda \chi_E(\lambda) \phi(\lambda)d\tau(\lambda)=\int_\Lambda \chi_E(\lambda) \hat{r}d\tau(\lambda)&=\int_{S^{d-1}} \left[\int_{\phi^{-1}(\hat{r})}\chi_E(\lambda)d\mu_{\hat{r}}(\lambda)\right]\hat{r}d\nu(\hat{r})=\int_{S^{d-1}} q(\hat{r})\hat{r}d\nu(\hat{r}),
\end{align}
where $q(\hat{r})=\int_{\phi^{-1}(\hat{r})}\chi_E(\lambda)d\mu_{\hat{r}}(\lambda)$ is a measurable function bounded in the region $[0,1]$.
Since this holds for every $\vec{w}\in\mc{Z}$, we have
\begin{align}
\label{Eq:zonotope-inclusion}
    \mc{Z}\subseteq \left\{\int_{S^{d-1}}q(\hat{r})\hat{r}\;d\nu(\hat{r})\;|\; q\in[0,1] \;\text{is measurable}\right\}.
\end{align}

To see the converse inclusion, for any $B\subset \mf{B}(S^{d-1})$ let $E_B=\phi^{-1}(B)\subset\mc{F}$.
Then for $A\subset S^{d-1}$ we have
\begin{align}
\nu_{E_B}(A)=\tau_{E_B}(\phi^{-1}(A))
&=\tau(E_B\cap \phi^{-1}(A))\notag\\
&=\tau(\phi^{-1}(B)\cap\phi^{-1}(A))\notag\\
&=\tau(\phi^{-1}(B\cap A))=\nu(B\cap A),\notag
\end{align}
and so $\nu(B\cap A)=\int_{A}q_{E_B}d\nu$, which means that $q_{E_B}=\chi_{B}$ almost everywhere.
Hence, for every $B\subset \mf{B}(S^{d-1})$, we have $\mc{Z}\ni\phi(E_B)=\int_{S^{d-1}}\chi_{B}(\hat{r})\;\hat{r}\;d\nu(\hat{r}) $.
Any measurable function can be approximated arbitrarily well by a convex combination of indicator functions. 
Since $\mc{Z}$ is closed and convex, the reverse inclusion of Eq. \eqref{Eq:zonotope-inclusion} therefore follows.
Note that since $\phi^{-1}(S^{d-1})=\Lambda$, this argument also implies that 
\begin{align}
   \int_{S^{d-1}}\hat{r}\;d\nu(\hat{r})=\pi(\Lambda).
\end{align}

We now turn to the second equality in Eq. \eqref{Eq:zonoid-generator-representation-appendix}.
Consider the vector-valued measure $\pi_0(X)=\int_X\hat{r}d\nu(\hat{r})$, where $\nu$ is the positive measure constructed above, and its associated zonoid $\mc{Z}_0:=\pi_0(\mathfrak{B}(S^{d-1}))$.
Since every indicator function is measurable, we have $\mc{Z}_0\subseteq\mc{Z}$.
To prove the converse inclusion, suppose that $q\in[0,1]$ is an aribtrary non-negative simple function, so that $q=\sum_{i=1}^n a_i\chi_{E_i}$ with the $\{E_i\}_{i=1}^n$ being a disjoint partition of $S^{d-1}$ and $a_1\leq a_2\leq \cdots\leq a_n$.
Observe that we can express this as
\begin{align}
\label{Eq:simple-convex}
    q=a_1\sum_{i=1}^n\chi_{E_i}+(a_2-a_1)\sum_{i=2}^n\chi_{E_i}+\cdots+(a_n-a_{n-1})\chi_{E_n}+(1-a_n)\chi_{\emptyset}.
\end{align}
With the $E_i$ being disjoint, we have $\sum_{i=k}^n\chi_{E_k}=\chi_{B_k}$ with $B_k=\cup_{i=k}^n E_i$.
Thus, Eq. \eqref{Eq:simple-convex} expresses $q$ as a convex combination of indicator functions.
Hence,
\begin{align}
    \int_{S^{d-1}}q(\hat{r})\hat{r}d\nu(\hat{r})=a_1\int_{B_1}\hat{r}d\nu(\hat{r})+(a_2-a_1)\int_{B_2}\hat{r}d\nu(\hat{r})+\cdots+(1-a_n)\int_{\emptyset}\hat{r}d\nu(\hat{r}).
\end{align}
Each integral on the RHS belongs to $\mc{Z}_0$, and since $\mc{Z}_0$ is a zonoid, it is convex.
Therefore, the LHS also belongs to $\mc{Z}_0$.
Since $\mc{Z}_0$ is compact, it is closed when considering limits of simple functions, and hence $\mc{Z}\subseteq\mc{Z}_0$.

\end{proof}

We extend Lemma \ref{Lem:Zonoid-representation-Appendix} to zonoids generated by POVMs.
This is Lemma \ref{Lem:POVM-zonoid-generator} in the main text, which we repeat here. 
\setcounter{lemma}{1}
\begin{lemma}
Let $\mc{K}_{\Pi}=\text{conv}(\Pi(\mc{F}))$ be the zonoid generated by POVM $\Pi:\mc{F}\to\text{Pos}(\mathbb{C}^d)$ on a measurable space $(\Lambda,\mc{F})$.
Then there exists a Borel probability measure $\omega$ on $\text{D}(\mathbb{C}^d)$ such that
    \begin{align}
    \label{Eq:zonoid-generator-representation-appendix-POVM}
    \mc{K}_{\Pi}=\left\{d\int_{\text{D}(\mathbb{C}^d)}q(\rho)\rho\;d\omega(\rho)\;|\; q\in[0,1] \;\text{is measurable}\right\}=\left\{d\int_{X}\rho\;d\omega(\rho)\;|\; X\in\mf{B}(\tD(\mathbb{C}^d))\right\}.
\end{align}
  Moreover, the POVM $X\mapsto d\int_X\rho d\omega(\rho)$ is operationally equivalent to $\Pi$.

\end{lemma}
\begin{proof}
For a POVM $\Pi:\mc{F}\to \text{Pos}(\mathbb{C}^d)$, we follow analogously to Proposition~3.1 in Ref.~\cite{Aubrun-2015a}.
Consider the total variation $\tau:=|\Pi|$ of $\Pi$ w.r.t. the trace norm, which in this case simply acts as the trace measure $\tau(E)=\tr[\Pi(E)]$.
Like before, there exists a measurable function $\phi:\Lambda\to \text{Herm}(\mathbb{C}^d)$ such that $d\Pi=\phi d\tau$ and $\Vert\phi(\lambda)\Vert_{\tr}=1$ for all $\lambda$.
In fact, $\phi(\lambda)\geq 0$ almost everywhere because  $\Pi(\mc{F})\subset\text{Pos}(\mathbb{C}^d)$; i.e. $\phi(\lambda)\in\text{D}(\mbb{C}^d)$ is a density matrix almost everywhere.
Since $\tau(\Lambda)=\tr[\mathbb{I}]=d$, we push forward the normalized measure $\frac{1}{d}\tau$ by $\phi$ and again apply the Disintegration Theorem.
This leads to 
\begin{align}
    \mc{K}_{\Pi}=\left\{d\int_{\mathcal{D}(\mbb{C}^{d})}q(\rho)\rho\;d\omega(\rho)\;|\; q\in[0,1] \;\text{is measurable}\right\},
\end{align}
where $\rho$ is a trace-one positive operator and $\omega$ is a normalized measure.
The second equality in Eq. \eqref{Eq:zonoid-generator-representation-appendix-POVM} follows by the exact same argument used in Lemma \ref{Lem:Zonoid-representation-Appendix}.

To show that the POVM $X\mapsto d\int_X\rho d\omega(\rho)$ is operationally equivalent to $\Pi$, consider any assemblage $\{M_{a|x}\}$ that can be simualted by $\Pi$.
Then by definition there exists a post-processing map $p(a|x,\lambda)$ such that \[M_{a|x}=\int_\Lambda p(a|x,\lambda) d\Pi(\lambda)=d\int_\Lambda p(a|x,\lambda)\rho(\lambda)d\tau(\lambda).\]
As in Lemma \ref{Lem:Zonoid-representation-Appendix}, when we transform to the pushforward measure, $\tau\mapsto \omega$, the equality in Eq. \eqref{Eq:measure-change-formula} becomes $M_{a|x}=d\int_{\tD(\mathbb{C}^d)}q(a|x,\rho)\rho d\omega(\rho)$, where $q(a|x,\rho):=\int_{\rho^{-1}(\lambda)}p(a|x,\lambda)\mu_\rho(\lambda)$.
Since each $\mu_\rho$ is a probability measure for almost all $\rho$, we have that $\sum_{a}q(a|x,\rho)=1$, as desired.
For the negligible $\rho$ with $\mu_\rho$ not being a probability measure, we can simply replace $q(a|x,\rho)$ with any valid probability distribution, and the simulation remains unchanged.
\end{proof}

\begin{appcorollary}
\label{Cor:Zonoid-representation-Appendix-qubit}
    For any qubit POVM $\Pi:\mc{F}\to\text{Pos}(\mathbb{C}^2)$, there exists a Borel probability measure $\nu$ on $B^3$ such that
   \begin{align}
\label{Eq:Bloch-sphere-POVM-generating}
    \mc{K}_{\Pi}=\left\{\int_{B^3}q(\vec{s})(\mathbb{I}+\vec{s}\cdot\vec{\sigma})\;d\nu(\vec{s})\;|\; q\in[0,1] \;\text{is measurable}\right\}=\left\{\int_{X}(\mathbb{I}+\vec{s}\cdot\vec{\sigma})\;d\nu(\vec{s})\;|\; X\in\mf{B}(\tD(B^3))\right\}.
\end{align}
   Moreover, the POVM $X\mapsto\int_X(\mathbb{I}+\vec{s}\cdot\vec{\sigma})d\omega(\vec{s})$ is operationally equivalent to $\Pi$.    
\end{appcorollary}
\begin{proof}
Start from Eq. \eqref{Eq:zonoid-generator-representation-appendix-POVM} and parametrize $\text{D}(\mbb{C}^2)$ as
\begin{align}
    \rho=\rho(\vec{s})=\frac{1}{2}(\mathbb{I}+\vec{s}\cdot\vec{\sigma}),
\end{align}
with $\vec{s}\in B^3$ being the Bloch vector of $\rho$.
Let $Bl:\text{Herm}(\mathbb{C}^2)\to \mathbb{R}^3$ denote the map that projects a hermitian operator $M$ onto its Bloch vector.
By taking $\nu$ to be the pushforward of $\omega$ by $Bl$ and noting that $\nu$ will be fully supported on $B^3$ (due to the positivity of $\rho$),  we directly have Eq. \eqref{Eq:Bloch-sphere-POVM-generating}.
The same reasoning applies when changing measure in the equality $M_{a|x}=\int_{\tD(\mathbb{C}^d)}q(a|x,\rho)\rho d\omega(\rho)$,
thereby establishing the operational equivalence of the POVM $X\mapsto\int_X(\mathbb{I}+\vec{s}\cdot\vec{\sigma})d\omega(\vec{s})$.
\end{proof}

To close out this section of the Appendix, we turn to the support functions.
For a zonoid $\mc{Z}$, we use Eq. \eqref{Eq:zonoid-generator-representation-appendix}
\begin{align}
    h_{\mc{Z}}(\vec{n})=\max_{q\in[0,1]}\int_{S^{d-1}}q(\hat{r})\langle\vec{n},\hat{r}\rangle d\nu(\hat{r})&=\int_{S^{d-1}}\max\{0,\langle\vec{n},\hat{r}\rangle\} d\nu(\hat{r})\notag\\
    &=\frac{1}{2}\int_{S^{d-1}}\langle\vec{n},\hat{r}\rangle d\nu(\hat{r})+\frac{1}{2}\int_{S^{d-1}}|\langle\vec{n},\hat{r}\rangle| d\nu(\hat{r})\notag\\
    &=\frac{1}{2}\langle\vec{n},\vec{c}\rangle+\frac{1}{2}\int_{S^{d-1}}|\langle\vec{n},\hat{r}\rangle| d\nu(\hat{r}),
\end{align}
where we have used the equality $\max\{0,x\}=\frac{1}{2}(x+|x|)$.
Similarly, for a zonotope $\mc{Z}$ with generators $\{\vec{v}_i\}_{i=1}^k$, we have that
\begin{align}
    h_{\mc{Z}}(\vec{n})&=\max_{0\leq t_1,\cdots,t_k\leq 1}\sum_{i=1}^kt_i\langle\vec{n},\mbf{v}_i\rangle=\frac{1}{2}\sum_{i=1}^k\left(\langle\mbf{n},\mbf{v}_i\rangle+|\langle\mbf{n},\mbf{v}_i\rangle|\right)=\frac{1}{2}\langle\vec{n},\vec{c}\rangle+\frac{1}{2}\sum_{i=1}^k|\langle\mbf{n},\mbf{v}_i\rangle|.
\end{align}

\section{Proof of the Boundary Zonoid Uniqueness Theorem}
\label{App:boundary-zonoid-uniqueness}

In this section we provide a proof of Theorem~\ref{Thm:zonoid-unique}, restated below.
\setcounter{theorem}{1}
\begin{theorem}
Let $\mathcal{Z}\subset\mathbb{R}^3$ be a centrally-symmetric zonoid with $\ell(\mc{Z})=2$.  
Then there exists a unique probability measure $\nu$ on $S^2$ such that the POVM $\Pi_{\mc{Z}}:\mf{B}(S^2)\to\text{Pos}(\mathbb{C}^2)$ defined by $\Pi_{\mc{Z}}(X)=\int_X(\mathbb{I}+\hat{s}\cdot\vec{\sigma})d\nu(\hat{s})$  simulates the assemblage $\{M_{\vec{m}}\}_{\vec{m}\in\mc{Z}}$.
Moreover, any other POVM that can simulate the assemblage is operationally equivalent to $\Pi_{\mc{Z}}$.

\end{theorem}

\begin{proof}

The existence of a parent POVM for the assemblage  $\{M_{\vec{m}}\}_{\vec{m}\in\mc{Z}}$ follows from Proposition \ref{prop:qubits-zonoid-necessary}.
We now prove uniqueness. 

Suppose that $\Pi:\mathcal F\to\mathsf{Pos}(\mathbb C^2)$ is a POVM that can simulate the assemblage $\{M_{\vec{m}}\}_{\vec{m}\in\mc{Z}}$.
By Corollary \ref{Cor:Zonoid-representation-Appendix-qubit}, there exists a Borel probability measure $\nu$ on $B^3$ such that the POVM $X\mapsto\int_X(\mathbb{I}+\vec{s}\cdot\vec{\sigma})d\nu(\vec{s})$ is operationally equivalent to $\Pi$ and
\begin{align}
\label{Eq:measure-B3-equality}
\mc{K}_{\Pi}=\left\{\int_{B^3}q(\vec{s})(\mathbb{I}+\vec{s}\cdot\vec{\sigma})\;d\nu(\vec{s})\;|\; q\in[0,1] \;\text{is measurable}\right\}.
\end{align}
We first show that the assumption $\ell(\mc{Z})=2$ implies that $|\vec{s}|=1$ almost everywhere, which means that $\nu$ is supported on $S^2$.
To see this, we follow along the lines of Eq. \eqref{Eq:qubit-support-function-bound} and use Eq.~\eqref{Eq:measure-B3-equality} to compute
\begin{align}
h_{\mc Z}(\hat n)
&\leq\sup_{M\in \mc{K}_\Pi}\hat n\cdot Bl[M]
=\sup_{q\in[0,1]}2\int_{X}q(\vec{s})\hat n\cdot\vec s\,d\nu(\vec s)
=2\int_{H_{\hat n}}\hat n\cdot\vec s\,d\nu(\vec s),
\end{align}
where $H_{\hat n}:=\{\vec s\in B^3:\vec s\cdot\hat n>0\}$. 
Since $\int_{B^3}\vec s d\nu(\vec s)=0$,
\begin{align}
\int_{B^3}|\hat n\cdot\vec s|d\nu(\vec s)=2\int_{H_{\hat n}}\hat n\cdot\vec s d\nu(\vec s),
\end{align}
and therefore
\begin{align}
\label{Eq:support-bound}
h_{\mc Z}(\hat n)\leq\int_{B^3}|\hat n\cdot\vec s|d\nu(\vec s)\qquad\forall\hat n.
\end{align}
Integrating uniformly over the unit sphere and using $\ell(\mc Z)=2$ gives
\begin{align}
\frac12=\int_{S^2}h_{\mc Z}(\hat n)d\mu(\hat n)
\leq\int_{B^3}d\nu(\vec s)\int_{S^2}|\hat n\cdot\vec s|d\mu(\hat n)
=\frac12\int_{B^3}\Vert\vec s\Vert d\nu(\vec s)\leq\frac12.
\end{align}
Hence, equality holds throughout. 
This implies that $\Vert\vec s\Vert=1$ almost everywhere and thus $\nu$ is supported on $S^2$. 
Returning to the above analysis, the difference between the two sides of Eq.~\eqref{Eq:support-bound} is a nonnegative continuous function of $\hat n$, and since $\int_{B^3}|\hat{n}\cdot\vec{s}|d\nu(\vec{s})=\int_{S^2}|\hat{n}\cdot\hat{s}|d\nu(\hat{s})$, we must have that
\begin{align}
\label{Eq:support-Eq}
h_{\mc Z}(\hat n)=\int_{S^2}|\hat n\cdot\hat s|d\nu(\hat s)\qquad\forall\hat n.
\end{align}
Since $S^2$ corresponds to the collection of pure density matrices, we can summarize our findings so far as follows.
\begin{applemma}
\label{Lem:Support-S2}
    Suppose $\Pi:\mathfrak{B}(\tD(\mathbb{C}^2))\to\text{Pos}(\mathbb{C}^2)$ is a POVM of the form $\Pi:X\mapsto \int_X\rho d\nu(\rho )$ that can simulate an assemblage $\{M_{\vec{m}}\}_{\vec{m}\in\mc{Z}}$ with $\ell(\mc{Z})=2$.
    Then $\Pi$ is supported on $\tD^1(\mathbb{C}^2)$, the collection of pure qubit density operators.
\end{applemma}

Now, it is well known that for every centrally-symmetric zonoid $\mc Z\subset{R}^3$, there exists a \textit{unique} even Borel measure $\nu_{\mc Z}$ on $S^2$ such that $h_{\mc Z}(\hat n)=\int_{S^2}|\hat n\cdot\hat s|d\nu_{\mc Z}(\hat s)$~\cite[Thm.~4.14]{HugWeil2020LecturesConvexGeometry} (see also Eq. \eqref{Eq:support-representation}). 
Therefore, if we can prove that our constructed measure $\nu$ is even, uniqueness would imply that $\nu=\nu_{\mc Z}$. 
Consequently, the POVM $X\mapsto \int_X(\mathbb{I}+\hat{s}\cdot\vec{\sigma})d\nu$ that is operationally equivalent to $\Pi$ would also be uniquely determined by $\mc{Z}$.


To prove that $\nu$ is even, let $H_{\hat{n}}:=\{\hat s\in S^2:\hat n\cdot s>0\}$ denote the $\hat{n}$-hemisphere on $S^2$ and $G_{\hat{n}}:=\{\hat s\in S^2:\hat n\cdot s=0\}$ its equator.
For any $\hat{n}\in S^2$, let $\vec m(\hat n)\in\mc Z$ be a vector such that $h_{\mc Z}(\hat n)=\hat n\cdot\vec m(\hat n)$. 
Since $M_{\vec m(\hat n)}\in \mc{K}_\Pi$, there exists a measurable function $q_{\hat n}:S^2\to[0,1]$ such that
\begin{align}
\frac12(\mbb I+\vec m(\hat n)\cdot\vec\sigma)=\int_{S^2}q_{\hat n}(\hat s)(\mbb I+\hat s\cdot\vec\sigma)d\nu(\hat s).
\label{Eq:boundary-response}
\end{align}
Comparing the identity and Bloch-vector components gives $\int_{S^2}q_{\hat n}(\hat s)d\nu(\hat s)=1/2$ and $\vec m(\hat n)=2\int_{S^2}q_{\hat n}(\hat s)\hat s d\nu(\hat s)$. 
Therefore, using $\int_{S^2}\hat s d\nu(\hat s)=0$, the assumption that $h_{\mc{Z}}(\hat{n})=\hat{n}\cdot\vec{m}(\hat{n})$, and Eq.~\eqref{Eq:support-Eq}, we have
\begin{align}
\label{Eq:support-function-almost-all}
h_{\mc Z}(\hat n)
&=\int_{S^2}[2q_{\hat n}(\hat s)-1]\hat n\cdot\hat s d\nu(\hat s)\leq\int_{S^2}|\hat n\cdot\hat s|d\nu(\hat s)=h_{\mc Z}(\hat n).
\end{align}
Equality must therefore hold almost everywhere. 
Since $|2q_{\hat n}-1|\leq1$, this implies $q_{\hat n}(\hat s)=1$ whenever $\hat n\cdot\hat s>0$ and $q_{\hat n}(\hat s)=0$ whenever $\hat n\cdot\hat s<0$, up to sets of $\nu$-measure zero.
Our findings here thus establish the following lemma.
\begin{applemma}
\label{Lem:measure-0-equator}
If $\nu(G_{\hat{n}})=0$, then $q_{\hat{n}}$ is equal to the indicator function $\chi_{H_{\hat{n}}}$ $\nu$-almost everywhere.
Consequently, $\nu(H_{\hat{n}})=\frac{1}{2}$ and there exists a unique $\vec{m}=\vec{m}(\hat{n})\in\mc{Z}$ such that $h_{\mc{Z}}(\hat{n})=\hat{n}\cdot\vec{m}$.
\end{applemma}
Notice that $\nu(G_{\hat{n}})=0$ for $\mu$-almost all $\hat{n}\in S^2$, where $\mu$ denotes the uniform measure.
Indeed, by Fubini's Theorem,
\begin{align}
\int_{S^2}\nu(G_{\hat n})d\mu(\hat n)
=\int_{S^2}\nu(d\hat s)\int_{S^2}\mathbf 1_{\{\hat n\cdot\hat s=0\}}d\mu(\hat n)=0,
\end{align}
which says that $\nu(G_{\hat{n}})$ vanishes $\mu$-almost everywhere.
We can then complete the proof of Theorem \ref{Thm:zonoid-unique} by arguing that any measure $\nu$ satisfying $\nu(H_{\hat{n}})=1/2$ for $\mu$-almost all $\hat{n}\in S^2$ must be even.
This argument is based on known results in the study of hemispherical transforms \cite{Rubin-2014a}.
Letting $\mc{M}(S^2)$ denote the collection of finite signed Borel measures on $S^2$, a hemispherical transform maps any measure  $\nu\in\mc{M}(S^2)$ to the function $\mc{H}\nu$ on $S^2$ given by $\mc{H}\nu(\hat{n})=\nu(H_{\hat{n}})$.
A key fact for our purposes is that the hemispherical transform is injective when restricted to the family of odd measures on $S^2$ (Corollary 5.1. in \cite{Rubin-2014a}).  
To apply this on a general measure $\nu$, decompose it as $\nu=\nu_{\text{even}}+\nu_{\text{odd}}$,
where $\nu_{\text{odd}}(X):=(\nu(X)-\nu(-X))/2$ is an odd signed measure.
In our case, we have that $\nu(H_{\hat{n}})=\nu(H_{-\hat{n}})$, which means that $\nu_{\text{odd}}(H_{\hat{n}})=0$ for $\mu$-almost all $\hat{n}$.
Hence $\mc{H}\nu_{\text{odd}}$ is equivalent to the zero function when integrated against any other function, which is still sufficient for the injectivity of $\mc{H}$ to hold (see Theorem 4.1 in \cite{Rubin-2014a}). 
Hence, $\nu=\nu_{\text{even}}$.


\end{proof}

\section{Uniqueness of the Two-Copy Boundary Parent}
\label{App:two-copy-boundary-parent}

In this section we prove Theorem \ref{Thm:two-copy-unique}, restated here for convenience.

\setcounter{theorem}{3}
\begin{theorem}
    Consider arbitrary noisy spin assemblages $\mathfrak{M}_{\A}=\{M_{\vec{m}}\}_{\vec{m}\in\mc{Z}_\A}$ and $\mathfrak{M}_{\B}=\{M_{\vec{m}}\}_{\vec{k}\in\mc{Z}_\B}$ such that $\ell(\mc{Z}_{\A})=\ell(\mc{Z}_{\B})=2$.
    Let $\Pi_{\A}$ and $\Pi_{\B}$ be POVMs that can simulate $\mathfrak{M}_{\A}$ and $\mathfrak{M}_{\B}$, respectively. 
    Then any bipartite POVM $\Pi_{\A\B}$ that can simulate the product assemblage $\mathfrak{M}_{\A}\otimes \mathfrak{M}_{\B}$ must be operationally equivalent to the product POVM $\Pi_{\A}\otimes\Pi_{\B}$.
\end{theorem}

\begin{proof}
Consider the zonoid $\mc{K}_{\Pi}=\text{conv}(\Pi(\mc{F}))$ of any parent POVM $\Pi_{\A\B}$ that can simulate the product assemblage $\mathfrak{M}_{\A}\otimes \mathfrak{M}_{\B}$.
By Lemma \ref{Lem:POVM-zonoid-generator},  there exists a Borel probability measure $\omega_{\A\B}$ on $\text{D}_{\A\B}:=\text{D}(\mathbb{C}^2\otimes\mathbb{C}^2)$ such that the POVM $X\mapsto 4\int_X\rho_{\A\B}d\omega_{\rho_{\A\B}}$ is operationally equivalent to $\Pi$.
Let $\omega_{\A}$ and $\omega_{\B}$ denote the pushforward measures of $\omega_{\A\B}$ by the partial trace maps $\tr_{\B}$ and $\tr_{\A}$, respectively.
Since $\{M_{\vec{m}}\}_{\vec{m}\in\mc{Z}_{\A}}\subset\tr_\B[\mc{K}_{\Pi}]$, by Lemma \ref{Lem:Support-S2} it follows that $\omega_{\A}$ has full support on $\text{D}_{\A}^1:=\text{D}^1(\mathbb{C}^2)$, the set of pure qubit density matrices.
By definition of the pushforward measure, we have
\begin{align}
    1&=\omega_{\A}(\text{D}^1_{\A})=\omega_{\A\B}(\tr_{\B}^{-1}(\text{D}^1_{\A}))=\omega_{\A\B}(\{\rho_{\A\B}\;|\;\tr_{\B}(\rho_{\A\B})\in \text{D}^1_{\A}\})=\omega_{\A\B}(\text{D}_{\A}^1\otimes \text{D}_{\B}), 
\end{align}
where the last equality follows from the fact that $\tr_{B}(\rho_{\A\B})$ is pure iff $\rho_{\A\B}=\psi_{\A}\otimes\rho_{\B}$ for some pure state $\psi=\op{\psi}{\psi}$.
A similar argument can be made for $\omega_{\B}$, and so we conclude that $\omega_{\A\B}$ is supported on the set of product pure states, i.e. $\omega_{\A\B}(\tD^{1}_{\A}\otimes \tD^1_{\B})=1$.
Letting $\nu_{\A\B}$ denote the pushforward of $\omega_{\A\B}$ under the Bloch vector mapping $\rho_{\A\B}\mapsto (Bl(\rho_\A), Bl(\rho_{\B}))$, we have that $\Pi$ is operationally equivalent to the POVM
\begin{equation}
    X\mapsto \int_X (\mathbb{I}+\hat{s}\cdot\vec{\sigma})\otimes(\mathbb{I}+\hat{t}\cdot\vec{\sigma})d\nu_{\A\B}(\hat{s},\hat{t})\quad\text{for}\quad X\in\mf{B}(S^2\times S^2).
\end{equation}
To prove the theorem, we will show that $\nu_{\A\B}$ is the product measure
\begin{align}
    \nu_{\A\B}=\nu_{\A}\times\nu_{\B},
\end{align}
where $\nu_{\A}$ is the unique generating measure for $\mc{Z}_{\A}$ whose existence is established in Theorem \ref{Thm:zonoid-unique}, and likewise $\nu_{\B}$.

The fact that $\Pi$ can simulate the four-outcome POVM $\{M_{\pm\vec{m}}\otimes M_{\pm\vec{m}'}\}$ means we can write 
 \begin{align}
 \label{Eq:Simulate-EQ}
    M_{\pm\vec{m}}\otimes M_{\pm\vec{m}'}&=\int_{S^2\times S^2} q(\pm\vec{m},\pm\vec{m}'|\hat{s},\hat{t})(\mathbb{I}+\hat{s}\cdot\vec{\sigma})\otimes (\mathbb{I}+\hat{t}\cdot\vec{\sigma})d\nu(\hat{s},\hat{t})
\end{align}  
for some response function $q(\pm\vec{m},\pm\vec{m}'|\hat{s},\hat{t})\geq 0$ satisfying
\begin{align}
    1=q(\vec{m},\vec{m}'|\hat{s},\hat{t})+q(\vec{m},-\vec{m}'|\hat{s},\hat{t})+q(-\vec{m},\vec{m}'|\hat{s},\hat{t})+q(-\vec{m},-\vec{m}'|\hat{s},\hat{t}).
\end{align}
We will first show that $q(\pm\vec{m},\pm\vec{m}'|\hat{s},\hat{t})$ is essentially equivalent to the product of indicator functions almost everywhere.
Recall from Lemma \ref{Lem:measure-0-equator} that when $\nu_{\A}(\hat{G}_{\hat{n}})=\nu_{\B}(\hat{G}_{\hat{n}'})=0$, there exists unique vectors $\vec{m}(\hat{n})\in\mc{Z}_{\A}$ and $\vec{m}'(\hat{n}')\in\mc{Z}_{\B}$ such that $h_{\mc{Z}_{\A}}(\hat{n})=\hat{n}\cdot\vec{m}(\hat{n})$ and $h_{\mc{Z}_{\B}}(\hat{n}')=\hat{n}'\cdot\vec{m}'(\hat{n}')$.
These vectors also satisfy $\vec{m}(-\hat{n})=-\vec{m}(\hat{n})$ and $\vec{m}'(-\hat{n}')=-\vec{m}'(\hat{n}')$.

\begin{applemma}
\label{Lem:response-factorize}
    If $\hat{n},\hat{n}'\in S^2$ satisfy $\nu_{\A}(G_{\hat{n}})=\nu_{\B}(G_{\hat{n}'})=0$, then
    $q(\pm\vec{m}(\hat{n}),\pm\vec{m}'(\hat{n}')|\hat{s},\hat{t})=\chi_{H_{\pm\hat{n}}}(\hat{s})\chi_{H_{\pm\hat{n}'}}(\hat{t})$ $\nu$-almost everywhere.
\end{applemma}
\begin{proof}
By the Disintegration Theorem (Chapter 10 in \cite{Bogachev-2007a}), there exists a family of conditional measures $\{\nu_{\hat{s}}\}$ on $S^2$ such that we can express Eq. \eqref{Eq:Simulate-EQ} as
\begin{subequations}
\begin{align}
     M_{\vec{m}(\hat{n})}\otimes M_{\vec{m}'(\hat{n})}&=\int_{S^2} (\mathbb{I}+\hat{s}\cdot\vec{\sigma})d\nu(\hat{s},S^2) \otimes \int_{S^2}q(\vec{m},\vec{m}'|\hat{s},\hat{t})(\mathbb{I}+\hat{t}\cdot\vec{\sigma})d\nu_{\hat{s}}(\hat{t})\\\
     M_{\vec{m}(\hat{n})}\otimes M_{-\vec{m}'(\hat{n})}&=\int_{S^2} (\mathbb{I}+\hat{s}\cdot\vec{\sigma})d\nu(\hat{s},S^2) \otimes \int_{S^2}q(\vec{m},-\vec{m}'|\hat{s},\hat{t})(\mathbb{I}+\hat{t}\cdot\vec{\sigma})d\nu_{\hat{s}}(\hat{t})
\end{align}
\end{subequations}
where $\nu(\cdot,S^2)$ denotes the pushforward measure of $\nu$ under the projection $S^2\times S^2\to S^2$ of the first coordinate.
Summing these gives
\begin{align}
\label{Eq:Simulate-EQ-disintegrate}
&M_{\vec m(\hat{n})}\otimes\mathbb{I}=\int_{S^2}(\mathbb{I}+\hat{s}\cdot\vec{\sigma})\otimes\left[\int_{S^2}\left(q(\vec m,\vec m'|\hat{s},\hat{t})+q(\vec m,-\vec m'|\hat{s},\hat{t})\right)(\mathbb{I}+\hat{t}\cdot\vec{\sigma})d\nu_{\hat{s}}(\hat{t})\right]d\nu(\hat{s},S^2),
\end{align}
which after tracing over the second system yields
\begin{align}
&M_{\vec m(\hat{n})}=\int_{S^2}(\mathbb{I}+\hat{s}\cdot\vec{\sigma})f_{\vec m}(\hat{s})d\nu(\hat{s},S^2),\qquad f_{\vec m}(\hat{s}):=\int_{S^2}\left(q(\vec m,\vec m'|\hat{s},\hat{t})+q(\vec m,-\vec m'|\hat{s},\hat{t})\right)d\nu_{\hat{s}}(\hat{t}).
\end{align}
The response function $f_{\vec{m}}$ lies in the interval $[0,1]$ since $0\leq q(\vec m,\vec m'|\hat{s},\hat{t})+q(\vec m,-\vec m'|\hat{s},\hat{t})\le 1$.
Because this holds for all $\vec{m}\in\mc{Z}_{\A}$, Theorem \ref{Thm:zonoid-unique} requires that $\nu(\cdot,S^2)=\nu_{\A}$. 
Moreover, Lemma~\ref{Lem:measure-0-equator} gives $f_{\vec{m}}(\hat{n})=\chi_{H_{\hat{n}}}(\hat{s})$ $\nu_{\A}$-almost everywhere. 
This further implies that \[q(\vec m,\vec m'|\hat{s},\hat{t})+q(\vec m,-\vec m'|\hat{s},\hat{t})
\underset{a.e}{=}\chi_{H_{\hat n}}(\hat{s}),\]  
where $\underset{a.e}{=}$ denotes $\nu_{\A}$-almost everywhere. 
Similar arguments show that 
\begin{align}
q(-\vec m,\vec m'|\hat{s},\hat{t})+q(-\vec m,-\vec m'|\hat{s},\hat{t})&\underset{a.e}{=}\chi_{H_{-\hat n}}(\hat{s})\notag\\
q(\vec m,\vec m'|\hat{s},\hat{t})+q(-\vec m,\vec m'|\hat{s},\hat{t})&\underset{a.e}{=}\chi_{H_{\hat n'}}(\hat{t})\notag\\
q(\vec m,-\vec m'|\hat{s},\hat{t})+q(\vec m,-\vec m'|\hat{s},\hat{t})&\underset{a.e}{=}\chi_{H_{-\hat n'}}(\hat{t}).
\end{align}
By considering the four sectors $H_{\pm\hat{n}}\times H_{\pm\hat{n}'}$, these equalities collectively imply that $q(\hat{s},\hat{t}|\vec{m},\vec{m}')$ equals $\chi_{H_{\hat{n}}}(\hat{s})\chi_{H_{\hat{n}'}}(\hat{t})$ $\nu$-almost everywhere.

\end{proof}

For any $\hat{n},\hat{n}'\in S^2$ with $\nu_{\A}(G_{\hat{n}})=\nu_{\B}(G_{\hat{n}'})=0$, Lemma \ref{Lem:response-factorize} allows us to write 
\begin{subequations}
\begin{align}
    M_{\vec{m}(\hat{n})}\otimes M_{\vec{m}'(\hat{n}')}&=\int_{H_{\hat{n}}\times H_{\hat{n}'}} (\mathbb{I}+\hat{s}\cdot\vec{\sigma}) \otimes (\mathbb{I}+\hat{t}\cdot\vec{\sigma})d\nu(\hat{s},\hat{t})\notag\\
    &=\int_{H_{\hat{n}}\times H_{\hat{n}'}}(\mathbb{I}+\hat{s}\cdot\vec{\sigma}) \otimes (\mathbb{I}+\hat{t}\cdot\vec{\sigma})d(\nu_{\A}\times\nu_{\B})(\hat{s},\hat{t}).
\end{align} 
\end{subequations}
Taking the difference gives
\begin{align}
    0&=\int_{H_{\hat{n}}\times H_{\hat{n}'}} (\mathbb{I}+\hat{s}\cdot\vec{\sigma}) \otimes (\mathbb{I}+\hat{t}\cdot\vec{\sigma})d\Delta(\hat{s},\hat{t}),
\end{align}
where $\Delta=\nu_{\A\B}-\nu_\A\times\nu_\B$ is a signed measure.
For all $\lambda\in\{0,1,2,3\}$ such that $\tr[M_{\vec{m}'}\sigma_\lambda]\not=0$ contract the second system by $\frac{1}{2}\sigma_\lambda$ to obtain
\begin{align}
    0=\int_{H_{\hat{n}}}(\mathbb{I}+\hat{s}\cdot\vec{\sigma})d\tau_{\vec{m}',\lambda}(\hat{s}),\qquad\text{where}\qquad d\tau_{\vec{m}',\lambda}(\hat{s}):=\int_{H_{\hat{n}'}} t_\lambda d\Delta(\hat{s},\hat{t})\quad\text{for}\quad t_0=1,\;\;(t_1,t_2,t_3)=\hat{t}.
\end{align} 
The integrand on the LHS says that (i) $0=\tau_{\vec{m}',\lambda}(H_{\hat{n}})$ and (ii) $0=\int_{H_{\hat{n}}}\hat{s}d\tau_{\vec{m}',\lambda}(\hat{s})$.
Since (i) holds for $\mu$-almost all $\hat{n}$, following the proof of Theorem \ref{Thm:zonoid-unique}, injectivity of the hemispherical transform implies that $\tau_{\vec{m}',\lambda}$ is an even measure.  
Turning to (ii), we have that the vector-valued measure $X\mapsto \int_X \hat{s}d\tau_{\vec{m}',\lambda}(\hat{s})$ vanishes on $\mu$-almost all hemispheres $X=H_{\hat{n}}$.
However, these is an odd measure since $\tau_{\vec{m}',\lambda}$ is even, and it therefore must each be identically zero by injectivity of the hemispherical transform for odd signed measures.
Then for all $X\in\mf{B}(S^2)$, we have
\[0=\int_X \hat{s}d\tau_{\vec{m}',\lambda}(\hat{s})=\int_X \hat{s}g(\hat{s})d|\tau_{\vec{m}',\lambda}|(\hat{s})\]
where we have used the Radon-Nikodym theorem to write $d\tau_{\vec{m}',\lambda}=gd|\tau_{\vec{m}',\lambda}|$ with $|g|=1$ and $|\tau_{\vec{m}',\lambda}|$ being the total variation of $\tau_{\vec{m}',\lambda}$ \cite{Rudin-1991a}.
Since $|\tau_{\vec{m}',\lambda}|$ is a non-negative measure and the previous expression holds for all $X$, it follows that $\hat{s}g(\hat{s})$ vanishes $|\tau_{\vec{m}',\lambda}|$-almost everywhere. 
But $\Vert\hat{s} g(\hat{s})\Vert=1$, which means the latter is possible only if $|\tau_{\vec{m}',\lambda}|$ (and hence also $\tau_{\vec{m},\lambda}$) is the identically zero measure.

For any $X\in\mf{B}(S^2)$ we can then compute
\begin{align}
\label{Eq:Product-partial}
    \int_{X\times H_{\hat{n}'}}(\mathbb{I}+\hat{s}\cdot\vec{\sigma})\otimes (\mathbb{I}+\hat{t}\cdot\vec{\sigma})d\Delta(\hat{s},\hat{t})&=\int_{X\times H_{\hat{n}'}}(\mathbb{I}+\hat{s}\cdot\vec{\sigma})\otimes\sum_\lambda t_\lambda \sigma_\lambda d\Delta(\hat{s},\hat{t})\notag\\
    &=\sum_{\lambda=0}^3\int_{X\times H_{\hat{n}'}}(\mathbb{I}+\hat{s}\cdot\vec{\sigma})t_\lambda d\Delta(\hat{s},\hat{t})\otimes\sigma_\lambda\notag\\
    &=\sum_{\lambda=0}^3\int_X(\mathbb{I}+\hat{s}\cdot\vec{\sigma})d\tau_{\hat{m}',\lambda}(\hat{s})\otimes\sigma_\lambda=0.
\end{align}
We now repeat an analogous argument by contracting the first system in the previous equation by $\frac{1}{2}\sigma_\lambda$.
This yields the expression
\begin{align}
    0=\int_{H_{\hat{n}'}}(\mathbb{I}+\hat{t}\cdot\vec{\sigma})d\omega_{X,\lambda}(\hat{t}),\qquad\text{where}\qquad d\omega_{X,\lambda}(\hat{t}):=\int_{X} s_\lambda d\Delta(\hat{s},\hat{t}).
\end{align} 
Since this holds for $\mu$-almost all vectors $\hat{n}'$, we have that $\omega_{X,\lambda}$ is the identically zero measure.
Therefore, for any $X,Y\in\mf{B}(S^2)$, we have
\begin{align}
    \Delta(X\times Y)=\int_{X\times Y}d\Delta(\hat{s},\hat{t})&=\int_Yd\omega_{X,0}(\hat{t})=0.
\end{align}
Hence, $\nu_{\A\B}=\nu_\A\times\nu_\B$ are equal measures on all Borel measurable product sets, and therefore also on all measurable sets by completion.
This completes the proof.

\end{proof}

\end{document}